\documentclass[a4paper, 11pt]{article}

\usepackage{amsmath, amsfonts, amsthm, graphicx, enumitem, colonequals, fullpage, hyperref}

\newcommand{\ie}{i.e.}
\newcommand{\eg}{e.g.}

\newcommand{\realR}{\mathbb{R}}
\newcommand{\compC}{\mathbb{C}}

\newcommand{\strip}{\mathbb{S}}
\newcommand{\cyld}{\mathbb{S}^c}
\newcommand{\gfn}{\mathbf{g}}
\newcommand{\Jlike}{\mathbf{J}}
\newcommand{\Clike}{\tilde{C}}
\newcommand{\Jinv}{\mathbf{I}}

\newcommand{\bigO}{\mathcal{O}}

\newcommand{\F}{\mathcal{F}}
\newcommand{\G}{\mathcal{G}}

\newcommand{\z}{\hat{z}}
\newcommand{\T}{\hat{T}}
\renewcommand{\P}{\hat{P}}
\renewcommand{\S}{\hat{S}}

\DeclareMathOperator{\Ai}{Ai}
\DeclareMathOperator{\Tr}{Tr}
\DeclareMathOperator{\diag}{diag} 
\DeclareMathOperator{\supp}{supp}

\DeclareMathOperator*{\arccot}{arccot}
\DeclareMathOperator*{\const}{const}

\title{Random matrices with equispaced external source}

\author{Tom Claeys\thanks{Universit\'{e} Catholique de Louvain, Chemin du cyclotron 2, B-1348 Louvain-La-Neuve, Belgium \newline
email: \href{mailto:tom.claeys@uclouvain.be}{\protect\nolinkurl{tom.claeys@uclouvain.be}}} \ and
Dong Wang\thanks{Department of Mathematics, National University of Singapore, Singapore 119076 \newline
  email: \href{mailto:matwd@nus.edu.sg}{\protect\nolinkurl{matwd@nus.edu.sg}}}}

\newtheorem{cor}{Corollary}
\newtheorem{thm}{Theorem}
\newtheorem{prop}{Proposition}
\newtheorem{lem}{Lemma}

\theoremstyle{remark}
\newtheorem{rmk}{Remark}
\numberwithin{equation}{section}
\begin{document}

\maketitle

\begin{abstract}
We study Hermitian random matrix models with an external source matrix which has equispaced eigenvalues, and with an external field such that the limiting mean density of eigenvalues is supported on a single interval as the dimension tends to infinity. We obtain strong asymptotics for the multiple orthogonal polynomials associated to these models, and as a consequence for the average characteristic polynomials. One feature of the multiple orthogonal polynomials analyzed in this paper is that the number of orthogonality weights of the polynomials grows with the degree. Nevertheless we are able to characterize them in terms of a pair of $2\times 1$ vector-valued Riemann-Hilbert problems, and to perform an asymptotic analysis of the Riemann-Hilbert problems.
\end{abstract}

\section{Introduction}
We consider random matrix ensembles under the influence of an
external source matrix with equidistant eigenvalues. The ensembles
consist of the space of $n\times n$ Hermitian matrices with a
probability distribution of the form
\begin{equation}\label{prob}
\frac{1}{Z_n}\exp(-n\Tr \left[V(M)-AM\right])dM,
\end{equation}
where
\begin{equation}
  d M=\prod_{i<j}d \Re M_{ij}d \Im M_{ij}\ \prod_{j=1}^n d M_{jj}.
\end{equation}
The \emph{external field} $V(x)$ is a real analytic function which has sufficiently fast growth at infinity,
\begin{equation} \label{Vinfty}
  \lim_{x \to \pm\infty}\frac{V(x)}{|x|+1}=+\infty,
\end{equation}
and the \emph{external source matrix} $A$ is a deterministic Hermitian matrix. Due to the unitary invariance of the model, we assume, without loss of generality,
\begin{equation} \label{eq:A_is_diagonal}
  A = \diag(a_1, a_2, \dotsc, a_n).
\end{equation}
In our paper, we further assume the eigenvalues of $A$ are equispaced on a certain interval, such that $a_j = a(j-1)/n + b$ where $a$ and $b$ are constants. By arguments of shifting and scaling, it suffices to consider the case
\begin{equation} \label{eq:a_j_are_equispaced}
  a_{j} = \frac{j-1}{n}, \quad \text{where $j = 1, 2, \dotsc, n$,}
\end{equation}
and we work with the external source matrix $A$ given by \eqref{eq:A_is_diagonal} and \eqref{eq:a_j_are_equispaced} throughout this paper. The normalization constant $Z_n$ in \eqref{prob} depends on $n$ and $V$. In the
simplest example we have $V(x)=\frac{x^2}{2}$, which gives the
Gaussian Unitary Ensemble (GUE) in external source $A$. If we allow singularities of $V$, and let $V(x) = (x - \frac{m}{n} \log x) \chi_{x > 0}$, we have the complex Wishart ensemble that has wide applications in statistics and wireless communication, see \eg\ \cite{Baik-Ben_Arous-Peche05}.

\medskip

Random matrix ensembles with external source were introduced in
\cite{Brezin-Hikami98, Zinn_Justin97}, and are intimately connected
to multiple orthogonal polynomials \cite{Bleher-Kuijlaars04a}.
If the external field is the classical one $V(x) = \frac{x^2}{2}$ or $V(x) = (x - \frac{m}{n} \log x) \chi_{x > 0}$, \ie, the ensemble becomes the GUE with external source or the complex Wishart ensemble, more techniques are available for asymptotic analysis, and for a large class of external source matrices, including the equispaced one defined by \eqref{eq:A_is_diagonal} and \eqref{eq:a_j_are_equispaced}, the asymptotics can be obtained. See \cite{El_Karoui07} for the complex Wishart ensemble. However, when the external field $V(x)$ is general, the asymptotic analysis of the random matrix ensembles with external source has only had success for particular choices of external source matrices. Asymptotics for large $n$ have been studied in
\cite{Bleher-Kuijlaars04,
Bleher-Kuijlaars05,Aptekarev-Bleher-Kuijlaars05, Bleher-Kuijlaars07, Adler-van_Moerbeke07,
Bleher-Delvaux-Kuijlaars10, Aptekarev-Lysov-Tulyakov11} in the case
where the external source matrix $A$ has two different eigenvalues $a$ and $-a$ with equal
multiplicity, and in \cite{Baik-Wang10a, Baik-Wang10,
Bertola-Buckingham-Lee-Pierce11, Bertola-Buckingham-Lee-Pierce11a}
when $A$ has a bounded, or slowly growing with $n$, number of
non-zero eigenvalues. Large $n$ asymptotics for general external
source matrices have been considered in the physics literature, see
\eg\ \cite{Eynard-Orantin09}, but rigorous asymptotic results have
not been obtained to the best of our knowledge except for the two
above-mentioned cases. We remark that the GUE with external source and the complex Wishart ensemble have other generalizations, the complex Wigner matrix model with external source and the complex sample covariance matrix model respectively. They have also been studied extensively, see \eg\ \cite{Bai-Silverstein10}.

\medskip

Let us first recall some general properties about random
matrix ensembles with external source. An ensemble of the form \eqref{prob} with eigenvalues of the external source matrix being $a_1, \dotsc, a_n$
induces a probability distribution on
the eigenvalues $\lambda_1, \ldots, \lambda_n$ of the matrices given
by \cite{Brezin-Hikami98, Harish-Chandra57, Itzykson-Zuber80}
\begin{equation}\label{jpdf0}
 \frac{1}{Z'_n} \det(e^{na_i\lambda_j})^n_{i,j=1} \, \Delta(\lambda) \, \prod_{j=1}^n e^{-nV(\lambda_j)}\ \prod_{j=1}^n d\lambda_j,
\end{equation}
where $Z'_n =\const \cdot Z_n \cdot \Delta(a)$, and $\Delta(\lambda) = \prod_{i < j} (\lambda_j - \lambda_i)$ and $\Delta(a) = \prod_{i<j} (a_j - a_i)$ are Vandermonde determinants.
A remarkable fact is that the average characteristic
polynomial of such an ensemble \eqref{prob} satisfies orthogonality
conditions: indeed, let
\begin{equation}\label{averagechar}
  p_n^{(n)}(z) \colonequals \mathbb E_n(\det(zI-M))=\mathbb E_n'(\prod_{j=1}^n(z-\lambda_j)),
\end{equation}
where $\mathbb E_n$ is the average with respect to \eqref{prob}, and
$\mathbb E_n'$ is the average with respect to \eqref{jpdf0}, then it
was proved in \cite{Bleher-Kuijlaars04a} that $p_n^{(n)}$ is
characterized as the unique monic polynomial of degree $n$
satisfying the orthogonality conditions
\begin{equation}\label{orthoII}
\int_{\mathbb R}p_n^{(n)}(x)e^{na_jx}e^{-nV(x)}dx=0,\quad\text{for
$ j=1,\ldots, n.$}
\end{equation}
These are the orthogonality conditions for the so-called type II
multiple orthogonal polynomials with respect to $n$ different
orthogonality weights $e^{na_jx}e^{-nV(x)}$, $j=1,\ldots, n$.
Specialized to our situation $a_j=\frac{j-1}{n}$ for $j=1,\ldots,
n$, the joint probability distribution of the eigenvalues takes the
form
\begin{equation}\label{jpdf}
  \frac{1}{Z'_n}\prod_{i<j} (\lambda_j - \lambda_i) \
  \prod_{i<j} (e^{\lambda_j}-e^{\lambda_i})\ \prod_{j=1}^n
  e^{-nV(\lambda_j)}\ \prod_{j=1}^n d\lambda_j,
\end{equation}
and the monic type II multiple orthogonal polynomials $p^{(n)}_j(x)$, where $j = 0, 1, \dotsc$ is the degree, are characterized by
\begin{equation}\label{orthoIIb}
  \int_{\mathbb R}p^{(n)}_j(x)e^{kx}e^{-nV(x)}dx=0,\quad\text{for $k=0,\ldots, j-1$.}
\end{equation}

It is well-known that the point process \eqref{jpdf0} is
determinantal \cite{Zinn_Justin97}, and its two-point correlation
kernel can be written in terms of multiple orthogonal polynomials.
If $a_j=\frac{j-1}{n}$, the kernel takes the form
\cite{Bleher-Kuijlaars04a}
\begin{equation}\label{kernel}
  K_n(x,y)=e^{-\frac{n}{2}(V(x)+V(y))}\sum_{j=0}^{n-1}\frac{1}{h_j^{(n)}} p_j^{(n)}(x)Q_j^{(n)}(y),
\end{equation}
where $p_j^{(n)}(x)$ are the type II monic multiple orthogonal
polynomials characterized by \eqref{orthoIIb}, and $Q_j^{(n)}(y)=q^{(n)}_j(e^y)$ are linear combinations of $e^{ky}$ with $k = 0, 1, 2, \dotsc, j$, subjected to the orthogonality conditions
\begin{equation}\label{orthoI}
  \int_{\mathbb R} x^k q_j^{(n)}(e^x)e^{-nV(x)}dx=0,\qquad\mbox{ for
    $k=0,\ldots, j-1$,}
\end{equation}
where $q^{(n)}_j$ is a monic polynomial of degree $j$.
Finally the constants $h_j^{(n)}$ are given by
\begin{equation} \label{kappa}
  h_j^{(n)}=\int_{\mathbb R}p_j^{(n)}(x)q_j^{(n)}(e^x)e^{-nV(x)}dx.
\end{equation}
The orthogonality conditions (\ref{orthoIIb}) and (\ref{orthoI}) for
$p_j^{(n)}$ and $q_j^{(n)}$ can also be written at once as
\begin{equation}\label{ortho}
\int_{\mathbb
R}p_j^{(n)}(x)q_k^{(n)}(e^x)e^{-nV(x)}dx=0,\qquad\mbox{ for $j\neq
k\in\mathbb N=\{1,2,\ldots\}$.}
\end{equation}
Note that the multiple weights $e^{kx} e^{-nV(x)}$ constitute an AT system \cite[Section 4.4]{Nikishin-Sorokin91}, and hence $p^{(n)}_j$ and $q^{(n)}_j$ are uniquely defined, and $h^{(n)}_j \neq 0$ \cite{Coussement-Van_Assche01}.
\begin{rmk}
  As the counterpart of $p^{(n)}_j(x)$, $Q^{(n)}_j(x)$ is the $j$-th multiple orthogonal polynomial of type I, up to the constant factor $h_j^{(n)}$. Generally the type I multiple orthogonal polynomials are not polynomials, but in the present setting, $Q^{(n)}_j(x)$ is a polynomial in $e^x$.
\end{rmk}
\begin{rmk}
If the external field $V$ is a quadratic polynomial, distributions of the form (\ref{jpdf0}) can also be realized in models consisting of $n$ non-intersecting Brownian motions. In particular,
(\ref{jpdf}) is the joint probability distribution at an intermediate time of $n$ non-intersecting Brownian motions starting at one point and ending at $n$ equidistant points. Such a model has been studied in \cite{Johansson04}. Different endpoint configurations have been studied e.g.\ in \cite{Adler-Delepine-van_Moerbeke09, Adler-Orantin-van_Moerbeke10}.
\end{rmk}
In analogy to
(\ref{averagechar}), $q_n^{(n)}$ can also be interpreted as an
average over the determinantal point process (\ref{jpdf}). We will
prove the following result in Appendix \ref{subsec:proof_of_algebraic_prop}.
\begin{prop}\label{propq}
Let $V$ be real analytic satisfying \eqref{Vinfty}.
  We have the identities
\begin{equation}\label{q}
  q_n^{(n)}(e^z)= \mathbb{E}_n(\det(e^{zI} - e^{M})) = \mathbb E_n'(\prod_{j=1}^n(e^z-e^{\lambda_j})),
\end{equation}
where $\mathbb{E}_n$ denotes the expectation associated to \eqref{prob} with $A$ given by \eqref{eq:A_is_diagonal}--\eqref{eq:a_j_are_equispaced}, and $\mathbb{E}_n'$ is the expectation associated to \eqref{jpdf}.
\end{prop}

\medskip

The main goal of this paper is to obtain asymptotics for the average
characteristic polynomials $p_n^{(n)}$ of the random matrix ensemble
as $n\to\infty$. In addition we will also obtain asymptotics for the
dual polynomials $q_n^{(n)}$.
A key observation is that $p_{n+k}^{(n)}$ and $q_{n+k}^{(n)}$ can be characterized in terms of $1\times 2$ vector-valued Riemann-Hilbert (RH) problems.
These RH problems are different from the known $(n+k+1) \times (n+k+1)$ RH
problems characterizing the multiple orthogonal polynomials $p^{(n)}_{n+k}$ and $q^{(n)}_{n+k}$ \cite{Geronimo-Kuijlaars-Van_Assche00} and from the classical RH problem for orthogonal polynomials \cite{Fokas-Its-Kitaev92}.
Since $n$ is a large parameter in our settings, the $1\times 2$ RH
problem will be much more convenient for asymptotic analysis than a RH problem of large size. As a
drawback, our RH problem is non-standard in the sense that the
entries of the solution live in different domains. This requires a
modification of the Deift/Zhou steepest descent method
to analyze the RH problem asymptotically. The transformation
$\Jlike$ will play a crucial role here: it allows us to transform the $1\times 2$ RH problem to
a scalar shifted RH problem, and to obtain small norm estimates for the solution to this RH problem.

\medskip

 A
crucial role in the description of the asymptotic behavior of the
polynomials will be played by an equilibrium measure. By
(\ref{jpdf}), the joint probability density function of eigenvalues
is maximal for the $n$-tuples $(\lambda_1,\ldots, \lambda_n)$ for
which
\begin{equation}
  \sum_{i<j}\log|\lambda_i-\lambda_j|^{-1} + \sum_{i<j}\log|e^{\lambda_i}-e^{\lambda_j}|^{-1}+ n\sum_{j=1}^n V(\lambda_j)
\end{equation}
is minimal. As in \cite[Section 6.1]{Deift99}, one
can then expect that the limiting mean distribution of the
eigenvalues of the random matrices is given by the equilibrium
measure $\mu_V$ which minimizes the energy functional
\begin{equation} \label{energy}
  I_V(\mu) \colonequals \frac{1}{2} \iint \log \lvert t-s \rvert^{-1} d\mu(t)d\mu(s) + \frac{1}{2} \iint \log \lvert e^t-e^s \rvert^{-1}  d\mu(t)d\mu(s) + \int V(s)d\mu(s),
\end{equation}
among all Borel probability measures $\mu$ supported on $\mathbb R$.
This is in analogy to the equilibrium measure corresponding to a
matrix model of the form (\ref{prob}) without external source, which
is given as the unique minimizer of the energy
\begin{equation}\label{energy1}
\iint \log \lvert t-s \rvert^{-1} d\mu(t)d\mu(s) + \int V(s)d\mu(s).
\end{equation}
Following the proof in \cite{Deift99} of existence and uniqueness of
the minimizer of (\ref{energy1}), we will show existence and
uniqueness of the equilibrium measure minimizing (\ref{energy}).

\begin{thm}\label{theorem: exist uni}
Let $V$ be real analytic, satisfying the growth condition
\eqref{Vinfty}. Then there exists a unique measure $\mu=\mu_V$ with compact support which
minimizes the functional \eqref{energy} among all probability
measures on $\mathbb R$.
\end{thm}

\begin{rmk}
It should be noted that the growth condition \eqref{Vinfty} is
stronger than the usual logarithmic growth needed to have a unique
minimizer of the one-matrix logarithmic energy \eqref{energy1}. This
is a consequence of the second term in (\ref{energy}).
\end{rmk}
The proof of this result will be given in Section \ref{section:
exist uni} but it does not give any information about the measure
$\mu_V$ itself. For that reason, in what follows, we will restrict
ourselves to a class of external fields $V$ for which the
equilibrium measure behaves nicely and is supported on a single
interval.

We say that a real analytic external field $V$ satisfying
(\ref{Vinfty}) is {\em one-cut regular} if there exists an
absolutely continuous measure $d\mu_V(x)=\psi_V(x)dx$ satisfying the
properties
\begin{enumerate}[label=(\roman{*})]
\item \label{enu:equ_measure_intro_1}
  $\supp\,\mu_V=[a,b]$ for $a<b\in\mathbb R$, and $\int d\mu_V(x) = 1$,
\item \label{enu:equ_measure_intro_2}
  $\psi_V(x)>0$ for $x\in(a,b)$,
\item \label{enu:equ_measure_intro_3}
  $\lim_{x\to a_+}\frac{\psi_V(x)}{\sqrt{x-a}}$ and $\lim_{x\to b_-}\frac{\psi_V(x)}{\sqrt{b-x}}$ exist and are different from zero,
\item \label{enu:equ_measure_intro_4}
  for $x\in[a,b]$, there exists a constant $\ell$ depending on $V$
such that
\begin{equation}\label{var eq}
\int \log \lvert t-x \rvert^{-1} d\mu_V(t) + \int \log \lvert
e^t-e^x \rvert^{-1} d\mu_V(t) + V(x)+\ell=0,
\end{equation}
\item \label{enu:equ_measure_intro_5}
  for $x\in\mathbb R\setminus[a,b]$, we have
\begin{equation}\label{var ineq}
\int \log \lvert t-x \rvert^{-1} d\mu_V(t) + \int \log \lvert
e^t-e^x \rvert^{-1} d\mu_V(t) + V(x)+\ell>0.
\end{equation}
\end{enumerate}
Properties \ref{enu:equ_measure_intro_4} and \ref{enu:equ_measure_intro_5} are variational conditions for $\mu_V$, and it follows from standard arguments that a measure satisfying \ref{enu:equ_measure_intro_1}, \ref{enu:equ_measure_intro_2}, \ref{enu:equ_measure_intro_4} and \ref{enu:equ_measure_intro_5} minimizes the energy functional (\ref{energy}).
Under the condition that $V$ is one-cut regular, we obtain large $n$ asymptotics for
$p_n^{(n)}(z)$ and $q_n^{(n)}(e^z)$ defined by \eqref{ortho}, and state it in the following theorem. For the purpose of a subsequent paper, we give slightly more general asymptotics for $p^{(n)}_{n+k}(z)$ and $q^{(n)}_{n+k}(e^z)$, where $k$ is a constant integer.

\medskip

Suppose the
equilibrium measure $\mu_V$ associated to $V$ is supported on a single interval $[a,b]$ and the density function is $\psi_V(x)$. In order to be
able to formulate our results, let us define $c_0 \in \realR$ and
$c_1 \in \realR^+$ such that
\begin{gather}
  c_0 = {} \frac{b + a}{2}, \label{eq:first_def_of_c_0} \\
  c_1 \sqrt{\frac{1}{4} + \frac{1}{c_1}} - \log \frac{\sqrt{\frac{1}{4} + \frac{1}{c_1}} - \frac{1}{2}}{\sqrt{\frac{1}{4} + \frac{1}{c_1}} + \frac{1}{2}}=\frac{b - a}{2}. \label{eq:first_def_of_c_1}
\end{gather}
Note that $c_1$ is well defined since as $c_1$ runs from $0$ to
$+\infty$, the left-hand side of \eqref{eq:first_def_of_c_1}
increases monotonically from $0$ to $+\infty$. Then we define the
transformation
\begin{equation} \label{eq:Joukowsky_like_transform}
  \Jlike(s) = \Jlike_{c_1, c_0}(s)  \colonequals  c_1 s + c_0 - \log\frac{s - \frac{1}{2}}{s +
  \frac{1}{2}}
\end{equation}
for $s \in \compC \setminus [-\frac{1}{2}, \frac{1}{2}]$, where the logarithm corresponds to arguments between $-\pi$ and $\pi$. For $s < -\frac{1}{2}$, $\Jlike_{c_1, c_0}(s)$ has a maximum at $s_a$, and for $s > \frac{1}{2}$, $\Jlike_{c_1,
c_0}(s)$ has a minimum at $s_b$, where
\begin{equation} \label{eq:defn_of_s_a_and_s_b}
  s_a = -\sqrt{\frac{1}{4}+\frac{1}{c_1}}, \quad s_b = \sqrt{\frac{1}{4}+\frac{1}{c_1}}.
\end{equation}
The extrema $s_a$ and $s_b$ are also characterized by identities $a=\Jlike_{c_1, c_0}(s_a)$ and $b=\Jlike_{c_1, c_0}(s_b)$.

In Section \ref{subsec:g_fun_and_equ_measure}, a region $D
\subset \compC$ is defined by Proposition \ref{prop:Joukowsky_like}, and it is shown there that $\Jlike$ maps
$\compC \setminus \overline{D}$ biholomorphically into $\mathbb
C\setminus[a,b]$, and maps $D\setminus[-\frac{1}{2},\frac{1}{2}]$
biholomorphically into $\strip \setminus [a, b]$, where
\begin{equation}\label{strip}
  \strip  \colonequals  \{ z\in \compC \mid -\pi < \Im z < \pi \}.
\end{equation}
 Let the functions $\Jinv_1$ and $\Jinv_2$ be inverse functions of
$\Jlike$ for these two branches respectively: $\Jinv_1$ is the
inverse map of $\Jlike_{c_1,c_0}$ from $\mathbb C\setminus[a,b]$ to
$\mathbb C\setminus \overline D$, and $\Jinv_2$ is the inverse map
of $\Jlike_{c_1,c_0}$ from $\strip \setminus[a,b]$ to $D \setminus [-\frac{1}{2}, \frac{1}{2}]$:
\begin{align}
 \Jinv_1(\Jlike(s))= {}& s,& & \text{for $s\in \mathbb C\setminus\overline D$,} \label{eq:defn_of_I_1} \\
 \Jinv_2(\Jlike(s))= {}& s,& & \text{for $s\in D \setminus [-\frac{1}{2}, \frac{1}{2}]$}. \label{eq:defn_of_I_2}
\end{align}
 Writing, for $x \in (a, b)$,
\begin{align}
  \Jinv_+(x) \colonequals {}& \lim_{\epsilon\to 0_+} \Jinv_1(x+i\epsilon)=\lim_{\epsilon\to 0_+} \Jinv_2(x-i\epsilon), \label{eq:defn_of_I_+} \\
  \Jinv_-(x) \colonequals {}& \lim_{\epsilon\to 0_+} \Jinv_1(x-i\epsilon)=\lim_{\epsilon\to 0_+} \Jinv_2(x+i\epsilon), \label{eq:defn_of_I_-}
\end{align}
we have that $\Jinv_+(x)$ lies in the upper half plane,
$\Jinv_-(x)$ in the lower half plane, and their loci are the upper and lower boundaries of $D$ (denoted as $\gamma_1$ and $\gamma_2$ in Proposition \ref{prop:Joukowsky_like}) respectively.
The mapping $\Jlike$ outside and inside $D$ is illustrated in Figures \ref{fig:domain_D} and \ref{fig:domain_D_with_cut}. The proof of Proposition \ref{prop:Joukowsky_like} is given in Appendix \ref{subsec:Proof_of_Prop_Joukowsky_like}. In Figure \ref{fig:examples_of_gammas} we give two examples of $\gamma_1$ and $\gamma_2$ by numerical simulation.
\begin{figure}[h]
  \centering
  \includegraphics{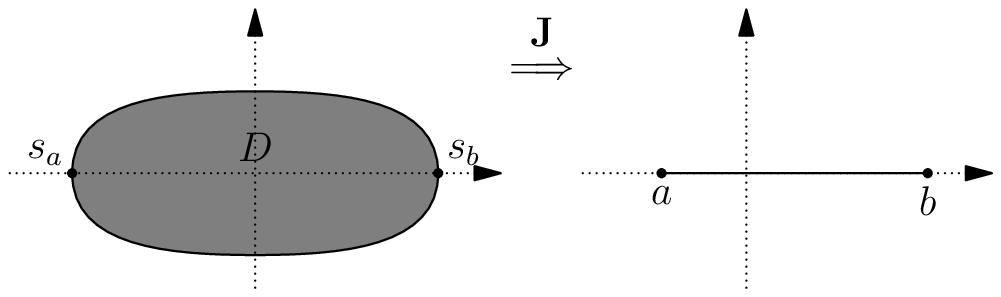}
  \caption{Mapping $\Jlike$ outside $D$.}
  \label{fig:domain_D}
\end{figure}

\begin{figure}[h]
  \centering
  \includegraphics{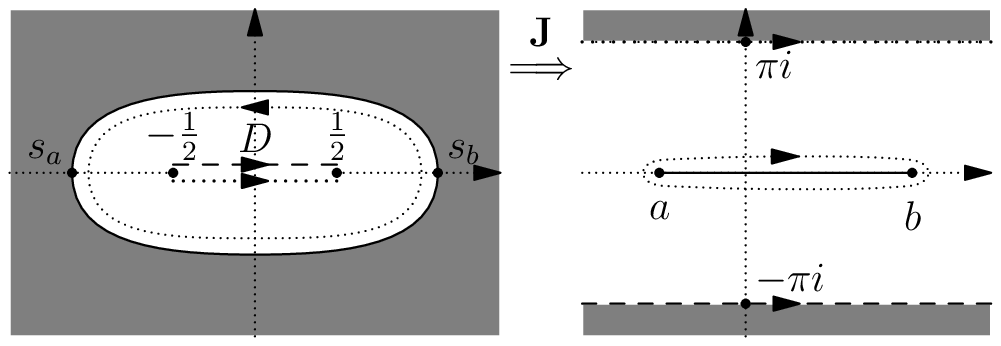}
  \caption{Mapping $\Jlike$ inside $D$.}
  \label{fig:domain_D_with_cut}
\end{figure}

\begin{figure}[ht]
  \centering
  \begin{minipage}[b]{0.35\linewidth}
    \includegraphics[width=\linewidth]{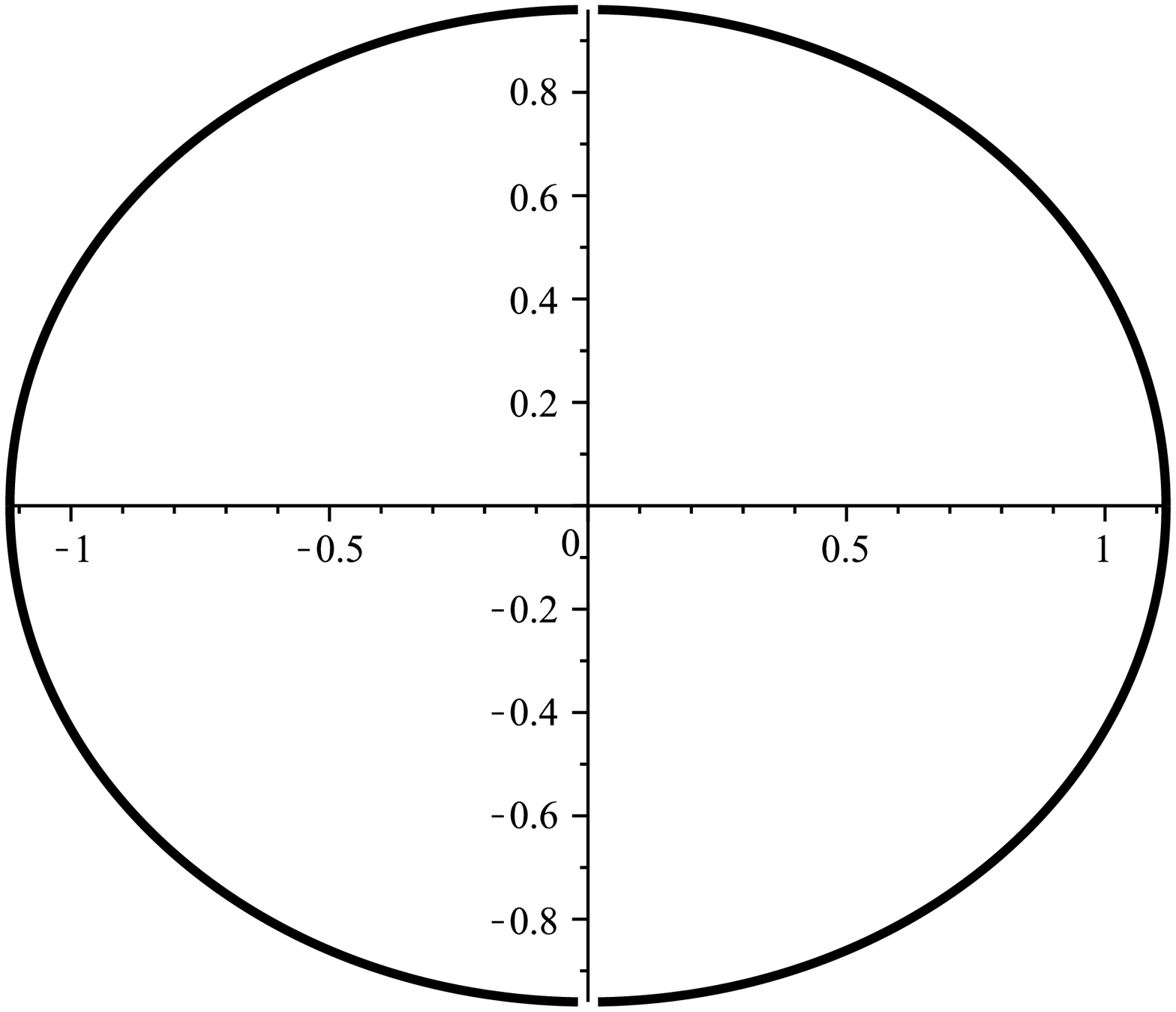}
  \end{minipage}
  \begin{minipage}[b]{0.35\linewidth}
    \includegraphics[width=\linewidth]{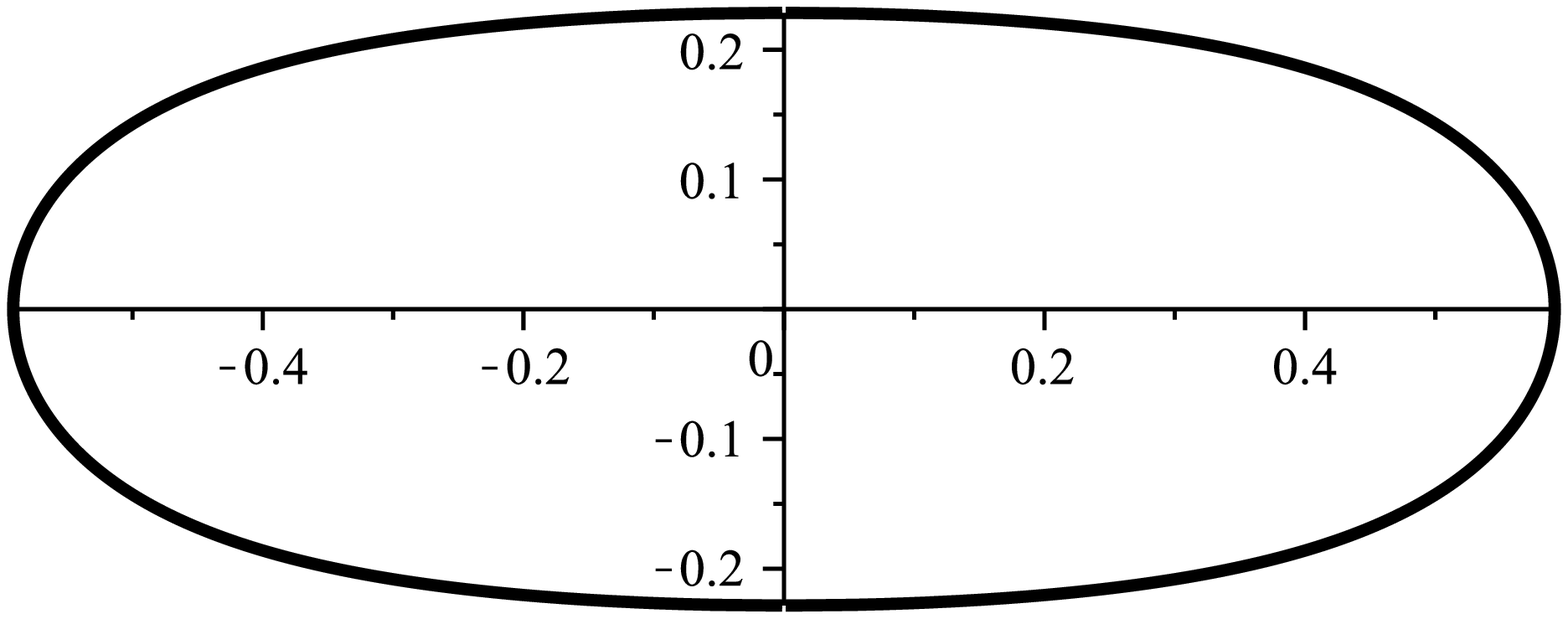}
  \end{minipage}
  \caption{The shapes of $\gamma_1$ and $\gamma_2$ when $c_1 = 1$ (left) and $c_1 = 10$ (right), where $\gamma_1$ is the upper boundary of the region and $\gamma_2$ the lower boundary.}
  \label{fig:examples_of_gammas}
\end{figure}

Let the functions $G_k(s)$ and $\hat{G}_k(s)$ be defined as
\begin{equation} \label{eq:defn_of_G_and_G_hat}
  G_k(s)  \colonequals  c^k_1 \frac{(s + \frac{1}{2})(s - \frac{1}{2})^k}{\sqrt{s^2 - \frac{1}{4} - \frac{1}{c_1}}}, \quad \hat{G}_k(s)  \colonequals  i \frac{e^{k(\frac{c_1}{2} + c_0)}}{\sqrt{c_1}} \frac{(s - \frac{1}{2})^{-k}}{\sqrt{s^2 - \frac{1}{4} - \frac{1}{c_1}}},
\end{equation}
where the square root $\sqrt{s^2 - \frac{1}{4} - \frac{1}{c_1}}$ has its branch cut along the upper edge of $D$ ($\gamma_1$ defined in Proposition \ref{prop:Joukowsky_like}) in $G_k(s)$, along the lower edge of $D$ ($\gamma_2$ defined in Proposition \ref{prop:Joukowsky_like}) in $\hat{G}_k(s)$, and $\sqrt{s^2 - \frac{1}{4} - \frac{1}{c_1}} \sim s$ as $s \to \infty$ in both cases. Further we define
\begin{align}
  r_k(x)  \colonequals  {}& 2 \lvert G_k(\Jinv_+(x)) \rvert, & \theta_k(x)  \colonequals {}& \arg(G_k(\Jinv_+(x))), \label{eq:defn_r_theta} \\
  \hat{r}_k(x)  \colonequals {}& 2 \lvert \hat{G}_k(\Jinv_-(x)) \rvert, & \hat{\theta}_k(x)  \colonequals {}& \arg(\hat{G}_k(\Jinv_-(x))), \label{eq:defn_r_theta_hat}
\end{align}
for $x$ in $(a, b)$.

We also need to define the functions
\begin{equation} \label{eq:expr_of_gfn_and_tilde_gfn-intro}
\gfn(z)  \colonequals  \int^b_a \log(z-x) \psi_V(x) dx,\quad \tilde \gfn(z)  \colonequals  \int^b_a \log(e^z-e^x) \psi_V(x) dx,
\end{equation}
with the branch cut of the logarithms for $z \in (-\infty, x)$ and $e^z \in (0, e^x)$, and $\psi_V$ is the equilibrium density. Let
\begin{equation}\label{def phi}
  \phi(z) \colonequals \gfn(z)+\tilde{\gfn}(z)-V(z)-\ell
\end{equation}
for $z \in \strip \setminus (-\infty, b)$, where $\ell$ is a constant to make $\phi(a) = \phi(b) = 0$ (see \eqref{var eq} and \eqref{var eq g}). Then we will see later on, see Section \ref{section: local par}, that
\begin{equation}
  f_b(z)  \colonequals  \left( -\frac{3}{4}\phi(z) \right)^{\frac{2}{3}}
\end{equation}
is a well defined analytic function in a certain neighborhood of $b$, with $f_b(b) = 0$, $f'_b(b) > 0$. Similarly,
\begin{equation}
  f_a(z)  \colonequals  \left( -\frac{3}{4}\phi(z) \pm \frac{3}{2} \pi i \right)^{\frac{2}{3}}
\end{equation}
(where the sign is $+$ in $\compC^+$ and $-$ in $\compC^-$,) is a well defined analytic function in a certain neighborhood of $a$, with $f_a(a) = 0$, $f'_a(a) < 0$.

Since both $p^{(n)}_j(z)$ and $q^{(n)}_j(e^z)$ are analytic functions that are real for $z\in\mathbb R$, it suffices to give their asymptotics in the upper half plane and the real line. For the ease of the statement of the theorem, we divide the upper half plane into regions $A_\delta$, $B_\delta$, $C_\delta$ and $D_\delta$ where $\delta$ is a small enough positive parameter, such that $C_\delta$ and $D_\delta$ are semicircles with radius $\delta$ and centered at $a$ and $b$ respectively, $B_\delta$ consists of complex numbers not in $C_\delta$ or $D_\delta$, with real part between $a$ and $b$ and imaginary part between $0$ and $\frac{\delta}{2}$, and $A_\delta = \compC^+ \setminus (B_\delta \cup C_\delta \cup D_\delta)$. See Figure \ref{figure:four_regions} for the shapes of the four regions.

   \begin{figure}[h]
   \begin{center}
     \includegraphics{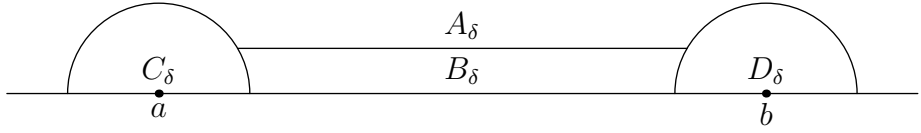}
   \caption{The four regions in the complex upper half plane where asymptotics for the multiple orthogonal polynomials $p^{(n)}_{n+k}(z)$
   and $q^{(n)}_{n+k}(e^z)$ will be given in different formulas.}
\label{figure:four_regions}
\end{center}
\end{figure}

\begin{thm}\label{theorem: asympt2}
  Let $V$ be one-cut regular.
  As $n\to\infty$, we have the following asymptotics of $p^{(n)}_{n+k}(z)$ and $q^{(n)}_{n+k}(e^z)$, $k\in\mathbb Z$, uniform for $z$ in regions $A_{\delta}$, $B_{\delta}$, $C_{\delta}$ and $D_{\delta}$, if $\delta$ is small enough.
  \begin{enumerate}[label=(\alph{*})]
  \item \label{enu:theorem: asympt2:A}
    In region $A_{\delta}$ and on its boundary,
    \begin{align}
      p^{(n)}_{n+k}(z) = {}& (1+\bigO(n^{-1})) G_k(\Jinv_1(z)) e^{n\gfn(z)},  \label{eq:asy_of_p_outside} \\
      q^{(n)}_{n+k}(e^z) = {}& (1+\bigO(n^{-1})) \hat{G}_k(\Jinv_2(z)) e^{n\tilde\gfn(z)}, \label{eq:asy_of_q_outside}
    \end{align}
    where \eqref{eq:asy_of_q_outside} is valid for $\Im z < \pi$.
  \item \label{enu:theorem: asympt2:B}
    In region $B_{\delta}$,
    \begin{align}
      p^{(n)}_{n+k}(z) ={}& (1 + \bigO(n^{-1})) G_k(\Jinv_1(z)) e^{n\gfn(z)} + (1 + \bigO(n^{-1})) G_k(\Jinv_2(z)) e^{n(V(z) - \tilde{\gfn}(z) + \ell)}, \label{eq:asy_of_p^n_n+k_bulk_region} \\
      q^{(n)}_{n+k}(e^z) ={}& (1+\bigO(n^{-1})) \hat{G}_k(\Jinv_2(z)) e^{n\tilde\gfn(z)} + (1+\bigO(n^{-1})) \hat{G}_k(\Jinv_1(z)) e^{n(V(z) -\gfn(z) + \ell)}. \label{eq:asy_of_q^n_n+k_bulk_region}
    \end{align}
    If $x \in (a,b)$ is on the boundary of region $B_{\delta}$, then
    \begin{align}
      p^{(n)}_{n+k}(x) = {}& r_k(x)e^{n\int\log|x-y|d\mu_V(y)}\left[\cos \left( n\pi \int_{x}^bd\mu_V(t) +\theta_k(x)\right) + \bigO(n^{-1}) \right], \label{eq:asy_of_p_bulk_real} \\
      q^{(n)}_{n+k}(e^x) = {}& \hat{r}_k(x) e^{n\int \log \lvert e^x - e^y \rvert d\mu_V(y)} \left[\cos \left( n\pi \int^b_x d\mu_V(t) + \hat{\theta}_k(x)\right) + \bigO(n^{-1}) \right].
    \end{align}
  \item \label{enu:theorem: asympt2:C}
    In region $C_{\delta}$, let $\Ai$ denote the Airy function \cite{Abramowitz-Stegun64}. Then
    \begin{multline}
      p^{(n)}_{n+k}(z) = \sqrt{\pi} \left[ \left( (1 + \bigO(n^{-1})) G_k(\Jinv_1(z)) - (1 + \bigO(n^{-1})) iG_k(\Jinv_2(z)) \right) n^{\frac{1}{6}} f^{\frac{1}{4}}_a(z) \Ai(n^{\frac{2}{3}}f_a(z)) \right. \\
      - \left. \left( (1 + \bigO(n^{-1})) G_k(\Jinv_1(z)) + (1 + \bigO(n^{-1})) iG_k(\Jinv_2(z)) \right) n^{-\frac{1}{6}} f^{-\frac{1}{4}}_a(z) \Ai'(n^{\frac{2}{3}}f_a(z)) \right] \\
      \times e^{\frac{n}{2} (\gfn(z) - \tilde{\gfn}(z) + V(z) + \ell)},
    \end{multline}
    \begin{multline}
      q^{(n)}_{n+k}(e^z) = \sqrt{\pi} \left[ \left( (1 + \bigO(n^{-1}))\hat{G}_k(\Jinv_2(z)) - (1 + \bigO(n^{-1})) i\hat{G}_k(\Jinv_1(z)) \right) n^{\frac{1}{6}} f^{\frac{1}{4}}_a(z) \Ai(n^{\frac{2}{3}}f_a(z)) \right. \\
      - \left. \left( (1 + \bigO(n^{-1})) \hat{G}_k(\Jinv_2(z)) + (1 + \bigO(n^{-1})) i\hat{G}_k(\Jinv_1(z)) \right) n^{-\frac{1}{6}} f^{-\frac{1}{4}}_a(z) \Ai'(n^{\frac{2}{3}}f_a(z)) \right] \\
      \times e^{\frac{n}{2} (\tilde{\gfn}(z) - \gfn(z) + V(z) + \ell)},
    \end{multline}
    where $f^{\frac{1}{4}}_a(z)$ has branch cut on $(a, b)$, and $f^{\frac{1}{4}}_a(x) > 0$ for $x < a$. In particular, if $z = a + f'_a(a)^{-1} n^{-2/3} t$ with $t$ bounded, then
    \begin{multline} \label{eq:Airy_p_C}
      e^{\frac{n}{2}(\tilde{\gfn}(z) - \gfn(z) - V(z) - \ell)} p^{(n)}_{n+k}(z) = \\
      (-1)^k \sqrt{2\pi} \left( \frac{1}{4} + \frac{1}{c_1} \right)^{-\frac{1}{8}} \left( \sqrt{\frac{1}{4} + \frac{1}{c_1}} + \frac{1}{2} \right)^{k-1} c_1^{k - \frac{1}{2}} (-f'_a(a))^{\frac{1}{4}} n^{\frac{1}{6}} \left( \Ai(t) + \bigO(n^{-\frac{1}{3}}) \right),
    \end{multline}
    \begin{multline} \label{eq:Airy_q_C}
      e^{\frac{n}{2}(\gfn(z) - \tilde{\gfn}(z) - V(z) - \ell)} q^{(n)}_{n+k}(e^z) = \\
      (-1)^k \sqrt{2\pi} \left( \frac{1}{4} + \frac{1}{c_1} \right)^{-\frac{1}{8}} \left( \sqrt{\frac{1}{4} + \frac{1}{c_1}} + \frac{1}{2} \right)^{-k} e^{k(\frac{c_1}{2} + c_0)} (-f'_a(a))^{\frac{1}{4}} n^{\frac{1}{6}} \left( \Ai(t) + \bigO(n^{-\frac{1}{3}}) \right).
    \end{multline}
  \item \label{enu:theorem: asympt2:D} In region $D_{\delta}$,
    \begin{multline} \label{eq:asy_of_p^n_n+k_right_edge_region}
      p^{(n)}_{n+k}(z) = \sqrt{\pi} \left[ \left( (1 + \bigO(n^{-1}))G_k(\Jinv_1(z)) - (1 + \bigO(n^{-1})) iG_k(\Jinv_2(z)) \right) n^{\frac{1}{6}} f^{\frac{1}{4}}_b(z) \Ai(n^{\frac{2}{3}}f_b(z)) \right. \\
      - \left. \left( (1 + \bigO(n^{-1})) G_k(\Jinv_1(z)) + (1 + \bigO(n^{-1})) iG_k(\Jinv_2(z)) \right) n^{-\frac{1}{6}} f^{-\frac{1}{4}}_b(z) \Ai'(n^{\frac{2}{3}}f_b(z)) \right] \\
      \times e^{\frac{n}{2} (\gfn(z) - \tilde{\gfn}(z) + V(z) + \ell)},
    \end{multline}
    \begin{multline}
      q^{(n)}_{n+k}(e^z) = \sqrt{\pi} \left[ \left( (1 + \bigO(n^{-1}))\hat{G}_k(\Jinv_2(z)) - (1 + \bigO(n^{-1})) i\hat{G}_k(\Jinv_1(z)) \right) n^{\frac{1}{6}} f^{\frac{1}{4}}_b(z) \Ai(n^{\frac{2}{3}}f_b(z)) \right. \\
      - \left. \left( (1 + \bigO(n^{-1})) \hat{G}_k(\Jinv_2(z)) + (1 + \bigO(n^{-1})) i\hat{G}_k(\Jinv_1(z)) \right) n^{-\frac{1}{6}} f^{-\frac{1}{4}}_b(z) \Ai'(n^{\frac{2}{3}}f_b(z)) \right] \\
      \times e^{\frac{n}{2} (\tilde{\gfn}(z) - \gfn(z) + V(z) + \ell)}.
    \end{multline}
    If $z = b + f'_b(b)^{-1} n^{-2/3} t$ with $t$ bounded, then
    \begin{multline} \label{eq:asy_of_p^n_n+k_very_near_b}
      e^{\frac{n}{2}(\tilde{\gfn}(z) - \gfn(z) - V(z) - \ell)} p^{(n)}_{n+k}(z) = \\
      \sqrt{2\pi} \left( \frac{1}{4} + \frac{1}{c_1} \right)^{-\frac{1}{8}} \left( \sqrt{\frac{1}{4} + \frac{1}{c_1}} - \frac{1}{2} \right)^{k-1} c^{k-\frac{1}{2}}_1 f'_b(b)^{\frac{1}{4}} n^{\frac{1}{6}} \left( \Ai(t) + \bigO(n^{-\frac{1}{3}}) \right),
    \end{multline}
    \begin{multline} \label{eq:Airy_q_D}
      e^{\frac{n}{2}(\gfn(z) - \tilde{\gfn}(z) - V(z) - \ell)} q^{(n)}_{n+k}(e^z) = \\
      \sqrt{2\pi} \left( \frac{1}{4} + \frac{1}{c_1} \right)^{-\frac{1}{8}} \left( \sqrt{\frac{1}{4} + \frac{1}{c_1}} - \frac{1}{2} \right)^{-k} e^{k(\frac{c_1}{2} + c_0)} f'_b(b)^{\frac{1}{4}} n^{\frac{1}{6}} \left( \Ai(t) + \bigO(n^{-\frac{1}{3}}) \right).
    \end{multline}
  \item \label{enu:theorem: asympt2:5}
    The inner product $h^{(n)}_{n+k}$ of $p^{(n)}_{n+k}(z)$ and $q^{(n)}_{n+k}(e^z)$ defined in \eqref{kappa} has the asymptotics
    \begin{equation}
      h^{(n)}_{n+k} = 2\pi c^{k+\frac{1}{2}}_1 e^{k(\frac{c_1}{2} + c_0)} e^{n\ell} (1 + \bigO(n^{-1})).
    \end{equation}
  \end{enumerate}
\end{thm}
The above result is only valid if the equilibrium measure $\mu_V$ is supported on a single interval. In the case of a multi-interval support, several non-trivial modifications are needed to make the asymptotic analysis of the polynomials work. For instance, the mapping $\Jlike$ would have to be modified.
In general it is not easy to determine whether an external field $V$ is one-cut regular or not, or to find the support $[a, b]$ of the
equilibrium measure and the density function $\psi_V$. However, if the external field is strongly convex, i.e.\ $V''(x)$ is bounded from below by a positive constant for $x\in\mathbb R$, then $V$ is one-cut regular, and we can compute the support and density function of the equilibrium measure explicitly in terms of the functions $\Jinv_\pm$ defined before.
\begin{thm}\label{theorem: convex}
  If $V$ is a real analytic strongly convex function, then $V$ is one-cut regular. Moreover, the quantities $c_0$ and $c_1$ that are related to the endpoints $a, b$ of the support of the equilibrium measure by \eqref{eq:first_def_of_c_0} and \eqref{eq:first_def_of_c_1} are obtained by solving a pair of equations \eqref{intro mod eq 1} and \eqref{intro mod eq 2} expressed in $V$, and $a, b$ are determined by $c_1$ and $c_0$ by \eqref{eq:defining_formula_of_a_b}, \eqref{eq:first_def_of_c_0} and \eqref{eq:first_def_of_c_1}. The density function $\psi_V$ is given by
  \begin{equation} \label{eq:formula_of_density_function_in_z}
    \psi_V(x)= \frac{1}{2\pi^2}\int_a^b V''(u)\log \left|\frac{\Jinv_+(u) - \Jinv_-(x)}{\Jinv_+(u)-\Jinv_+(x)}\right| d u.
  \end{equation}
\end{thm}

\begin{rmk}
The conditions of Theorem \ref{theorem: convex} are sufficient but
far from necessary to have one-cut regularity. See Example 2 in Appendix \ref{section: equi examples} for
non-convex one-cut regular external fields.
\end{rmk}

For the random matrix model without external source, it is well known that
\begin{enumerate}[label=(\arabic{*})]
\item \label{enu:3_convergence_of_equ_meansure:1} the empirical distribution of the eigenvalues of the random matrix,
\item \label{enu:3_convergence_of_equ_meansure:2} the normalized counting measure of the $n$-Fekete set,
\item \label{enu:3_convergence_of_equ_meansure:3} the normalized counting measure of the zeros of the orthogonal polynomial (which is the average characteristic polynomial of the random matrix),
\end{enumerate}
all converge to the equilibrium measure as the dimension $n \to \infty$. The counterpart of \ref{enu:3_convergence_of_equ_meansure:3} in our equispaced external source model, in case that the external field $V$ is one-cut regular, is a direct consequence of Theorem \ref{theorem: asympt2}\ref{enu:theorem: asympt2:B}.
\begin{cor} \label{cor:conv}
  Let $V$ be one-cut regular, and $p^{(n)}_n$ and $q^{(n)}_n$ be defined by \eqref{orthoIIb} and \eqref{orthoI} respectively. Suppose real numbers $z_j$ and $\z_j$ are zeros of $p^{(n)}_n(z)$ and $q^{(n)}_n(e^z)$ respectively, and $\mu_n=\frac{1}{n}\sum_{j=1}^n\delta_{z_j}$ and $\hat{\mu}_n=\frac{1}{n}\sum_{j=1}^n\delta_{\z_j}$ respectively. Then as $n\to\infty$, $\mu_n$ and $\hat\mu_n$ converge weakly to $\mu_V$.
\end{cor}
Counterparts of \ref{enu:3_convergence_of_equ_meansure:1}, \ref{enu:3_convergence_of_equ_meansure:2} and \ref{enu:3_convergence_of_equ_meansure:3} can also be proved by mimicking the arguments in \cite[Sections 6.3 and 6.4]{Deift99}. Although we are not going to pursue this approach, we remark that all the counterparts of \ref{enu:3_convergence_of_equ_meansure:1}--\ref{enu:3_convergence_of_equ_meansure:3} should not rely on the assumption of one-cut regularity.

\subsubsection*{Outline}

In Section \ref{section: exist uni}, we prove the uniqueness and existence of
the equilibrium measure, as stated in Theorem \ref{theorem: exist
uni}. In Section \ref{section:equi}, we explain in detail how one can construct the
equilibrium measure $\mu_V$ and its density in the case of a strongly convex
external field $V$, by solving a scalar RH problem and by using the
transformation $\Jlike$. This also leads to the proof
of Theorem \ref{theorem: convex}. In Section \ref{section: RH p}, we
characterize the polynomials $p^{(n)}_{n+k}$ in terms of a $1\times 2$ RH
problem, and we analyse this RH problem asymptotically for large $n$.
In Section \ref{section: RH
q}, we formulate a similar RH problem and perform a similar asymptotic analysis for the polynomials $q^{(n)}_{n+k}$. In Section \ref{section:
proof results}, we use the results obtained from the RH analysis to
prove Theorem \ref{theorem: asympt2} and Corollary \ref{cor:conv}.
In Appendix \ref{section: proofs lemmas}, we prove Proposition
\ref{propq} and several technical lemmas used in this paper. In Appendix \ref{section: equi examples} we give explicit formulas for the equilibrium measure for quadratic and quartic $V$ as examples. In Appendix \ref{sec:saddle_point_for_quadratic} we derive the asymptotics for the polynomials $p^{(n)}_n$ for quadratic $V$ using an
integral representation and the classical steepest descent method. In this derivation we show that the transformation $\Jlike$ also arises in a more direct way in the equispaced external source model.

\medskip

The main novel feature of this paper is the successful asymptotic analysis of the non-standard RH problem which characterizes the multiple orthogonal polynomials. Although the resulting large $n$ asymptotics for the polynomials resemble those for usual orthogonal polynomials relevant in the one-matrix model without external source, the RH method used to obtain those asymptotics had to be modified in a nontrivial way. We feel that the modification of the RH method, with in particular the use of the transformation $\Jlike$, is the main contribution of the present paper. We believe it is the first time that a RH analysis has been carried through for multiple orthogonal polynomials with a growing number of orthogonality weights.

\section{Proof of Theorem \ref{theorem: exist uni}}
\label{section: exist uni} Following \cite[Section 6.2]{Deift99}
(see also \cite{Johansson98}), one can prove the existence of a unique
Borel probability measure minimizing the energy $I_V(\mu)$ given in \eqref{energy}, which
can conveniently be written as
\begin{equation}\label{energy2}
 I_V(\mu)=\iint k^V(t,s)d\mu(t)d\mu(s),\end{equation}
 with
 \begin{equation}\label{kV}
 k^V(t,s)=\frac{1}{2}\log|t-s|^{-1}+\frac{1}{2}\log|e^t-e^s|^{-1}+\frac{1}{2}V(t)+\frac{1}{2}V(s).
\end{equation}
From the inequality $|v-u|\leq \sqrt{1+v^2}\sqrt{1+u^2}$ for
$v,u\in\mathbb R$, we obtain
 \begin{equation}
\frac{1}{2}\log|t-s|^{-1}+\frac{1}{2}\log|e^t-e^s|^{-1}\geq
-\frac{1}{4}\log(1+t^2)-\frac{1}{4}\log(1+s^2)-\frac{1}{4}\log(1+e^{2t})-\frac{1}{4}\log(1+e^{2s}).
\end{equation}
If $V$ satisfies the growth condition (\ref{Vinfty}), it easily
follows that there exists a constant $c_V$ such that $k^V(t,s)\geq
c_V$ for all $s,t\in\mathbb R$. Thus $I_V(\mu)\geq c_V$ for any
probability measure $\mu$, which implies that
$E_V=\inf\{I_V(\mu)\}\geq c_V$, where the infimum is taken over all
probability measures on $\mathbb R$. This is the crucial estimate
for proving the existence of a unique equilibrium measure. The
existence follows, exactly as in \cite[Section 6.2]{Deift99}, from
the construction of a vaguely convergent tight sequence $\mu_n$ of
measures with limit $\mu$ such that $I_V(\mu)=E_V$, as well as the
fact that any minimizer must have compact support.

\medskip

The uniqueness is slightly more complicated, and we need the
following lemma for it:
\begin{lem}
Let $\mu$ be a finite signed measure on $\mathbb R$ such that $\int
d\mu=0$ and with compact support. Then
\begin{align}
&\iint\log|x-y|^{-1}d\mu(x)d\mu(y)\geq 0,\\
&\iint\log|e^x-e^y|^{-1}d\mu(x)d\mu(y)\geq 0.
\end{align}
\end{lem}
The first inequality was showed in \cite[Lemma 6.41]{Deift99}, and
the second part can be proved by replacing $x\mapsto e^x$ and
$y\mapsto e^y$ in the proof.

\medskip

Now assume that we have two measures $\mu_V$ and $\tilde\mu$ such
that $I_V(\mu_V)=I_V(\tilde\mu)=E_V$. Then, for
$\mu_t=\mu_V+t(\tilde\mu-\mu_V)$ and $t\in[0,1]$, we have
\begin{multline}I_V(\mu_t)=\frac{1}{2}I(\mu_V,\mu_V)+\frac{1}{2}\tilde I(\mu_V,\mu_V)+\int V(x)d\mu_V(x)\\
+t\left(I(\mu_V,\tilde\mu-\mu_V)+\tilde I(\mu_V,\tilde\mu-\mu_V)+ \int V(x)d(\tilde\mu-\mu_V)(x))\right)\\
+t^2\left(\frac{1}{2}I(\tilde\mu-\mu_V,\tilde\mu -\mu_V)+
\frac{1}{2}\tilde I(\tilde\mu-\mu_V,\tilde\mu -\mu_V)\right),
\end{multline}
where
\begin{align}&I(\mu,\nu)=\iint\log|x-y|^{-1}d\mu(x)d\nu(y),\\
&\tilde I(\mu,\nu)=\iint\log|e^x-e^y|^{-1}d\mu(x)d\nu(y).
\end{align} The above lemma ensures that $I_V(\mu_t)$ is a convex
function of $t$. But since $\mu_t$ is a probability measure, we have
$I_V(\mu_t)\geq I_V(\mu_0)=I_V(\mu_1)=E_V$, and hence
$I_V(\mu_t)=E_V$ for all $t\in[0,1]$. In particular this implies
\begin{equation}
  \frac{1}{2}I(\tilde\mu-\mu_V,\tilde\mu -\mu_V)+
\frac{1}{2}\tilde I(\tilde\mu-\mu_V,\tilde\mu -\mu_V)=0,
\end{equation}
and using
a similar argument as in \cite{Deift99}, this implies that
$\mu_V=\tilde\mu$, which yields the uniqueness of the equilibrium
measure.

\section{Construction of the equilibrium measure}\label{section:equi}

In this section we assume the external field $V$ is a convex real analytic function and $V''(x)$ is bounded below by a positive constant for all $x \in \realR$. We are going to show that $V$ is one-cut regular, by an explicit construction of its equilibrium measure. The strategy of our construction is as follows. First in Section \ref{subsec:support_of_equilibrium} we give the support of the equilibrium measure $[a,b]$ without proof. Then in Section \ref{subsec:g_fun_and_equ_measure} we compute the density of the equilibrium measure, based on the information of the support. The density function is expressed in terms of the so-called $\gfn$-functions $\gfn(z), \tilde{\gfn}(z)$ and their derivatives, which are characterized by a RH problem. At last in Section \ref{subsec:verification_of_eqiolibrium_measure} we verify that the measure with the support and the density obtained in the first two steps satisfy the criteria of one-cut regularity, and conclude that it is the unique equilibrium measure that we want to construct.

\begin{rmk}In what follows, it may seem that the values of the endpoints $a$ and $b$ appear out of the blue, but if the external field $V(x)$ is quadratic, the endpoints (as well as $\gfn(x)$ and $\tilde{\gfn}(x)$) can be computed by a classical steepest-descent method. This computation is shown in Appendix \ref{sec:saddle_point_for_quadratic} as our inspiration.
\end{rmk}
\begin{rmk}
  If an external field is non-convex but we know a priori that it is one-cut regular with support $[a, b]$, then the method in Section \ref{subsec:g_fun_and_equ_measure} can still be applied and allows us to obtain the expression of the density function of the equilibrium measure.
\end{rmk}

\subsection{The support of the equilibrium measure} \label{subsec:support_of_equilibrium}

Let $\Jlike_{x_1,x_0}$ be defined as before by
\begin{equation}\Jlike_{x_1, x_0}(s)=x_1 s + x_0 - \log \frac{s - \frac{1}{2}}{s + \frac{1}{2}},\label{Jx0x1}\end{equation} and let
$\gamma=\Jlike_{x_1,x_0}^{-1}([a,b])$, depending on $x_1,x_0$, be the boundary of the region $D$ defined in the Introduction, consisting of the curves $\gamma_1$ and $\gamma_2$, encircling the interval $[-\frac{1}{2},\frac{1}{2}]$ in the counterclockwise direction, see also Proposition \ref{prop:Joukowsky_like} below.
Since $\Jlike_{x_1,x_0}(s)\in[a,b]$ for $s\in\gamma$, $V'(\Jlike_{c_1,c_0}(s))$ is well defined for $s$ in a neighborhood of the curve $\gamma$, if $V$ is real analytic.
\begin{lem} \label{lem:c_0_and_c_1}
  Given any strongly convex real analytic function $V$, i.e.\ such that $V''(x)\geq c>0$ for all $x\in\mathbb R$, the system of equations with unknowns $x_0$ and $x_1$
  \begin{align}
    x^{-1}_1= {}& \frac{1}{2\pi i}\oint_\gamma V'(\Jlike_{x_1,x_0}(s))ds, \label{intro mod eq 1} \\
    1  = {}& \frac{1}{2\pi i}\oint_\gamma \frac{V'(\Jlike_{x_1,x_0}(s))}{s-\frac{1}{2}}ds, \label{intro mod eq 2}
  \end{align}
  has a solution $x_0 = c_0 \in \realR$ and $x_1 = c_1 \in \realR^+$.
\end{lem}
We will prove Lemma \ref{lem:c_0_and_c_1} in Appendix A. Based on this lemma, we construct the support, and furthermore the density function, of the equilibrium measure. We do not prove the uniqueness of the solution of equations \eqref{intro mod eq 1} and \eqref{intro mod eq 2}, for this uniqueness is a consequence of the uniqueness of the equilibrium measure by Theorem \ref{theorem: exist uni}, as from different solutions we construct different equilibrium measures.

 Here and later we take the value of the parameters $c_0$ and $c_1$ as the pair of solutions of \eqref{intro mod eq 1} and \eqref{intro mod eq 2}. Then we claim that $a$ and $b$, the two edges of the support of the equilibrium measure, are given by
\begin{equation} \label{eq:defining_formula_of_a_b}
    a=\Jlike_{c_1, c_0}(s_a), \quad b=\Jlike_{c_1, c_0}(s_b),
\end{equation}
where $s_a = -\sqrt{\frac{1}{4}+\frac{1}{c_1}}$, $s_b = \sqrt{\frac{1}{4}+\frac{1}{c_1}}$.
Then it is easy to verify that equations (\ref{eq:first_def_of_c_0})--(\ref{eq:first_def_of_c_1}) are satisfied.

\subsection{The $\gfn$-functions and the density function of the equilibrium measure} \label{subsec:g_fun_and_equ_measure}

Under the assumption that the external field $V$ is one-cut regular, with equilibrium measure $d\mu_V(x)=\psi_V(x)dx$ supported on $[a,b]$ as we claimed in \eqref{eq:defining_formula_of_a_b}, we construct two functions $\gfn(z)=\int \log(z-x)d\mu_V(x)$ and $\tilde \gfn(z)=\int \log(e^z-e^x)d\mu_V(x)$ as in \eqref{eq:expr_of_gfn_and_tilde_gfn-intro}. To describe the domain of the function $e^{\tilde{\gfn}(z)}$, we introduce the notation of the cylinder $\cyld$ which is formed by identifying the two edges of the strip $\strip$. If a function $f(z)$ is defined for $z \in \strip$, the limits $f(x \pm \pi i) = \lim_{z \to x \pm \pi i, z \in \strip} f(z)$ exist point-wise, and furthermore $f(x + \pi i) = f(x - \pi i)$, we say $f$ is defined on $\cyld$. The properties \ref{enu:equ_measure_intro_1}--\ref{enu:equ_measure_intro_5} in the Introduction satisfied by $\mu_V$ are then translated into properties satisfied by $\gfn$ and $\tilde\gfn$ as follows.
\begin{enumerate}[label=(\roman{*})]
\item \label{enu:constr_g_1}
  For $x \in (-\infty, a)$,
  \begin{equation} \label{eq:relation_of_g_tilde_g_<a}
    \gfn_+(x) = \gfn_-(x) + 2\pi i, \quad \tilde{\gfn}_+(x) = \tilde{\gfn}_-(x) + 2\pi i,
  \end{equation}
  and then $e^{\gfn(z)}$ is analytic in $\mathbb C\setminus[a,b]$ and $e^{\tilde \gfn(z)}$ is analytic on the cylinder with slit $\cyld \setminus [a,b]$; $e^{\gfn(z)} \sim z$ as $z \to \infty$, $e^{\tilde{\gfn}(z)} \sim e^z$ as $\Re z \to +\infty$ and $e^{\tilde{\gfn}(z)} = \bigO(1)$ as $\Re z \to -\infty$,
\item \label{enu:constr_g_2}
  for $x\in(a,b)$, we have
\begin{equation}\label{psi g}
\psi_V(x)=-\frac{1}{2\pi i}(\gfn_+'(x)-\gfn_-'(x))=-\frac{1}{2\pi
i}(\tilde \gfn_+'(x)-\tilde \gfn_-'(x))>0,
\end{equation}
\item \label{enu:constr_g_3}
  as $z \to a$ or $z \to b$, the limits of $\gfn(z)$, $\tilde{\gfn}(z)$, $\gfn'(z)$ and $\tilde{\gfn}'(z)$ exist, and as $x\to a_+$ or $x \to b_-$ for $x \in (a, b)$,
\begin{equation} \label{square root}
  \lim_{x \to a_+}\frac{\gfn_+'(x)-\gfn_-'(x)}{\sqrt{x-a}}, \quad \lim_{x \to a_+}\frac{\tilde \gfn_+'(x)-\tilde \gfn_-'(x)}{\sqrt{x-a}}, \quad \lim_{x \to b_-}\frac{\gfn_+'(x)-\gfn_-'(x)}{\sqrt{b - x}}, \quad \lim_{x \to b_-}\frac{\tilde \gfn_+'(x)-\tilde \gfn_-'(x)}{\sqrt{b - x}}
\end{equation}
all exist and are all different from zero,
\item \label{enu:constr_g_4}
 for $x\in[a,b]$, there exists a constant $\ell$ such that
 \begin{equation}\label{var eq g}
 \gfn_\pm(x)+\tilde \gfn_\mp(x)-V(x)-\ell=0,
 \end{equation}
 \item \label{enu:constr_g_5}
   for $x\in\mathbb R\setminus [a,b]$, we have
 \begin{equation}\label{var ineq g}
 \gfn_\pm(x)+\tilde \gfn_\mp(x)-V(x)-\ell<0.
 \end{equation}
\end{enumerate}

Let us consider the derivatives
\begin{equation} \label{eq:defn_of_G_and_tilde_G}
    G(x) \colonequals \gfn'(x) \quad \text{and} \quad \tilde{G}(x) \colonequals \tilde{\gfn}'(x).
\end{equation}
The properties \ref{enu:constr_g_1}, \ref{enu:constr_g_3} and \ref{enu:constr_g_4} for $\gfn(x)$ and $\tilde{\gfn}(x)$ then imply that $G$
and $\tilde G$ need to satisfy the following RH problem:

\subsubsection*{RH problem for $G$ and $\tilde G$}

\begin{enumerate}[label=(\alph{*})]
\item \label{enu:RHP_G_1} $G$ is analytic in $\mathbb C\setminus
[a,b]$, $\tilde G$ is analytic in $\cyld \setminus[a,b]$,
\item \label{enu:RHP_G_2} for $x\in[a,b]$, we have
\begin{equation}\label{jump G}
 G_\pm(x)+\tilde G_\mp(x)-V'(x)=0,
 \end{equation}
\item \label{enu:RHP_G_3} we have the asymptotic conditions that $G_{\pm}(x)$ and $\tilde{G}_{\pm}(x)$ are bounded for all $x \in [a, b]$, and
 \begin{align}
&\label{as G}G(z)=\frac{1}{z}+\bigO(z^{-2}),&&\mbox{ as $z\to\infty$,}\\
&\label{as tilde G1}\tilde
G(z)=1+\bigO(e^{-z}),&&\mbox{ as $\Re z\to +\infty$},\\
&\label{as tilde G2}\tilde G(z)=\bigO(1),&&\mbox{ as $\Re z\to
-\infty$}.
 \end{align}
\end{enumerate}

The main technical difficulty in solving the RH problem for $G$ and
$\tilde G$ lies in the fact that the two functions live on different
domains: $G$ is defined in the complex plane with slit $[a,b]$, and
$\tilde G$ is defined in the cylinder $\cyld$ with slit $[a,b]$. In
order to resolve this problem, we will use the transformation \eqref{eq:Joukowsky_like_transform} $\Jlike(s)$ that maps $\compC \setminus [-\frac{1}{2}, \frac{1}{2}]$ to both $\compC$ and $\strip$. Recall that $s_a$ and $s_b$ are the two critical points of $\Jlike(s)$ given by \eqref{eq:defn_of_s_a_and_s_b} and that they satisfy the identity \eqref{eq:defining_formula_of_a_b}. The following property will be used in the construction of $G$ and $\tilde{G}$.
\begin{prop} \label{prop:Joukowsky_like}
  There are an arc $\gamma_1$ from $s_a$ to $s_b$ in the upper half plane, and an arc $\gamma_2$ from $s_a$ to $s_b$
  in the lower half plane, such that
  \begin{enumerate}[label=(\alph*)]
  \item \label{enu:prop:Joukowsky_like:a}
    $\Jlike(\gamma_1) = \Jlike(\gamma_2) = [a,b]$, and the mapping is homeomorphic on these two curves.
  \item \label{enu:prop:Joukowsky_like:b}
    Denote the region enclosed by $\gamma_1$ and $\gamma_2$ by $D$.
    Then $\Jlike(\compC \setminus \bar{D}) = \compC \setminus [a,b]$, and the mapping is univalent.
  \item \label{enu:prop:Joukowsky_like:c}
    $\Jlike(D \setminus [-\frac{1}{2},\frac{1}{2}]) = \strip \setminus [a,b]$, the mapping is univalent, and the upper and lower sides of $(-\frac{1}{2},\frac{1}{2})$ are mapped to $\realR - \pi i$ and $\realR + \pi i$ respectively.
  \end{enumerate}
\end{prop}

Let us now define the function $M(s)$ by
\begin{equation}\label{def M}
  M(s)  \colonequals
  \begin{cases}
    G(\Jlike(s)) & \textnormal{for $s \in \compC \setminus \bar{D}$,} \\
    \tilde{G}(\Jlike(s)) & \textnormal{for $s \in D \setminus [-\frac{1}{2},\frac{1}{2}]$,}
  \end{cases}
\end{equation}
so that $M$ is analytic in $\compC \setminus (\gamma_1 \cup \gamma_2 \cup [-\frac{1}{2}, \frac{1}{2}])$. Note that the domain of $\tilde{G}$ can be extended from $\strip$ to $\cyld$, so that $M(s)$ can be analytically continued to
$(-\frac{1}{2},\frac{1}{2})$ accordingly. The RH conditions for $G, \tilde G$
are now transformed to the following conditions for $M$.
\subsubsection*{RH problem for $M$}
\begin{enumerate}[label=(\alph{*})]
\item \label{enu:RHP_M_1} $M$ is analytic in $\mathbb C\setminus (\gamma_1 \cup \gamma_2 \cup \{ -\frac{1}{2}, \frac{1}{2}\})$,
\item \label{enu:RHP_M_2} $M$ satisfies the jump condition
\begin{equation}
  M_+(s) + M_-(s) = V' \left(\Jlike(s) \right), \quad
  \textnormal{for $s \in \gamma_1\cup\gamma_2$,} \label{eq:RHP_for_M(s)_easy_1}
\end{equation}
\item \label{enu:RHP_M_3} $M_{\pm}(s)$ is bounded on $\gamma_1$ and $\gamma_2$, and $M$ has the asymptotics
\begin{align}
  M(s) = {}& \frac{1}{c_1 s} + \bigO(s^{-2}), & & \textnormal{as $s \to \infty$,} \label{eq:outer_boundary_condition} \\
  M(s) = {}& 1 + \bigO(s-\frac{1}{2}), & & \textnormal{as $s \to \frac{1}{2}$,} \label{eq:inner_boundary_condition_at_1}
  \\
  M(s) = {}& \bigO(1), & & \textnormal{as $s \to -\frac{1}{2}$.} \label{eq:inner_boundary_condition_at_0}
\end{align}
\end{enumerate}
It is straightforward to solve this scalar RH problem. We write
\begin{equation}
  U(s)=V'(\Jlike(s)),
\end{equation}
and note that $U$ is analytic in a neighborhood of $\gamma_1 \cup \gamma_2$, since $V$ is real analytic. Then it is readily verified that the
unique solution $M$ to the above RH problem for $M$ is given by
\begin{equation}\label{formula M}
  M(s)=
  \begin{cases}
    {\displaystyle -\frac{1}{2\pi i}\oint_{\gamma} \frac{U(\xi)}{\xi -s}d\xi}, & \text{for $s\in\mathbb C\setminus D$},\\
    {\displaystyle \frac{1}{2\pi i}\oint_{\gamma} \frac{U(\xi)}{\xi - s}d\xi}, & \text{for $s\in D$},
\end{cases}
\end{equation}
where $\gamma$ is the closed curve which is the union of $\gamma_1$ and $\gamma_2$ and has counterclockwise orientation.
In particular, \eqref{eq:outer_boundary_condition} and
\eqref{eq:inner_boundary_condition_at_1} follow from the system
of equations \eqref{intro mod eq 1} and \eqref{intro mod eq 2} in Lemma \ref{lem:c_0_and_c_1} satisfied
by $c_0, c_1$.

Now we can give an expression for $\gfn(z)$, $\tilde{\gfn}(z)$ and the density function $\psi_V(x)$ of the equilibrium measure, under the assumption that the support of the equilibrium measure is known. Recall that $\Jinv_1$ is the inverse map of
$\Jlike$ from $\mathbb C\setminus[a,b]$ to $\mathbb C\setminus
\overline D$,  $\Jinv_2$ is the inverse map of $\Jlike$ from
$\mathbb \strip\setminus[a,b]$ to $D\setminus [-\frac{1}{2},\frac{1}{2}]$, and their boundary values define $\Jinv_\pm(x)$, see \eqref{eq:defn_of_I_1}--\eqref{eq:defn_of_I_-}. We have $\Jinv_+(x) \in \gamma_1$, $\Jinv_-(x) \in
\gamma_2$, and $\Jinv_-(x) = \overline{\Jinv_+(x)}$. To obtain a
formula for the density $\psi_V(x)$ of the equilibrium measure, note
that it follows from (\ref{psi g}) and the identities $G=\gfn', \tilde{G} = \tilde{\gfn}'$ that
\begin{equation}\label{psizeta}
\psi_V(x)=-\frac{1}{2\pi i}(G_+(x)-G_-(x)) = -\frac{1}{2\pi i}(\tilde{G}_+(x) - \tilde{G}_-(x)),\quad\text{for
$x\in[a,b]$.}
\end{equation}
 From \eqref{psizeta} and \eqref{def M}, we
obtain
\begin{equation}\label{psiM}
  \psi_V(x) = -\frac{1}{2\pi i}(M_+(\Jinv_+(x))-M_-(\Jinv_-(x))) = -\frac{1}{2\pi i}(M_+(\Jinv_-(x))-M_-(\Jinv_+(x))),\quad\text{for $x\in[a,b]$,}
\end{equation}
where the boundary values of $M$ correspond to the orientations of $\gamma_1$ and $\gamma_2$, from left to right.
Applying the first identity in \eqref{psiM} and the formula \eqref{formula M} for $M(s)$, we let $z=x+i\epsilon$, $\epsilon>0$, approach $x$ from above and have
\begin{equation} \label{eq:explicit_expr_of_psi_V}
  \begin{split}
    \psi_V(x) = {}& \lim_{\epsilon\to 0} \frac{-1}{4\pi^2} \oint_\gamma U(\xi)\left(\frac{1}{\xi-\Jinv_1(z)}-\frac{1}{\xi-\overline{\Jinv_1(z)}}\right)d\xi \\
    = {}& \lim_{\epsilon\to 0} \frac{1}{4\pi^2} \int_a^b V'(u) \left( \frac{\Jinv_+'(u)}{\Jinv_+(u)-\Jinv_1(z)}-\frac{\Jinv_+'(u)}{\Jinv_+(u) - \overline{\Jinv_1(z)}} - \frac{\Jinv_-'(u)}{\Jinv_-(u)-\Jinv_1(z)} +\frac{\Jinv_-'(u)}{\Jinv_-(u)- \overline{\Jinv_1(z)}}\right)d u \\
    = {}& \lim_{\epsilon\to 0} \frac{1}{2\pi^2}\int_a^b V'(u) \Re\left( \frac{\Jinv_+'(u)}{\Jinv_+(u)-\Jinv_1(z)}-\frac{\Jinv_+'(u)}{\Jinv_+(u) - \overline{\Jinv_1(z)}} \right) d u \\
    = {}& \lim_{\epsilon\to 0} \frac{-1}{2\pi^2}\int_a^b V'(u) \Re \frac{d}{du} \log \left(\frac{\Jinv_+(u) - \overline{\Jinv_1(z)}}{\Jinv_+(u)-\Jinv_1(z)}\right) d u \\
    = {}& \lim_{\epsilon\to 0} \frac{1}{2\pi^2}\int_a^b V''(u) \Re \log \left(\frac{\Jinv_+(u) - \overline{\Jinv_1(z)}}{\Jinv_+(u)-\Jinv_1(z)}\right) d u \\
    = {}& \frac{1}{2\pi^2}\int_a^b V''(u)\log \left|\frac{\Jinv_+(u) - \Jinv_-(x)}{\Jinv_+(u)-\Jinv_+(x)}\right| d u.
  \end{split}
\end{equation}

\subsection{Proof of Theorem \ref{theorem: convex}} \label{subsec:verification_of_eqiolibrium_measure}

We showed so far that the equilibrium measure associated to the external field $V$ has the density function $\psi_V$ as we have constructed in Section \ref{subsec:g_fun_and_equ_measure}, as long as it is supported on the single interval $[a, b]$ that is given by \eqref{eq:defining_formula_of_a_b}. However, we have not proved that $[a, b]$ is the correct support yet.
We will show that the measure with support $[a,b]$ and density function $\psi_V(x)$ satisfies the properties \ref{enu:equ_measure_intro_1}--\ref{enu:equ_measure_intro_5} stated in the Introduction for one-cut regular equilibrium measures, which implies that the constructed measure is indeed the true equilibrium measure. Note that these properties are equivalent to properties \ref{enu:constr_g_1}--\ref{enu:constr_g_5} in Section \ref{subsec:g_fun_and_equ_measure}.

From the construction of $\psi_V(x)$, it is normalized, \ie, $\int^b_a \psi_V(x) dx = 1$. This follows from the asymptotics of $G$ and $\tilde{G}$, given in \eqref{as G} and \eqref{as tilde G1}, and the definitions of $\gfn$ and $\tilde{\gfn}$, the antiderivatives of $G$ and $\tilde{G}$.

For $x \in (a,b)$, it is geometrically obvious that $\lvert \Jinv_+(u) - \Jinv_-(x) \rvert > \lvert \Jinv_+(u) - \Jinv_+(x) \rvert$, and then $\Re \log ((\Jinv_+(u) - \Jinv_-(x))/(\Jinv_+(u)-\Jinv_+(x))) > 0$ for all $u \in (a, b)$. Substituting this inequality into \eqref{eq:explicit_expr_of_psi_V} and noting that $V''$ is positive, we have that $\psi_V(x) > 0$ for all $x \in (a, b)$. Similarly we have $\psi_V(x) \to 0$ for $x \to a_+$ and $x \to b_-$.

The identity \eqref{var eq} that gives condition \ref{enu:equ_measure_intro_4} in the Introduction, or equivalently the identity \eqref{var eq g} that gives condition \ref{enu:constr_g_4} in Section \ref{subsec:g_fun_and_equ_measure}, is obvious from the construction of $\psi_V$. Thus we only need to prove the remaining two properties for the equilibrium measure hold, \ie, $\psi_V(x)$ vanishes like a square root as $x \to a_+$ or $x \to b_-$, and $G_+(x) + \tilde{G}_-(x) - V(x) < \ell$ for $x < a$ or $x > b$.

Let the function $H$ be defined by
\begin{equation} \label{eq:defn_of_H}
H(z)=\left(G(z)+\tilde G(z)-V'(z)\right)^2.
\end{equation}
It is well defined where $G, \tilde{G}, V$ are defined, and it can only be discontinuous on $[a, b]$. However, by \eqref{psi g} and \eqref{var eq g},
\begin{equation} \label{eq:H_in_(a,b)}
  H_+(x) = (\tilde{G}_+ - \tilde{G}_-)^2 = -4\pi^2 \psi_V(x)^2 = (G_- - G_+)^2 = H_-(x).
\end{equation}
Hence $H(z)$ can be defined on $(a, b)$ so that $a, b$ become isolated singularities. If we express $G(z)$ and $\tilde{G}(z)$ in terms of $M(s)$ and then by the contour integral as in \eqref{def M} and \eqref{formula M}, we find that $G(z)$ and $\tilde{G}(z)$ grows at most logarithmically at $a$ and $b$. Thus $a$ and $b$ are removable singularities of $H(z)$, and $H(z)$ can be defined analytically in $\strip$ where $V$ is defined, \ie, an open region containing the real line. Furthermore, by \eqref{eq:H_in_(a,b)} and the fact that $\psi_V(x) \to 0$ as $x \to a_+$ or $x \to b_-$, we have that $H(a) = H(b) = 0$.

To show that $\psi_V(x)$ vanishes like a square root at $a$ and $b$, by \eqref{eq:H_in_(a,b)} it suffices to show that $a, b$ are simple zeros of $H(z)$. We consider $b$ first. From \eqref{eq:H_in_(a,b)} and \eqref{eq:defn_of_H}, we see that $H(x)$ changes sign as the real variable $x$ increases around $b$, so if $b$ is not a simple zero, it has multiplicity at least $3$, and then $\frac{d}{dx}\sqrt{H(x)}$, which is well defined for $x \in (b, \infty)$, would tend to $0$ as $x \to b_+$. But we have for all $x > b$
\begin{equation} \label{eq:negative_derivative_sqare_root_H}
  \begin{split}
     \frac{d}{dx} \left( G(x) + \tilde{G}(x) - V'(x) \right) = {}& \gfn''(x) + \tilde{\gfn}''(x) - V''(x) \\
    = {}& -\int^b_a \psi_V(s) \left( \frac{1}{(x - s)^2} + \frac{e^x e^s}{(e^x - e^s)^2} \right) ds - V''(x) < -V''(x).
  \end{split}
\end{equation}
Since $V''(x)$ is bounded below by a positive constant, $\frac{d}{dx} \sqrt{H(x)}$ cannot approach $0$. Thus $b$ is a simple zero of $H(z)$. Similarly $a$ is a simple zero.

To show that $G_+(x) + \tilde{G}_-(x) - V(x) < \ell$ for $x > b$, we need only that $G_+(x) + \tilde{G}_-(x) - V(x)$ is decreasing, since at $x = b$ the identity $G_+(x) + \tilde{G}_-(x) - V(x) = \ell$ holds. The decreasing property is given by the negative derivative shown in \eqref{eq:negative_derivative_sqare_root_H}. Similarly we can show that $G_+(x) + \tilde{G}_-(x) - V(x) < \ell$ for $x < a$.

Now we have proved that the measure $\psi_V(x)$ on $[a, b]$ satisfies all the properties for one-cut regular equilibrium measures, so it is the unique equilibrium measure associated to $V$. Combining the results we have obtained in this section, we prove Theorem \ref{theorem: convex}.

\section{Asymptotic analysis for the type II multiple orthogonal polynomials}\label{section: RH p}

In this section, we write $p_j^{(n)}(x)$, the monic multiple orthogonal polynomials of type II satisfying orthogonality relations \eqref{orthoIIb}, as $p_j(x)$ if there is no confusion.

\subsection{RH problem characterizing the polynomials} \label{subsec:RH_of_Y}

Recall that the $j$-th degree monic polynomial $p_j(x) = p^{(n)}_j(x)$ is characterized by the orthogonality \eqref{orthoIIb}.
Consider the following modified Cauchy transform of $p_j$:
\begin{equation} \label{eq:defn_of_Clike_p_j}
  \Clike p_j(z)  \colonequals  \frac{1}{2\pi i} \int_{\realR} \frac{p_j(x)}{e^x - e^z} e^{-nV(x)} dx,
\end{equation}
which is well-defined for $z\in\strip\setminus\mathbb R$. Since $e^{-nV(x)}$ is real analytic and vanishes rapidly as $x \to \pm \infty$, for any polynomial $p(x)$, we have the following asymptotic expansion for $\Clike p(z)$ as $z \in \strip$ and $\Re z \to +\infty$:
\begin{equation} \label{eq:series_expansion_of_Cauchy_like}
  \begin{split}
    \Clike p(z) = {}& \frac{-1}{2\pi i e^z} \int_{\realR} \frac{p(x)}{1 - e^x/e^z} e^{-nV(x)} dx \\
    = {}& \frac{-1}{2\pi i} \sum^{M}_{k = 0} \left( \int_{\realR} p(x)e^{kx} e^{-nV(x)} dx \right) e^{-(k+1)z}+\bigO(e^{-(M+2)z}),
  \end{split}
\end{equation}
for any $M\in\mathbb N$, uniformly in $\Im z$. Thus due to the orthogonality,
\begin{equation} \label{eq:kappa_is_leading_coeff_of_Clike_p}
  \Clike p_j(z) = \frac{-h^{(n)}_j}{2\pi i} e^{-(j+1)z} +\bigO(e^{-(j+2)z}),
\end{equation}
where $h^{(n)}_j$ is given by \eqref{kappa}. For $x \in \realR$, a residue argument shows that
\begin{equation}
    (\Clike p_j)_+(x) - (\Clike p_j)_-(x)
    =  p_j(x)e^{-nV(x)}e^{-x}.
\end{equation}
Hence we conclude that if we consider $p_j(x)$ and $\Clike p_j(x)$ together and write them in vector form
\begin{equation}\label{def Y}
  Y(z) =Y^{(j,n)}(z) \colonequals  (p_j(z), \Clike p_j(z)),
\end{equation}
they satisfy the conditions

\subsubsection*{RH problem for $Y$}
\begin{enumerate}[label=(\alph{*})]
\item \label{enu:RH_of_Y:1}
  $Y = (Y_1, Y_2)$, where $Y_1$ is an analytic function defined on $\mathbb C$, and $Y_2$ is an analytic function on $\cyld \setminus \mathbb R$,
\item \label{enu:RH_of_Y:2}
  $Y$ has continuous boundary values $Y_\pm$ when
  approaching the real line from above and below, and we have
  \begin{equation}\label{RHP_Y:jump1}
    Y_+(x) = Y_-(x)
    \begin{pmatrix}
      1 & e^{-x}e^{-nV(x)} \\
      0 & 1
    \end{pmatrix},
    \quad \textnormal{for $x \in \realR$,}
  \end{equation}
  \addtocounter{enumi}{1}
  \begin{enumerate}[label=(\alph{enumi}\arabic{*}), leftmargin=0pt, itemindent=0pt]
  \item \label{enu2:c1_of_Y}
    as $z\to\infty$, $Y_1$ behaves as  $Y_1(z)=z^j+\bigO(z^{j-1})$,
  \item \label{enu2:c2_of_Y}
    as $e^z\to\infty$ (\ie, $\Re z\to +\infty$), $Y_2$ behaves as
    $Y_2(z)=\bigO(e^{-(j+1)z})$; as $e^z\to 0$ (\ie, $\Re z\to -\infty$),
    $Y_2(z)$ remains bounded.
  \end{enumerate}
\end{enumerate}

Conversely, the RH problem for $Y$ has a unique solution given by \eqref{def Y}. We give a proof of the uniqueness of the RH problem for $Y$ based on the uniqueness of the multiple orthogonal polynomials $p_j$.

\begin{thm}
  The solution to the RH problem for $Y$ above has a unique solution, given by $Y_1(z) = p_j(z)$ and $Y_2(z) = \Clike p_j(z)$, where $p_j(z)$ is the monic multiple orthogonal polynomial of type II defined by \eqref{orthoIIb}, and $\Clike p_j(z)$ is given in \eqref{eq:defn_of_Clike_p_j}.
\end{thm}

\begin{proof}
  First, \eqref{RHP_Y:jump1} in the jump condition \ref{enu:RH_of_Y:2} implies that $Y_1$ is an entire function, and condition \ref{enu2:c1_of_Y} implies that $Y_1$ grows like $z^j$ as $z \to \infty$. So $Y_1=:p$ is a monic polynomial of degree $j$.

  Now we show that if $Y = (Y_1, Y_2)$ satisfies all the conditions \ref{enu:RH_of_Y:1}--\ref{enu2:c2_of_Y} of the RH problem, then $Y_2$ is given in terms of $Y_1=p$ by
  \begin{equation} \label{eq:from_Y_1_to_Y_2}
    Y_2(z) = \frac{1}{2\pi i}\int_{\mathbb R}\frac{p(s)}{e^s-e^z}e^{-nV(s)}ds.
  \end{equation}
  By condition \ref{enu:RH_of_Y:2}, $Y_2$ satisfies
  \begin{equation}
    Y_{2,+}(x)-Y_{2,-}(x)=p(x)e^{-nV(x)-x}.
  \end{equation}
Consider the function
\begin{equation}
  U(u)=Y_2(\log u)-\frac{1}{2\pi i}\int_{\mathbb R}\frac{Y_1(s)}{e^s-u}e^{-nV(s)}ds,
\end{equation}
where we take the principal branch of the logarithm with branch cut on $\mathbb R^-$. Obviously $U(u)$ is analytic for $u \in \compC \setminus \realR$. By the jump condition $Y_2$ on the real line given by \eqref{RHP_Y:jump1} and the property that $Y_2(x + \pi i) = Y_2(x - \pi i)$, we verify that $U_+(u)=U_-(u)$ for $u \in (0, \infty)$ or $u \in (-\infty, 0)$, so that $U$ is an analytic function for $u \in \compC \setminus \{ 0 \}$. Note that since $p$ is a polynomial and $e^{-nV(s)}$ vanishes rapidly as $s \to \pm\infty$, we have
\begin{align}
\frac{1}{2\pi i}\int_{\mathbb R}\frac{p(s)}{e^s-u}e^{-nV(s)}ds = {}& \bigO(1) & & \text{as $u \to 0$,} \label{eq:U(u)_at_0} \\
\frac{1}{2\pi i}\int_{\mathbb R}\frac{p(s)}{e^s-u}e^{-nV(s)}ds = {}& \bigO(u^{-1}) & & \text{as $u \to \infty$.} \label{eq:U(u)_at_infty}
\end{align}
From \eqref{eq:U(u)_at_0} we find that $0$ is a removable singularity of $U(u)$ and then $U(u)$ is an entire function. Then from \eqref{eq:U(u)_at_infty} we have $U(u) = 0$ by Liouville's theorem. Therefore \eqref{eq:from_Y_1_to_Y_2} is proved.

At last we apply the expansion \eqref{eq:series_expansion_of_Cauchy_like} for $M=j-1$ to $Y_2$ given in \eqref{eq:from_Y_1_to_Y_2}. We see that the asymptotic condition $Y_2 = \bigO(e^{-(j+1)z})$ implies that
\begin{equation}
  \int_{\realR} p(x)e^{kx} e^{-nV(x)} dx = 0, \quad k = 0, \dotsc, j-1.
\end{equation}
Comparing this with \eqref{orthoIIb}, we see that $p=Y_1$ is indeed the monic multiple orthogonal polynomial $p_j$.
\end{proof}

Below we take $j = n+k$ where $k$ a constant integer, and our goal is to obtain the asymptotics for $Y = Y^{(n+k, n)}$ as $n \to \infty$.

\subsection{First transformation $Y\mapsto T$} \label{subsec:RH_of_T}

Recall $\gfn(z)$ and $\tilde{\gfn}(z)$ defined in \eqref{eq:expr_of_gfn_and_tilde_gfn-intro} on $\compC \setminus (-\infty,
b]$ and $\strip \setminus (-\infty, b]$. Denote $Y=Y^{(n+k,n)}$ and define $T$ as follows:
\begin{equation}\label{def T}
  T(z)  \colonequals
    e^{-\frac{n\ell}{2}} Y(z)\begin{pmatrix}e^{-n\gfn(z)}&0\\0&e^{n\tilde\gfn(z)}\end{pmatrix}
    e^{\frac{n\ell}{2}\sigma_3},
\end{equation}
where $\ell$ is the constant appearing in \eqref{var eq} and \eqref{var eq g}, and $\sigma_3=\left( \begin{smallmatrix}1&0\\0&-1\end{smallmatrix} \right)$. Then $T$ satisfies a RH problem with the same domain of analyticity as $Y$, but with a different asymptotic behavior and a different jump relation.

\subsubsection*{RH problem for $T$}
\begin{enumerate}[label=(\alph{*})]
\item $T = (T_1, T_2)$, where $T_1$ is analytic in $\mathbb C\setminus\mathbb R$, and $T_2$ is analytic in $\cyld \setminus \mathbb R$,
\item $T$ satisfies the jump relation
\begin{equation}\label{RHP_T:jump}
  T_+(x) = T_-(x)J_T(x),
  \quad \textnormal{for $x \in \realR$,}
\end{equation}
with
\begin{equation} \label{JT}
  J_T(x)=
  \begin{pmatrix}
    e^{n(\gfn_-(x) - \gfn_+(x))} & e^{n(\gfn_-(x) + \tilde{\gfn}_+(x) - V(z) - \ell) - x} \\
    0 & e^{n(\tilde{\gfn}_+(x) - \tilde{\gfn}_-(x))}
  \end{pmatrix},
\end{equation}
\addtocounter{enumi}{1}
\begin{enumerate}[label=(\alph{enumi}\arabic{*}), leftmargin=0pt, itemindent=0pt]
\item as $z\to\infty$, $T_1$ behaves as
$T_1(z)= z^{k}+\bigO(z^{k-1})$,
\item as $e^z\to\infty$, $T_2$ behaves as  $T_2(z)=\bigO(e^{-(k+1)z})$, and as $e^z \to 0$, $T_2$ behaves as $T_2 = \bigO(1)$.
\end{enumerate}
\end{enumerate}

\subsection{Second transformation $T\mapsto S$} \label{subsec:RH_of_S}

 For $x \in\mathbb R\setminus[a,b]$, it follows from the analyticity of $e^{\gfn}$ and \eqref{var ineq g} that
the jump matrix $J_T(x)$ tends to the identity matrix exponentially
fast in the limit $n\to\infty$. For $x \in (a,b)$, we decompose
the jump matrix into
\begin{multline}\label{facto}
  J_T(x)
  =
  \begin{pmatrix}
    1 & 0 \\
    e^{-n\phi_-(x) + x} & 1
  \end{pmatrix}
  \begin{pmatrix}
    0 & e^{n(\gfn_-(x) + \tilde{\gfn}_+(x) - V(x) - \ell) - x} \\
    -e^{n(-\gfn_+(x) - \tilde{\gfn}_-(x) + V(x) + \ell) + x} & 0
  \end{pmatrix} \\
  \times
  \begin{pmatrix}
    1 & 0 \\
    e^{-n\phi_+(x) + x} & 1
  \end{pmatrix},
\end{multline}
where the function $\phi(z) = \gfn(z) + \tilde{\gfn}(z) - V(z) - \ell$ is defined as in \eqref{def phi}. The function $\phi(x)$ has discontinuity on $(-\infty, b])$, and by \eqref{eq:relation_of_g_tilde_g_<a} and \eqref{var eq g} it satisfies
\begin{align}
  \phi_+(x) = {}& \phi_-(x) + 4\pi i & & \text{for $x < a$,} \label{eq:phi_across_real_line_left} \\
  \phi_+(x) = {}& -\phi_-(x) & & \text{for $x \in (a, b)$.} \label{eq:phi_across_real_line_middle}
\end{align}

\begin{figure}[htb]
  \centering
  \includegraphics{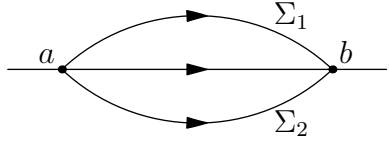}
  \caption{The lens $\Sigma_S$.}
  \label{fig:the_lens}
\end{figure}
Then we ``open the lens'', where the lens $\Sigma_S$ is a contour consisting of the real axis and two arcs from $a$ to $b$. We assume that one of the two arcs lies in the upper half plane and denote it by $\Sigma_1$, the other lies in the lower half plane and denote it by $\Sigma_2$, see Figure \ref{fig:the_lens}. We do not fix the shape of $\Sigma_S$ at this stage, but only require that $\Sigma_S$ is in $\strip$ and $V$ is analytic in a simply-connected region containing $\Sigma_S$.

Define
\begin{equation} \label{eq:def S}
  S(z)  \colonequals
  \begin{cases}
    T(z) & \textnormal{outside of  the lens,} \\
    T(z)
    \begin{pmatrix}
    1 & 0 \\
    e^{-n\phi(z) + z} & 1
  \end{pmatrix}
  & \textnormal{in the lower part of the lens,} \\
  T(z)
    \begin{pmatrix}
    1 & 0 \\
    -e^{-n\phi(z) + z} & 1
  \end{pmatrix}
  & \textnormal{in the upper part of the lens.}
  \end{cases}
\end{equation}
From the definition of $S$, we see that $S$ is discontinuous on the upper and lower arcs with jump matrix $\left( \begin{smallmatrix} 1 & 0 \\ e^{-n\phi(z) + z} & 1 \end{smallmatrix} \right)$. On $(a,b)$, it follows from \eqref{var eq g} and (\ref{facto}) that the jump matrix for $S$ takes the form $\left( \begin{smallmatrix} 0 & e^{-x} \\ -e^x & 0 \end{smallmatrix} \right)$. Summarizing, we have the following RH problem for $S$.

\subsubsection*{RH problem for $S$}
\begin{enumerate}[label=(\alph{*})]
\item $S= (S_1, S_2)$, where $S_1$ is analytic in $\mathbb C\setminus \Sigma_S$, and $S_2$ is analytic in $\cyld \setminus \Sigma_S$, and $\Sigma_S=\mathbb
R\cup\Sigma_1\cup\Sigma_2$ is the contour depicted in Figure \ref{fig:the_lens},
\item we have
\begin{equation}\label{RHP_S:jump}
  S_+(z) = S_-(z)J_S(z),\qquad\mbox{ for $z\in\Sigma_S$,}
  \end{equation}
  where (note that $e^{\phi(z)}$ is well defined for $z \in (-\infty, a)$ by \eqref{eq:phi_across_real_line_left})
  \begin{equation} \label{eq:defn_of_jump_matrix_J_S}
    J_S(z)=\begin{cases}
      \begin{pmatrix}
        1 & 0 \\
        e^{-n\phi(z) + z} & 1
      \end{pmatrix},&\textnormal{for $z\in\Sigma_1\cup\Sigma_2$,}\\
      \begin{pmatrix}
        0 & e^{-z} \\
        -e^z & 0
      \end{pmatrix},&\textnormal{for $z\in (a,b)$,}\\
      \begin{pmatrix}
        1 & e^{n\phi(z) - z} \\
        0 & 1
      \end{pmatrix},&\textnormal{for $z\in \mathbb R\setminus[a,b]$.}
    \end{cases}
  \end{equation}
  \addtocounter{enumi}{1}
  \begin{enumerate}[label=(\alph{enumi}\arabic{*}), leftmargin=0pt, itemindent=0pt]
  \item as $z\to\infty$, $S_1(z) = z^k + \bigO(z^{k-1})$,
  \item as $e^z\to\infty$, $S_2$ behaves as $S_2(z) = \bigO(e^{-(k+1)z})$, and as  $e^z \to 0$, $S_2$ behaves as $S_2(z) = \bigO(1)$.
  \end{enumerate}
\end{enumerate}

By (\ref{var eq g}), we have, for $x\in(a,b)$,
\begin{equation} \label{phi psi}
\phi_\pm'(x)=\gfn_\pm'(x) + \tilde{\gfn}_{\pm}'(x) - V'(x) = \gfn_\pm'(x)-\gfn_\mp'(x) = \mp 2\pi i\psi_V(x).
\end{equation}
Since $\psi_V(x) > 0$ for all $x \in (a, b)$, it follows from the Cauchy-Riemann conditions that
\begin{equation} \label{eq:inequality_for_phi}
\Re\phi(z)=\Re\left(\gfn(x) + \tilde{\gfn}(x) - V(x) -\ell\right)>0
\end{equation}
on both the upper arc and the lower arc, if these arcs are
chosen sufficiently close to $(a,b)$. As a consequence, the jump
matrices for $S$ on the lenses tend to the identity matrix as
$n\to\infty$. Uniform convergence breaks down when $x$ approaches
the endpoints $a$ and $b$. We need to construct local parametrices
near those points.

\subsection{Construction of local parametrices near $a$ and
$b$}\label{section: local par}

 Define
 \begin{equation}
   y_j  \colonequals  y_j(\zeta)=\omega^j\Ai(\omega^j\zeta),\quad \text{for $j=0,1,2$},
 \end{equation}
where $\omega=e^{\frac{2\pi i}{3}}$ and $\Ai$ is the Airy function.
\begin{figure}[h]
  \begin{center}
  \includegraphics{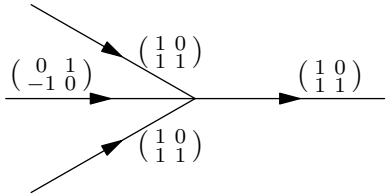}
  \end{center}
  \caption{The contour $\Gamma$ and the jump matrices for $A$.}
  \label{fig:Airy}
\end{figure}

Let
\begin{equation}
  \Gamma  \colonequals  e^{-\frac{2\pi i}{3}}\mathbb R^+\cup e^{\frac{2\pi i}{3}}\mathbb R^+\cup\realR
\end{equation}
be the contour consisting of four rays oriented each from the left to the right shown in Figure \ref{fig:Airy}, and define the $2 \times
2$ matrix-valued function $A$ in $\compC \setminus \Gamma$ as
\begin{equation}\label{def A}
  A(\zeta)  \colonequals
  \begin{cases}
    \sqrt{2\pi}e^{-\frac{\pi i}{4}}
    \begin{pmatrix}
      y_0 & -y_2\\
      y_0' & -y_2'
    \end{pmatrix},
    &\textnormal{for $0<\arg \zeta<\frac{2\pi}{3}$,} \\
    \sqrt{2\pi}e^{-\frac{\pi i}{4}}
    \begin{pmatrix}
      -y_1 & -y_2\\
      -y_1' & -y_2'
    \end{pmatrix},
    &\textnormal{for $\frac{2\pi}{3}<\arg \zeta<\pi$,} \\
    \sqrt{2\pi}e^{-\frac{\pi i}{4}}
    \begin{pmatrix}
      -y_2 & y_1\\
      -y_2' & y_1'
    \end{pmatrix},
    &\textnormal{for $-\pi<\arg \zeta<-\frac{2\pi}{3}$,} \\
    \sqrt{2\pi}e^{-\frac{\pi i}{4}}
    \begin{pmatrix}
      y_0 & y_1\\
      y_0' & y_1'
    \end{pmatrix},
    &\textnormal{for $-\frac{2\pi}{3}<\arg \zeta<0$}.
  \end{cases}
\end{equation}
Using the identity $y_0+y_1+y_2=0$, the fact that the Airy function
is an entire function, and the asymptotics as $\zeta\to\infty$ of
the Airy function, one verifies that $A$ satisfies the following
model RH problem. This RH problem (and equivalent forms of it) appeared many times in the literature and is often referred to as ``the Airy RH problem'', see for example \cite{Deift-Kriecherbauer-McLaughlin-Venakides-Zhou99, Deift99}.

\subsubsection*{RH problem for $A$}

\begin{enumerate}[label=(\alph{*})]
\item $A$ is a $2 \times 2$ matrix-valued function analytic in $\mathbb{C}\setminus\Gamma$.
\item $A$ satisfies the following jump relations on $\Gamma$,
  \begin{align}
    A_+(\zeta) = {}& A_-(\zeta)
    \begin{pmatrix}
      1 & 1\\
      0 & 1
    \end{pmatrix},
    & & \textnormal{for $\arg\zeta=0$,} \label{RHP A: b1} \\
    A_+(\zeta) = {}& A_-(\zeta)
    \begin{pmatrix}
      1 & 0\\
      1 & 1
    \end{pmatrix},
    & &\textnormal{for $\arg\zeta=\frac{2\pi}{3}$ or $\arg\zeta=\frac{-2\pi}{3}$,} \label{RHP A: b2}\\
   A_+(\zeta) = {}& A_-(\zeta)
    \begin{pmatrix}
      0 & 1\\
      -1 & 0
    \end{pmatrix},
    & & \textnormal{for $\arg\zeta=\pi$.} \label{RHP A: b3}
  \end{align}
\item $A$ has the following behavior at infinity,
  \begin{equation}\label{RHP A: c}
    A(\zeta)=\frac{1}{\sqrt{2}}\zeta^{-\frac{1}{4}\sigma_3}
    \begin{pmatrix}
      1 & 1 \\
      -1 & 1
    \end{pmatrix}
    e^{-\frac{\pi i}{4}\sigma_3}(I+\bigO(\zeta^{-3/2}))e^{-\frac{2}{3}\zeta^{3/2}\sigma_3}, \quad
    \textnormal{as $\zeta\to\infty$,}
  \end{equation}
  uniformly for $\zeta\in\mathbb{C}\setminus\Gamma$.
\end{enumerate}

\medskip

Using the regularity condition which says that $\lim_{x\to b_-}\frac{\psi_V(x)}{\sqrt{b-x}}$ exists and is positive, and the formulas of $\gfn(z)$ and $\tilde{\gfn}(z)$, and noting in addition that $\phi(b)=0$, we obtain the following local behavior for $\phi$ near $b$,
\begin{equation}
  \phi(z) = -c(z-b)^{3/2}+\bigO(\lvert z-b \rvert^{5/2}), \quad \textnormal{as $z\to
b$, where $c>0$.}
\end{equation}
Then in a neighborhood $U_b$ of $b$, there is a unique analytic
function $f_b$ satisfying $f_b(b) = 0$, $f_b'(b) > 0$ and
\begin{equation}\label{def fb}
\frac{2}{3}f_b(z)^{3/2}=-\frac{1}{2}\phi(z).
\end{equation}

Now we choose the lens $\Sigma_S$ in such a way that $f_b(z)$
maps the jump contour $U_b \cap \Sigma_S$ for $S$ on the jump
contour $\Gamma$ for $A$, and we define the $2 \times 2$ matrix-valued
function $P^{(b)}(z)$ on $U_b \setminus \Sigma_S$ as
\begin{equation}\label{Pb}
  P^{(b)}(z)  \colonequals  A(n^{2/3} f_b(z)) e^{-\frac{1}{2} (n \phi(z) - z)\sigma_3}.
\end{equation}
Using the jump relations \eqref{RHP A: b1}--\eqref{RHP A: b3} for $A$ and \eqref{eq:phi_across_real_line_left} and \eqref{eq:phi_across_real_line_middle} for $\phi(z)$, one verifies that
\begin{equation} \label{eq:jump_of_P^2}
  P^{(b)}_+(z) = P^{(b)}_-(z) J_S(z), \quad \textnormal{for $z \in U_b \cap \Sigma_S$,}
\end{equation}
where $J_S$ is given in \eqref{eq:defn_of_jump_matrix_J_S}. Since the determinant of $A$ is identically equal to $1$, $A$ is invertible, and so is
$P^{(b)}(z)$ for $z\in U_b \cap \Sigma_S$. By \eqref{eq:jump_of_P^2} and \eqref{RHP_S:jump}, we have
\begin{equation} \label{eq:jump_of_S_P^2}
  S_+(z)P^{(b)}_+(z)^{-1} = S_-(z)P^{(b)}_-(z)^{-1}, \quad \textnormal{for $z \in U_b \cap \Sigma_S$.}
\end{equation}

Similarly, near $a$,
\begin{equation}
  \phi(z)=-\tilde{c}(a-z)^{3/2}+\bigO(\lvert a-z \rvert^{5/2}) \pm 2\pi i, \quad \textnormal{as $z\to
a$, $\tilde{c}>0$,}
\end{equation}
where the sign of $\pm 2\pi i$ depends on whether $z$ is in the
upper or lower half plane. In a neighborhood $U_a$ of $a$, there is
a unique analytic function $f_a$ satisfying $f_a(a)=0$, $f_a'(a) < 0$, and
\begin{equation}\label{def fa}
  \frac{2}{3}f_a(z)^{3/2}=-\frac{1}{2}\phi(z) \pm \pi i.
\end{equation}
Again we can choose the lens $\Sigma_S$ in such a way that
$f_a(z)$ maps the jump contour $U_a \cap \Sigma_S$ for $S$ on the
jump contour $\Gamma$ for $A$. Then define the $2 \times 2$
matrix-valued function $P^{(a)}(z)$ on $U_a \setminus \Sigma_S$ as
\begin{equation}\label{Pa}
  P^{(a)}(z)  \colonequals  \sigma_3A(n^{2/3} f_a(z)) e^{-\frac{1}{2} (n \phi(z) - z)\sigma_3}\sigma_3.
\end{equation}
Similarly to \eqref{eq:jump_of_P^2} and \eqref{eq:jump_of_S_P^2}, we
have
\begin{equation}
  S_+(z)P^{(a)}_+(z)^{-1} = S_-(z)P^{(a)}_-(z)^{-1}, \quad \textnormal{for $z \in U_a \cap \Sigma_S$.}
\end{equation}

\begin{rmk}
Usually, a local parametrix serves as a local approximation to the solution of the RH problem. Since $S$ is vector-valued and our local parametrices $P^{(a)}$ and $P^{(b)}$ are $2\times 2$-valued, this is not quite true in our situation, but it will turn out later that large $n$ asymptotics for $S$ near $a$ and $b$ can be expressed in terms of $P^{(b)}$ and $P^{(a)}$, and thus in terms of the Airy function. Later in Section \ref{subsec:RH_of_P^infty_and_F}, we will build a vector-valued ``global parametrix'' $P^{(\infty)}$, which approximates $S$ away from the endpoints $a$ and $b$. Before introducing $P^{(\infty)}$, we perform one more transformation of the RH problem for $S$ in the next subsection.
\end{rmk}

\subsection{Third transformation $S\mapsto P$}

The following transformation will modify the jumps in the vicinity of $a$ and $b$: the jumps on $\Sigma_1$ and $\Sigma_2$ will be removed in $U_a$ and $U_b$. As a drawback, a discontinuity will appear on $\partial U_a$ and $\partial U_b$, but the jump matrices on these boundaries will be close to the identity matrix for large $n$.

Define
\begin{equation}\label{def P}
  P(z)  \colonequals
  \begin{cases}
    S(z) & \textnormal{for $z \in \compC \setminus (\overline{U_a} \cup \overline{U_b} \cup \Sigma_S)$,} \\
    S(z) P^{(a)}(z)^{-1} \frac{1}{\sqrt{2}} (n^{2/3} f_a(z))^{-\frac{1}{4}\sigma_3}
    \begin{pmatrix}
      1 & 1 \\
      -1 & 1
    \end{pmatrix}
    e^{-(\frac{\pi i}{4} - \frac{z}{2})\sigma_3} & \textnormal{for $z \in U_a \setminus \Sigma_S$,} \\
    S(z) P^{(b)}(z)^{-1} \frac{1}{\sqrt{2}} (n^{2/3} f_b(z))^{-\frac{1}{4}\sigma_3}
    \begin{pmatrix}
      1 & 1 \\
      -1 & 1
    \end{pmatrix}
    e^{-(\frac{\pi i}{4} - \frac{z}{2})\sigma_3} & \textnormal{for $z \in U_b \setminus \Sigma_S$.}
  \end{cases}
\end{equation}
Then $P$ is constructed in such a way that it has jumps on a contour

\begin{equation}\label{SigmaP}
  \Sigma_P  \colonequals  (\Sigma_S \setminus (U_a \cup U_b)) \cup \partial U_a \cup \partial U_b \cup [a, b]
\end{equation}
  as shown in Figure \ref{fig:Sigma_P}. We define $(n^{2/3} f_b(z))^{-\frac{1}{4}\sigma_3}$ and $(n^{2/3} f_a(z))^{-\frac{1}{4}\sigma_3}$ in such a way that they have branch cuts on $[a, b]$ and they are positive on $(b, \infty)$ and $(-\infty, a)$ respectively. The jumps inside the disks on $\realR \setminus [a, b]$ and the lips $\Sigma_1, \Sigma_2$ are equal to the identity matrix since $S(z) P^{(b)}(z)^{-1}$ and $S(z) P^{(a)}(z)^{-1}$ are analytic there, but there is a jump on $(a,b)$ due to the branch cuts of $(n^{2/3} f_b(z))^{-\frac{1}{4}\sigma_3}$ and $(n^{2/3} f_a(z))^{-\frac{1}{4}\sigma_3}$. Also note that, unlike $Y, T, S$ whose entries are all bounded in any bounded region of their domains, $P(z)$ has inverse fourth root singularities at $a$ and $b$.

\subsubsection*{RH problem for $P$}

\begin{enumerate}[label=(\alph{*})]
\item $P = (P_1, P_2)$, where $P_1$ is analytic in $\mathbb C\setminus \Sigma_P$, and
$P_2$ is analytic in $\cyld \setminus \Sigma_P$,
\item we have
  \begin{equation}\label{jump P}
    P_+(z) = P_-(z)J_P(z),\quad \textnormal{for $z \in \Sigma_P$,}
  \end{equation}
  where
  \begin{equation} \label{eq:last_RHP_for_P0}
    J_P(z)= 
    \begin{cases}
      J_S(z) & \textnormal{for $z \in \Sigma_S \setminus (\overline{U_a} \cup \overline{U_b})$,} \\
      \frac{1}{\sqrt{2}} e^{(\frac{\pi i}{4} - \frac{z}{2}) \sigma_3}
      \begin{pmatrix}
        1 & -1 \\
        1 & 1
      \end{pmatrix}
      (n^{2/3} f_a(z))^{\frac{1}{4} \sigma_3} P^{(a)}(z) &\textnormal{for $z \in \partial U_a$,} \\
      \frac{1}{\sqrt{2}} e^{(\frac{\pi i}{4} - \frac{z}{2}) \sigma_3}
      \begin{pmatrix}
        1 & -1 \\
        1 & 1
      \end{pmatrix}
      (n^{2/3} f_b(z))^{\frac{1}{4} \sigma_3} P^{(b)}(z) &\textnormal{for $z \in \partial U_b$,} \\
      \begin{pmatrix}
        0 & e^{-z} \\
        -e^z & 0
      \end{pmatrix}
      & \textnormal{for $z \in (a,b)$,} \\
    \end{cases}
  \end{equation}
  \addtocounter{enumi}{1}
  \begin{enumerate}[label=(\alph{enumi}\arabic{*}), leftmargin=0pt, itemindent=0pt]
  \item as $z \to \infty$, $P_1(z) = z^k + \bigO(z^{k-1})$,
  \item as $e^z \to +\infty$, $P_2$ behaves as $P_2(z) = \bigO(e^{-(k+1)z})$, and as $e^z \to 0$, $P_2$ behaves as $P_2(z) = \bigO(1)$.
  \item
 \begin{align}
    P(z) = {}& (\bigO(|z-a|^{-1/4}), \bigO(|z-a|^{-1/4})) & & \text{ as $z\to a$,} \\ 
    P(z) = {}& (\bigO(|z-b|^{-1/4}), \bigO(|z-b|^{-1/4})) & & \text{ as $z\to b$.} 
  \end{align}
  \end{enumerate}
\end{enumerate}

\begin{figure}[htp]
  \centering
  \includegraphics{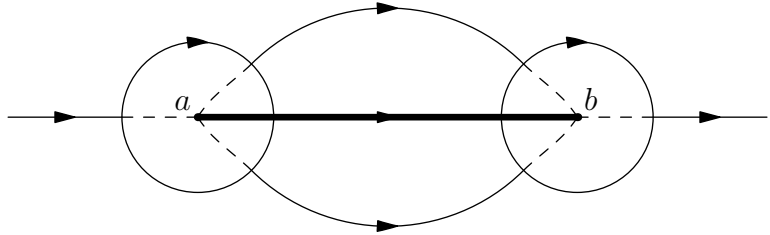}
  \caption[]{The contour $\Sigma_P$. On the boldface part of the contour, $J_P =
  \left( \begin{smallmatrix} 0 & e^{-z} \\ -e^z &0 \end{smallmatrix} \right)$ and on the other parts, $J_P \to I$ uniformly. Note that $\Sigma_P$ divides the complex plane into six regions: the two ``edge regions'' $U_a$ and $U_b$, the two ``bulk regions'' in the upper and lower parts of the lens and not in $U_a$ or $U_b$, and the two ``outside regions''. The dashed lines that belong to $\Sigma_S$ but not to $\Sigma_P$, together with the interval $(a, b)$, divide each edge region into four subregions, two inside the lens and two out of the lens.}
  \label{fig:Sigma_P}
\end{figure}

\subsection{Construction of the outer parametrix} \label{subsec:RH_of_P^infty_and_F}

For $z\in\partial U_a\cup\partial U_b$, the definition of the local parametrices (\ref{Pb}) and
(\ref{Pa}) together with the asymptotics (\ref{RHP A: c}) for $A$
imply that $J_P(z)=I+\bigO(n^{-1})$ as $n\to\infty$. For $z \in
\Sigma_S \setminus [a, b]$ and not included in $U_a$ or $U_b$, by the asymptotics of $\phi(z)$ given in \eqref{eq:inequality_for_phi} and \eqref{var ineq g}, we have that $J_P(z)$ decays exponentially as $n \to \infty$. Thus, in some sense, it is expected that
\begin{equation} \label{eq:derivation_of_outer_parametrix}
  P(z) \to P^{(\infty)}(z),
\end{equation}
where $P^{(\infty)}(z)$ has the same analyticity, asymptotic, and
periodicity properties, and has the jump condition
\begin{equation} \label{eq:RHP_for_outer_parametrix}
  P^{(\infty)}_+(x) = P^{(\infty)}_-(x)
  \begin{pmatrix}
    0 & e^{-x} \\
    -e^x & 0
  \end{pmatrix},
  \quad \textnormal{for $x\in(a,b)$}.
\end{equation}
We would like to construct a solution to the following RH problem:

\subsubsection*{RH problem for $P^{(\infty)}$}

\begin{enumerate}[label=(\alph{*})]
\item $P^{(\infty)} = (P^{(\infty)}_1, P^{(\infty)}_2)$, where $P^{(\infty)}_1$ is an analytic function
  in $\compC \setminus [a,b]$, and $P^{(\infty)}_2$ is an analytic
  function in $\cyld \setminus [a,b]$,
\item $P^{(\infty)}$ satisfies the jump relation
  \eqref{eq:RHP_for_outer_parametrix},
  \addtocounter{enumi}{1}
  \begin{enumerate}[label=(\alph{enumi}\arabic{*}), leftmargin=0pt, itemindent=0pt]
  \item \label{enu:RHP_of_P^infty:c1}
    as $z \to \infty$, $P^{(\infty)}_1(z) = z^k + \bigO(z^{k-1})$,
  \item \label{enu:RHP_of_P^infty:c2}
    as $e^z \to +\infty$, $P^{(\infty)}_2$ behaves as $P^{(\infty)}_2(z) = \bigO(e^{-(k+1)z})$, and as $e^z \to 0$, $P^{(\infty)}_2$ behaves as $P^{(\infty)}_2(z) = \bigO(1)$,
  \item \label{enu:RHP_of_P^infty:c3}
    \begin{align}
      P^{(\infty)}(z) = {}& (\bigO(|z-a|^{-1/4}), \bigO(|z-a|^{-1/4})) & & \text{ as $z\to a$,} \\ 
      P^{(\infty)}(z) = {}& (\bigO(|z-b|^{-1/4}), \bigO(|z-b|^{-1/4})) & & \text{ as $z\to b$.} 
    \end{align}
  \end{enumerate}
\end{enumerate}
After the construction of $P^{(\infty)}$, we will prove the
convergence \eqref{eq:derivation_of_outer_parametrix}.

We use the transformation $\Jlike_{c_1,c_0}$ defined before in
\eqref{eq:Joukowsky_like_transform}, where the parameters $c_1$ and $c_0$ depend on $a$ and $b$, the endpoints of the support of the equilibrium measure. Recall $s_a$ and $s_b$ defined
in \eqref{eq:defn_of_s_a_and_s_b} and the relation
\eqref{eq:defining_formula_of_a_b} between $s_a, s_b$ and
$a,b$. Below we write $\Jlike_{c_1, c_0}$ as $\Jlike$ if there is no
confusion.

By Proposition \ref{prop:Joukowsky_like}, $\Jlike$ maps $\compC \setminus \bar{D}$ conformally to
$\compC \setminus [\Jlike(s_a), \Jlike(s_b)]$, and maps $D \setminus
[-\frac{1}{2},\frac{1}{2}]$ conformally to $\strip \setminus
[\Jlike(s_a), \Jlike(s_b)]$, so that we can define the function $F(s)$ on
$\compC \setminus (\gamma_1 \cup \gamma_2 \cup
[-\frac{1}{2},\frac{1}{2}])$ by
\begin{equation} \label{eq:defn_of_F(s)_by_f_and_g0}
  F(s)  \colonequals
  \begin{cases}
    P^{(\infty)}_1(\Jlike(s)) & \textnormal{for $s \in \compC \setminus \bar{D}$,} \\
    P^{(\infty)}_2(\Jlike(s)) & \textnormal{for $s \in D \setminus [-\frac{1}{2},\frac{1}{2}]$.}
  \end{cases}
\end{equation}
Since $P_2^{(\infty)}$ is defined on $\cyld$, that is, it satisfies a periodic boundary condition on $\strip$,
we have that the definition of $F(s)$ can be extended to
$(-\frac{1}{2},\frac{1}{2})$.
In this way the transformation from $P^{(\infty)}$ to $F$ is invertible:
we can recover the outside parametrix $P^{(\infty)}$ by the
formula
\begin{align}\label{def Pinfty}
  P^{(\infty)}_1(z)= {}& F(\Jinv_1(z)),& & \text{for $z\in\compC\setminus [a,b]$,}\\
  P^{(\infty)}_2(z)= {}& F(\Jinv_2(z)),& & \text{for $z\in\strip\setminus [a,b],$}
\end{align}
where $\Jinv_1$ and $\Jinv_2$ are, as defined in \eqref{eq:defn_of_I_1} and \eqref{eq:defn_of_I_2}, the inverses of
$\Jlike$ mapping $\mathbb C\setminus[a,b]$ to $\mathbb C\setminus
\overline D$ and to $D\setminus[-\frac{1}{2},\frac{1}{2}]$
respectively.
All information about the vector-valued function $P^{(\infty)}$ is now carried by the single complex-valued function $F$, which is discontinuous on $\gamma_1 \cup \gamma_2$ by definition.

From condition \ref{enu:RHP_of_P^infty:c2} of the RH problem for $P^{(\infty)}$ and the definition of $\Jlike$,
it follows that $F$ has a removable
singularity at $-\frac{1}{2}$, and that $F$ has a zero of multiplicity $k+1$ at $\frac{1}{2}$ if $k \geq 0$, a removable singularity if $k=-1$, and a pole of order $-k-1$ if $k < -1$. The inverse fourth root singularities of $P^{(\infty)}$ at $a, b$ are transformed into inverse square root singularities of $F$ at $s_a, s_b$, because $\Jlike'(s)$ has simple zeros at $s_a$ and $s_b$.
In order to compute the jump relation satisfied by $F$,
note that
\begin{equation}
  e^{\Jlike(s)} = e^{c_1 s + c_0} \frac{s + \frac{1}{2}}{s -
  \frac{1}{2}},
\end{equation}
and
\begin{align}
  F_+(s) = {}& P^{(\infty)}_{1,+}(\Jlike(s)), & F_-(s) = {}& P^{(\infty)}_{2,-}(\Jlike(s)),&& \text{for $s\in\gamma_1$,} \\
  F_+(s) = {}& P^{(\infty)}_{2,+}(\Jlike(s)), & F_-(s) = {}& P^{(\infty)}_{1,-}(\Jlike(s)),&& \text{for $s\in\gamma_2$.}
\end{align}
It is now straightforward to verify the following RH conditions for $F$.

\subsubsection*{RH problem for $F$}

\begin{enumerate}[label=(\alph{*})]
\item $F$ is analytic in
$\compC\setminus(\gamma_1\cup\gamma_2)$ if $k \geq -1$, and analytic in $\compC \setminus (\gamma_1 \cup \gamma_2 \cup \{ \frac{1}{2} \})$ if $k < -1$,
\item for $s\in\gamma_1\cup\gamma_2$, we have
\begin{align}
  F_+(s) = & -e^{c_1 s + c_0} \frac{s + \frac{1}{2}}{s - \frac{1}{2}}F_-(s), & \textnormal{for $s \in \gamma_1$,}
  \label{eq:1st_jump_condition_of_RHP_for_F(s)0} \\
  F_+(s) = & e^{-c_1 s - c_0} \frac{s - \frac{1}{2}}{s + \frac{1}{2}}F_-(s), & \textnormal{for $s \in \gamma_2$,}
  \label{eq:2nd_jump_condition_of_RHP_for_F(s)0}
\end{align}
\item we have the asymptotic conditions
\begin{align}
  F(s) = {}& c^k_1 s^k + \bigO(s^{k-1}), & & \textnormal{as $s \to \infty$,} \label{eq:boundary_condition_of_F(s)_at_infty}  \\
  F(s) = {}& \bigO((s-\frac{1}{2})^{k+1}), & & \textnormal{as $s \to \frac{1}{2}$,} \label{eq:boundary_condition_of_F(s)_at_1} \\
  F(s) = {}& \bigO(\lvert s - s_a \rvert^{-\frac{1}{2}}) & & \text{as $s \to s_a$,} \\
  F(s) = {}& \bigO(\lvert s - s_b \rvert^{-\frac{1}{2}}) & & \text{as $s \to s_b$.} \label{eq:boundary_condition_of_F(s)_at_s_b}
\end{align}
\end{enumerate}
One can explicitly construct a solution $F$ to the above RH problem:
\begin{equation}\label{def F1}
  F(s) =
  \begin{cases}
     c^k_1 \frac{(s + \frac{1}{2})(s - \frac{1}{2})^k}{\sqrt{(s-s_a)(s-s_b)}} & \textnormal{for $s \in \compC
    \setminus \bar{D}$,} \\
     c^k_1 (s-\frac{1}{2})^{k+1} \frac{e^{-c_1 s - c_0}}{\sqrt{(s-s_a)(s-s_b)}} & \textnormal{for $s \in D$.}
  \end{cases}
\end{equation}
where $\sqrt{(s-s_a)(s-s_b)}$ is taken to be continuous in
$\compC \setminus \gamma_1$ and $\sqrt{(s-s_a)(s-s_b)} \sim s$ as $s \to \infty$.

Note that $F(s)$ and the function $G_k(z)$ defined in \eqref{eq:defn_of_G_and_G_hat} are related by (upon expressing $s_a$ and $s_b$ by \eqref{eq:defn_of_s_a_and_s_b})
\begin{equation} \label{eq:relation_between_G_and_F}
  G_k(s) =
  \begin{cases}
    F(s) & \text{if $s \in \compC \setminus \overline{D}$,} \\
    e^{\Jlike(s)} F(s) & \text{if $s \in D$.}
  \end{cases}
\end{equation}

\subsection{The convergence of $P \to P^{(\infty)}$}
\label{sec:convergence-p-to}

\begin{figure}[htp]
  \centering
  \includegraphics[scale=0.7]{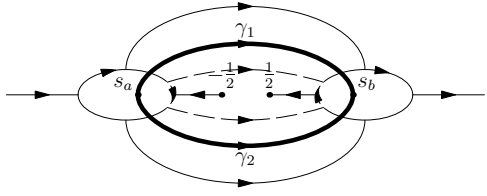}
  \caption{The contour $\Jlike^{-1}(\Sigma_P)$. The boldface part consists of $\gamma_1$ and $\gamma_2$, the solid part is $\Sigma'$ and the dashed part is $\Sigma''$.}
  \label{fig:Jlike_inverse_Sigma_P}
\end{figure}

We will now apply the same idea as in the construction of the outer parametrix to $P$, and want to transform the RH problem
to the $s$-plane using the transformation $z=\Jlike(s)$, in such a way that $P=(P_1, P_2)$ is transformed to a single complex-valued function $\F$. Therefore
we define $\F$ on $\compC \setminus
\Jlike^{-1}(\Sigma_P)$ analogous to
\eqref{eq:defn_of_F(s)_by_f_and_g0}:
\begin{equation}\label{def F2}
  \F(s)  \colonequals
  \begin{cases}
    P_1(\Jlike(s)) & \textnormal{if $s \in \compC \setminus \bar{D}$ and $\Jlike(s) \notin \Sigma_P$,} \\
    P_2(\Jlike(s)) & \textnormal{if $s \in D \setminus [-\frac{1}{2}, \frac{1}{2}]$ and $\Jlike(s) \notin \Sigma_P$.}
  \end{cases}
\end{equation}
The inverse of this transformation is given by
\begin{align}\label{def P F}
  P_1(z) = {}& \F(\Jinv_1(z)), & & \text{for $z\in\compC\setminus\Sigma_P$,} \\
  P_2(z) = {}& \F(\Jinv_2(z)), & & \text{for $z\in\strip\setminus\Sigma_P$.}
\end{align}
The jump contour of $\F$ will consist of the inverse image of $\Sigma_P$ under $\Jlike$.
We can decompose this jump contour $\Jlike^{-1}(\Sigma_P)$
into three different parts: $\gamma_1 \cup \gamma_2$, the part in $D$ and the part in $\compC \setminus \overline{D}$ as follows, see Figure \ref{fig:Jlike_inverse_Sigma_P}:
\begin{equation} \label{eq:defn_of_Sigma'}
  \Jlike^{-1}(\Sigma_P)=\Sigma'\cup\Sigma''\cup(\gamma_1\cup\gamma_2),\quad \text{where} \quad
\Sigma'=\Jinv_1(\Sigma_P\setminus[a,b]),\quad
\Sigma''=\Jinv_2(\Sigma_P\setminus[a,b]).
\end{equation}
 Similar to $F(s)$, the definition of $\F(s)$ can be
extended to $[-\frac{1}{2}, \frac{1}{2})$ because of the periodicity of $P_2$ and its behavior as $\Re z\to -\infty$. The RH problem for $P$, however,
no longer transforms to a scalar RH problem for $\F(s)$. For
$s\in \gamma_1\cup\gamma_2$, we still have
the scalar jump conditions
\begin{align}
 \label{jump F1} \F_+(s) = {}& -e^{\Jlike(s)} \F_-(s), & & \text{for $s\in\gamma_1$}, \\
  \label{jump F2}  \F_+(s) = {}& e^{-\Jlike(s)} \F_-(s), & & \text{for $s\in\gamma_2$},
\end{align}
but on the other parts of the jump contour, the jump conditions become non-local.
Since $\F_\pm(s) = P_{1,\pm}(\Jlike(s))$ for $s\in\Sigma'$ and $\F_\pm(s) = P_{2,\pm}(\Jlike(s))$ for $s\in\Sigma''$, where the orientation for $\Sigma'$ and $\Sigma''$ is that inherited from the orientation on $\Sigma_P$ through $\Jinv_1$ and $\Jinv_2$,
the jump conditions (\ref{jump P}) for $P$ transform into
\begin{align}
\F_+(s) = {}& J_{P, 11}(\Jlike(s)) \F_-(s) + J_{P,21}(\Jlike(s)) \F_-(\Jinv_2(\Jlike(s))), &&\text{for $s\in\Sigma'$}, \label{jump F3} \\
  \F_+(s) = {}& J_{P, 12}(\Jlike(s)) \F_-(\Jinv_1(\Jlike(s))) + J_{P,22}(\Jlike(s)) \F_-(s), &&\text{for $s\in\Sigma''$}, \label{jump F4}
\end{align}
where $J_P$ is the $2 \times 2$ jump matrix defined in \eqref{eq:last_RHP_for_P0}.
In other words, the boundary value $\F_+(\Jinv_1(z))$ depends not only on $\F_-(\Jinv_1(z))$, but also on $\F_-(\Jinv_2(z))$, and vice versa for $\F_+(\Jinv_2(z))$. For this reason, we will call the jump relations \eqref{jump F3}--\eqref{jump F4} ``shifted'' jump relations, and the RH problem for $\F$ a shifted RH problem, following the terminology of \cite{Gakhov90}.
The asymptotic conditions for $\F$ are the same as the ones for $F$. By conditions \ref{enu:RHP_of_P^infty:c1}--\ref{enu:RHP_of_P^infty:c3} of the RH problem for $P$, we have analogous to \eqref{eq:boundary_condition_of_F(s)_at_infty}--\eqref{eq:boundary_condition_of_F(s)_at_s_b} that
\begin{align}
  \F(s) = {}& c^k_1 s^k + \bigO(s^{k-1}), & & \textnormal{as $s \to \infty$,} \label{eq:boundary_condition_of_F_at_infty}  \\
  \F(s) = {}& \bigO((s-\frac{1}{2})^{k+1}), & & \textnormal{as $s \to \frac{1}{2}$.}  \label{eq:boundary_condition_of_F_at_1} \\
  \F(s) = {}& \bigO(\lvert s-s_a \rvert^{-1/2}), & & \text{as $s\to s_a$,} \label{F loc a} \\
  \F(s) = {}& \bigO(\lvert s-s_b \rvert^{-1/2}), & & \text{as $s\to s_b$.} \label{F loc b}
\end{align}
Since $F(s) \neq 0$ for all $s \in \compC \setminus
(\gamma_1 \cup \gamma_2 \cup \{ \frac{1}{2} \})$, while at $\frac{1}{2}$ the order of the pole of $\F(s)$ is at most equal to that of $F(s)$, we can define the analytic function
\begin{equation}\label{def R}
  R(s)  \colonequals  \frac{\F(s)}{F(s)}, \quad \textnormal{for $s \in \compC \setminus \Jlike^{-1}(\Sigma_P)$.}
\end{equation}
By \eqref{jump F1}, \eqref{jump F2} and
\eqref{eq:1st_jump_condition_of_RHP_for_F(s)0}, \eqref{eq:2nd_jump_condition_of_RHP_for_F(s)0},
it follows that $R$ is analytic across $(\gamma_1\cup\gamma_2)$. Furthermore, the RH problem for $F(s)$ and the shifted RH problem for $\F(s)$ yield the following shifted RH conditions satisfied by $R$.

\subsubsection*{Shifted RH problem for $R$}

\begin{enumerate}[label=(\alph{*})]
\item $R$ is analytic in $\compC\setminus
(\Sigma'\cup\Sigma'')$,
\item $R$ has the jump conditions
\begin{align}
 R_+(s) = {}& J_{R, 11}(s) R_-(s) + J_{R,21}(s) R_-(\Jinv_2(\Jlike(s))), && \text{for $s\in\Sigma'$}, \label{jump R1} \\
 R_+(s) = {}& J_{R, 12}(s) R_-(\Jinv_1(\Jlike(s))) + J_{R,22}(s) R_-(s), && \text{for $s\in\Sigma''$}, \label{jump R2}
\end{align}
where
\begin{align} 
  J_{R, 11}(s) = {}& J_{P, 11}(\Jlike(s)), & J_{R, 21}(s) = {}& J_{P,21}(\Jlike(s)) \frac{F(\Jinv_2(\Jlike(s)))}{F(s)}, \label{eq:J_R_upper} \\
  J_{R, 12}(s) = {}& J_{P,12}(\Jlike(s)) \frac{F(\Jinv_1(\Jlike(s)))}{F(s)}, & J_{R, 22}(s) = {}& J_{P, 22}(\Jlike(s)), \label{eq:J_R_lower}
\end{align}
\item $R$ is bounded, and $R(s)=1+\bigO(s^{-1})$ as $s\to\infty$.
\end{enumerate}
Substituting the asymptotic properties of $J_P$ stated in the beginning of Section \ref{subsec:RH_of_P^infty_and_F} and the formula \eqref{def F1} of $F(s)$ into \eqref{eq:J_R_upper} and \eqref{eq:J_R_lower}, as $n \to \infty$, we have the uniform asymptotic estimates
\begin{align}
   J_{R, 11}(s)= {}& 1+\bigO(n^{-1}), & J_{R, 21}(s)= {}& \bigO(n^{-1}), &\text{for $s$}\in {}& \Sigma', \label{eq:J_R:1} \\
   J_{R, 12}(s)= {}& \bigO(n^{-1}), & J_{R, 22}(s)= {}& 1+\bigO(n^{-1}), &\text{for $s$}\in {}& \Sigma''. \label{eq:J_R:2}
\end{align}
Moreover, for $s$ on the real parts of $\Sigma'$ and $\Sigma''$, $J_{R,21}$ vanishes identically: by (\ref{eq:last_RHP_for_P0}) and (\ref{eq:defn_of_jump_matrix_J_S}), we have
\begin{equation}
J_{R,21}(s)=0,\quad \text{for $s\in(\Sigma' \cup \Sigma'') \cap \mathbb R$.}\label{JR0}
\end{equation}

To obtain asymptotics for $R(s)$, we introduce an operator
$\Delta_R$ that acts on functions defined on
$\Sigma_R=\Sigma'\cup\Sigma''$. Let $f$ be a complex-valued function
defined on $\Sigma_R$. Then we define $g=\Delta_R f$ by
\begin{align}
  g(s) = {}& [J_{R, 11}(s) - 1]f(s) + J_{R, 21}(s) f(\Jinv_2(\Jlike(s))),&&\mbox{ for $s\in\Sigma'$,} \label{Delta1}\\
  g(s) = {}& J_{R, 12}(s) f(\Jinv_1(\Jlike(s)))+ [J_{R, 22}(s) - 1]f(s),&&\mbox{ for $s\in\Sigma''$}.\label{Delta2}
\end{align}
For bounded function $f(s)$, $g(s)$ is also bounded and decays rapidly as $\lvert s \rvert \to \infty$. If we regard $\Delta_R$ as a linear
operator from $L^2(\Sigma_R)$ to itself, we will see that it is bounded and that its operator norm
is $\bigO(n^{-1})$ as $n\to\infty$.
For that purpose, note first that, by (\ref{Delta1}) and (\ref{Delta2}),
\begin{multline}
\|\Delta_R f\|_{L^2(\Sigma_R)}
\leq \|[J_{R, 11} - 1]f\|_{L^2(\Sigma')}+\|J_{R, 21}f(\Jinv_2(\Jlike))\|_{L^2(\Sigma')}\\
+\|J_{R, 12}f(\Jinv_1(\Jlike))\|_{L^2(\Sigma'')}+\|[J_{R, 22} - 1]f\|_{L^2(\Sigma')}.
\end{multline}
Using the fact that $J_{R,11}-1$ and $J_{R,22}-1$ are uniformly $\bigO(n^{-1})$ on $\Sigma'$ and $\Sigma''$ as $n\to\infty$, see
(\ref{eq:J_R:1})--(\ref{eq:J_R:2}), we obtain that there exists a constant $c>0$ such that
\begin{equation}\label{norm 2}\|\Delta_R f\|_{L^2(\Sigma_R)}\leq \frac{c}{n}\|f\|_{L^2(\Sigma_R)}+\|J_{R, 21}f(\Jinv_2(\Jlike))\|_{L^2(\Sigma')}
+\|J_{R, 12}f(\Jinv_1(\Jlike))\|_{L^2(\Sigma'')}.\end{equation}
For the second term on the right-hand side, we have
\begin{equation}
  \begin{split}
    \|J_{R, 21}f(\Jinv_2(\Jlike))\|_{L^2(\Sigma')}^2= {}&\int_{\Sigma'} |f(\Jinv_2(\Jlike(s)))|^2 |J_{R,21}(s)|^2 ds\\ 
    = {}&\int_{\Sigma''} |f(u)|^2 |J_{R,21}(\Jinv_1(\Jlike(u)))|^2 |(\Jinv_1(\Jlike))'(u)| du \\ 
    \leq {}& \sup_{u\in\Sigma''}\left(|J_{R,21}(\Jinv_1(\Jlike(u)))|^2 |(\Jinv_1(\Jlike))'(u)|\right)\cdot \|f\|_{L^2(\Sigma_R)}^2 .
  \end{split}
\end{equation}
For $u\in\Sigma''$ bounded away from $\pm 1/2$, it is straightforward to verify by (\ref{eq:J_R:1}) and properties of the transformation $\Jlike$ that  $|J_{R,21}(\Jinv_1(\Jlike(u)))|^2 |(\Jinv_1(\Jlike))'(u)|$ is $\bigO(n^{-2})$ as $n\to\infty$, uniformly in $u$. For $u\in\Sigma''$ close to $\pm 1/2$, we observe by (\ref{JR0}) that $J_{R,21}(\Jinv_1(\Jlike(u)))=0$, which implies the existence of a constant $c'$ such that
\begin{equation}
  \|J_{R, 21}f(\Jinv_2(\Jlike))\|_{L^2(\Sigma')}\leq \frac{c'}{n}\|f\|_{L^2(\Sigma_R)}.
\end{equation}
Regarding the last term in (\ref{norm 2}),
\begin{equation}
  \begin{split}
    \|J_{R, 12}f(\Jinv_1(\Jlike))\|_{L^2(\Sigma'')}^2 = {}&\int_{\Sigma''} |f(\Jinv_1(\Jlike(s)))|^2 |J_{R,21}(s)|^2 ds \\ 
    = {}&\int_{\Sigma'} |f(u)|^2 |J_{R,12}(\Jinv_2(\Jlike(u)))|^2 |(\Jinv_2(\Jlike))'(u)| du \\ 
    \leq {}& \sup_{u\in\Sigma'}\left(|J_{R,12}(\Jinv_2(\Jlike(u)))|^2 |(\Jinv_2(\Jlike))'(u)|\right)\cdot \|f\|_{L^2(\Sigma_R)}^2,
  \end{split}
\end{equation}
 and it follows from (\ref{eq:J_R:2}) that
\begin{equation}
  \|J_{R, 12}f(\Jinv_1(\Jlike))\|_{L^2(\Sigma'')}\leq \frac{c''}{n}\|f\|_{L^2(\Sigma_R)}.
\end{equation}
From the above estimates, it follows that there exists a constant $M>0$ such that
\begin{equation}
\|\Delta_R f\|_{L^2(\Sigma_R)}\leq \frac{M}{n}\|f\|_{L^2(\Sigma_R)},\qquad \|\Delta_R\|_{L^2(\Sigma_R)}\leq \frac{M}{n} .
\end{equation}

Next, we define another bounded linear operator
$C_{\Delta_R}$ from $L^2(\Sigma_R)$ to itself, by
\begin{equation}
  C_{\Delta_R}(f)  \colonequals  C_-(\Delta_R(f)), \quad \text{where} \quad C_- g(s)=\frac{1}{2\pi i}\lim_{s'\to s_-}\int_{\Sigma_R} \frac{g(\xi)}{\xi
  -s'}d\xi,
\end{equation}
and the limit $s'\to s_-$ is taken when approaching the contour from
the minus side. The operator norm of $C_{\Delta_R}$ is also uniformly
$\bigO(n^{-1})$ as $n\to\infty$ since the Cauchy operator $C_-$ is bounded. Thus $(1
- C_{\Delta_R})$ can be inverted by a Neumann series for $n$
sufficiently large. We claim now that $R$ satisfies the integral equation
\begin{equation}\label{Plemelj}
R(s)=1+C(\Delta_R R_-)(s),\quad \text{where} \quad C g(s)=\frac{1}{2\pi
i}\int_{\Sigma_R} \frac{g(\xi)}{\xi
  -s}d\xi.
\end{equation}
To prove this claim, note that the solution to the RH problem for $R$
is unique because it is equivalent to the uniquely solvable RH problem for
$Y$. This means that it is sufficient to prove that the right-hand side of (\ref{Plemelj}), which we will denote by $\tilde R$ for simplicity, satisfies the RH conditions for $R$. Obviously $\tilde R(z)$ is bounded and tends
to $1$ as $z\to\infty$, and it suffices to prove that the solution $\tilde{R}$ satisfies the jump relations \eqref{jump R1} and \eqref{jump R2}. Using the Cauchy
operator identity $C_+-C_-=1$, it follows that
\begin{equation}
  \tilde{R}_+ - \tilde{R}_- = (1 + C_+(\Delta_R \tilde{R}_-)) - (1 + C_-(\Delta_R \tilde{R}_-)) = (C_+ - C_-) (\Delta_R \tilde{R}_-) = \Delta_R \tilde{R}_-,
\end{equation}
which implies indeed that $\tilde{R}$ satisfies the jump relations \eqref{jump R1} and \eqref{jump R2} for $R$. Hence we conclude that $R=\tilde R$, and \eqref{Plemelj} is proved.
Since $R$ satisfies \eqref{Plemelj}, we have, taking the limit where $s$ approaches the minus side of $\Sigma_R$,
\begin{equation} \label{eq:functional_equation_of_R_--1}
  R_- - 1 = C_-(\Delta_R R_-) = C_{\Delta_R}(R_- - 1) + C_-(\Delta_R(1)).
\end{equation}
By the invertibility of $(1 - C_{\Delta_R})$, \eqref{eq:functional_equation_of_R_--1} implies
\begin{equation} \label{eq:expression_of_R_-}
  R_- = 1 + (1 - C_{\Delta_R})^{-1} C_-(\Delta_R(1)).
\end{equation}
This further implies that
\begin{equation}\label{L2R}
\lVert R_- - 1 \rVert_{L^2(\Sigma_R)}=\bigO(n^{-1}),\quad \text{as $n\to\infty$}.
\end{equation} Substituting
\eqref{eq:expression_of_R_-} into \eqref{Plemelj}, we
obtain an expression for $R$:
\begin{equation}
  R = 1 + C(\Delta_R(1 + (1 - C_{\Delta_R})^{-1} C_-(\Delta_R(1)))).
\end{equation}
For $s$ at a small distance $\delta>0$ away from the contour $\Sigma_R$, \eqref{Plemelj} reads
\begin{equation}
  R(s)-1 =  \frac{1}{2\pi i}\int_{\Sigma_R} \frac{\Delta_R(R_--1)(\xi)}{\xi
  -s}d\xi+ \frac{1}{2\pi i}\int_{\Sigma_R} \frac{\Delta_R(1)(\xi)}{\xi
  -s}d\xi.
\end{equation}
The second term at the right-hand side of the above equation can be estimated by $\bigO(\delta^{-1} n^{-1})$, using the definition of the operator $\Delta_R$ and asymptotic properties of $J_R$.
Using in addition the Cauchy-Schwarz inequality applied on the first term on the right-hand side of the above equation, by (\ref{L2R}) we obtain
\begin{equation} \label{eq:R_away_from_jump_contour}
  \begin{split}
    \lvert R(s)-1 \rvert \leq {}&   \frac{1}{2\pi }\|\Delta_R(R_--1)\|_{L^2(\Sigma_R)}\cdot \|\frac{1}{\xi -s}\|_{L^2(\Sigma_R)}+\bigO(\delta^{-1} n^{-1}) \\
    \leq {}& \frac{1}{2\pi}\|\Delta_R\|_{L^2(\Sigma_R)}\cdot \|R_--1\|_{L^2(\Sigma_R)}\cdot \|\frac{1}{\xi -s}\|_{L^2(\Sigma_R)}+\bigO(\delta^{-1} n^{-1}) \\
    = {}& \bigO(\delta^{-1} n^{-1}) + \bigO(\delta^{-1} n^{-1}).
  \end{split}
\end{equation}
Although the estimate \eqref{eq:R_away_from_jump_contour} does not work well for $s$ in a $\delta$-neighborhood of $\Sigma_R$, we note that for such $s$, given that $\delta$ is small enough, the jump contour $\Sigma_R$ can always be deformed in such a way that $s$ lies at a distance $\delta$ away from it. After this deformation, the above argument can be applied to obtain the uniform estimate
\begin{equation}\label{asymp R}R(s)-1=\bigO(n^{-1}),\qquad \mbox{ as $n\to\infty$, $s\in\mathbb C\setminus\Sigma_R$}.\end{equation}


\medskip

Through \eqref{def R}, (\ref{def Pinfty}), and (\ref{def P F}), the uniform estimate \eqref{asymp R} yields
\begin{align}
  P_1(z) = {}& (1 + \bigO(n^{-1})) P^{(\infty)}_1(z), \quad \text{as $n \to \infty$, for $z \in \compC \setminus \Sigma_P$,} \label{eq:uniform_conv_P_1} \\
  P_2(z) = {}& (1 + \bigO(n^{-1})) P^{(\infty)}_2(z), \quad \text{as $n \to \infty$, for $z \in \cyld \setminus \Sigma_P$.} \label{eq:uniform_conv_P_2}
\end{align}
The asymptotics for $R$ can be used to obtain asymptotics for the
polynomials $p_n^{(n)}$ by inverting the transformations
\begin{equation}
  Y\mapsto T\mapsto S\mapsto P\mapsto R.
\end{equation}
We will do this in
Section \ref{section: proof results}.

\section{Asymptotic analysis for the type I multiple orthogonal polynomials}\label{section: RH q}

In a similar way as for the type II multiple orthogonal polynomials
$p_j^{(n)}(z)$, in this section we construct a RH problem for
$q_j^{(n)}(e^z)$, and we analyze this RH problem asymptotically when
$j=n+k$. Both the RH problem and the asymptotic analysis show many
similarities with the ones for the type II polynomials, and once
again the use of the transformation $\Jlike$ will turn out to be
crucial.

In this section, we write $q_j^{(n)}(x)$, the monic polynomials that define the multiple orthogonal polynomials of type I and satisfy the orthogonality relations \eqref{orthoI}, as $q_j(x)$ if there is no confusion.
\subsection{RH problem characterizing the polynomials}

Consider the Cauchy transform of
$q_j(e^z)$,
\begin{equation}
  Cq_j(z)  \colonequals  \frac{1}{2\pi i} \int_{\realR} \frac{q_j(e^s)}{s-z} e^{-nV(s)} ds.
\end{equation}
Due to the orthogonality \eqref{orthoI}, as $z \to \infty$,
\begin{equation}
  \begin{split}
    Cq_j(z) = {}&  \frac{-1}{2\pi i z} \int_{\realR} \frac{q_j(e^s)}{1-s/z} e^{-nV(s)} ds \\
    = {}& \frac{-1}{2\pi i z} \int_{\realR} \left( 1 + \frac{s}{z} + \frac{s^2}{z^2} + \cdots \right) q_j(e^s)
    e^{-nV(s)} ds \\
    = {}& \bigO(z^{-j-1}).
  \end{split}
\end{equation}
For $x \in \realR$, Cauchy's theorem implies
\begin{equation} \label{eq:jump+property_of_CQ_j}
  (Cq_j)_+(x) - (Cq_j)_-(x) = \frac{1}{2\pi i} \int_{\realR} \frac{q_j(e^s) e^{-nV(s)}}{s-x_+} ds -
  \frac{1}{2\pi i} \int_{\realR} \frac{q_j(e^s) e^{-nV(s)}}{s-x_-} ds
  = q_j(e^x) e^{-nV(x)}.
\end{equation}
Similar to \eqref{def Y}, let
\begin{equation} \label{def X}
  X(z) = X^{(j,n)}(z)  \colonequals  (q_j(e^z), Cq_j(z)).
\end{equation}
One verifies that $X$ satisfies the following RH problem.

\subsubsection*{RH problem for $X$}
\begin{enumerate}[label=(\alph{*})]
\item $X = (X_1, X_2)$, where $X_2$ is an analytic function defined on $\compC \setminus \realR$ and $X_1$ is an analytic function on $\cyld$,
\item $X$ has continuous boundary values $X_{\pm}$ when approaching the real line from above and below, and we have
  \begin{equation} \label{eq:RHP_of_X}
    X_+(x) = X_-(x)
    \begin{pmatrix}
      1 & e^{-nV(x)} \\
      0 & 1
    \end{pmatrix},\qquad
    \text{for $x \in \realR$,}
  \end{equation}
  \addtocounter{enumi}{1}
  \begin{enumerate}[label=(\alph{enumi}\arabic{*}), leftmargin=0pt, itemindent=0pt]
  \item as $z\to\infty$, $X_2$ behaves as $X_2(z)=\bigO(z^{-j-1})$,
  \item as $e^z\to\infty$ (\ie, $\Re z\to +\infty$), $X_1$ behaves as $X_1(z)=e^{jz}+\bigO(e^{(j-1)z})$; as $e^z\to 0$ (\ie, $\Re z\to -\infty$), $X_1$ remains bounded.
  \end{enumerate}
\end{enumerate}
In an analogous way as for the RH problem for $Y$ in Section \ref{subsec:RH_of_Y}, it can be shown that
$X = X^{(j,n)}$ given by \eqref{def X} is the unique solution to this RH
problem.

We will now perform an asymptotic analysis of the RH problem for
$X=X^{(n+k, n)}$ as $n\to\infty$, with $k$ a constant integer. This method will be to a large extent
analogous to the nonlinear steepest descent method done in the previous
section. Again we will construct a series of transformations of $X$
and end up with a shifted small-norm RH problem. In order to emphasize the
analogies with the previous section, we will use notations $\hat T, \hat S, \hat P, \hat R, \dotsc$  for the counterparts of the
functions $T, S, P, R,\dotsc$ used before.

\subsection{First transformation $X\mapsto \hat T$}

Recall the functions $\gfn(z)$ and $\tilde{\gfn}(z)$ defined in
\eqref{eq:expr_of_gfn_and_tilde_gfn-intro}, and define
\begin{equation} \label{def hatT}
  \hat T(z)  \colonequals  e^{-\frac{n\ell}{2}}X(z)
  \begin{pmatrix}
    e^{-n\tilde\gfn(z)}&0\\
    0&e^{n\gfn(z)}
  \end{pmatrix}
  e^{\frac{n\ell}{2}\sigma_3}.
\end{equation}
Analogously to $T$ in Section \ref{subsec:RH_of_T}, $\T$ satisfies the RH problem

\subsubsection*{RH problem for $\hat T$}
\begin{enumerate}[label=(\alph{*})]
\item $\T = (\T_1, \T_2)$, where $\hat T_2$ is analytic on $\compC \setminus \realR$ and
  $\hat T_1$ is defined and analytic in $\cyld \setminus \realR$,
\item $\T$ satisfies the jump relation
  \begin{equation} \label{eq:RHP_of_T}
    \hat T_+(x) = \hat T_-(x) J_{\T}(x),
    \quad \text{for $x \in \realR$.}
  \end{equation}
with
\begin{equation}
  J_{\T}(x) =
  \begin{pmatrix}
    e^{n(\tilde\gfn_-(x)-\tilde\gfn_+(x))} & e^{n(\tilde\gfn_-(x)+\gfn_+(x)-V(x)-\ell)} \\
    0 & e^{n(\gfn_+(x)-\gfn_-(x))}
  \end{pmatrix},
\end{equation}
  \addtocounter{enumi}{1}
  \begin{enumerate}[label=(\alph{enumi}\arabic{*}), leftmargin=0pt, itemindent=0pt]
  \item as $z\to\infty$, $\hat T_2$ behaves as
    $\hat T_2(z)=\bigO(z^{-(k+1)})$,
  \item as $e^z \to \infty$, $\hat T_1$
    behaves as $\hat T_1(z)= e^{kz}+\bigO(e^{(k-1)z})$, and as $e^z \to 0$, $\T_1$ behaves as $\T_1(z) = \bigO(1)$.
  \end{enumerate}
\end{enumerate}

\subsection{Second transformation $\hat T\mapsto \hat S$}

By \eqref{var eq g}, we have, like \eqref{facto}, the following factorization on $[a,b]$:
\begin{equation}
  J_{\T}(x) =
  \begin{pmatrix}
    1 & 0 \\
    e^{-n\phi_-(x)} & 1
  \end{pmatrix}
  \begin{pmatrix}
    0 & 1 \\
    - 1 & 0
  \end{pmatrix}
  \begin{pmatrix}
    1 & 0 \\
    e^{-n\phi_+(x)} & 1
  \end{pmatrix},
\end{equation}
where $\phi$ is defined in \eqref{def phi},
Recall the lens $\Sigma_S$ defined in Section \ref{subsec:RH_of_S} and shown in Figure \ref{fig:the_lens}. Similarly as in \eqref{eq:def S} for $S$, let us define $\hat S$ by
\begin{equation} \label{def hatS}
  \hat S(z)  \colonequals
  \begin{cases}
    \hat T(z) & \textnormal{outside of the lens,} \\
    \hat T(z)
    \begin{pmatrix}
      1 & 0 \\
      e^{-n \phi(z)} & 1
    \end{pmatrix}
    & \textnormal{in the lower part of the lens,} \\
    \hat T(z)
    \begin{pmatrix}
      1 & 0 \\
      -e^{-n\phi(z)} & 1
    \end{pmatrix}
    & \textnormal{in the upper part of the lens,}
  \end{cases}
\end{equation}
where $\phi(z)$ is defined in \eqref{def phi}. Then similar to the RH conditions satisfied by $S$ in Section \ref{subsec:RH_of_S}, we have the RH problem for $\S$ as follows.

\subsubsection*{RH problem for $\hat S$}
\begin{enumerate}[label=(\alph{*})]
\item $\S = (\S_1, \S_2)$, where $\hat S_2$ is analytic in $\mathbb C\setminus \Sigma_S$, and $\hat S_1$ is analytic in
$\cyld \setminus \Sigma_S$,
\item we have
\begin{equation}\label{RHP_hatS:jump}
  \hat S_+(z) = \hat S_-(z) J_{\S}(z),\qquad\mbox{ for $z\in\Sigma_S$,}
  \end{equation}
  where
  \begin{equation} \label{eq:defn_of_jump_matrix_hatJ_S}
    J_{\S}(z)=
    \begin{cases}
      \begin{pmatrix}
        1 & 0 \\
        e^{-n\phi(z)} & 1
      \end{pmatrix},&\textnormal{for $z\in\Sigma_1\cup\Sigma_2$,}\\
      \begin{pmatrix}
        0 & 1 \\
        -1 & 0
      \end{pmatrix},&\textnormal{for $z\in (a,b)$,}\\
      \begin{pmatrix}
        1 & e^{n\phi(z)} \\
        0 & 1
      \end{pmatrix},&\textnormal{for $z\in \mathbb R\setminus[a,b]$,}
    \end{cases}
  \end{equation}
  \addtocounter{enumi}{1}
  \begin{enumerate}[label=(\alph{enumi}\arabic{*}), leftmargin=0pt, itemindent=0pt]
  \item as $z\to\infty$, $\S_2$ behaves as $\hat S_2(z) = \bigO(z^{-(k+1)})$,
  \item as $e^z \to \infty$, $\S_1$ behaves as $\hat S_1(z) = e^{kz} + \bigO(e^{(k-1)z})$, and as $e^z \to 0$, $\S_1$ behaves as $\S_1(z) = \bigO(1)$.
  \end{enumerate}
\end{enumerate}

\subsection{Construction of local parametrices near $a$ and $b$}

In a similar way as for the construction of $P^{(a)}$ and $P^{(b)}$ in Section \ref{section: local par}, we can construct
local parametrices $\hat P^{(a)}$ and $\hat P^{(b)}$ in sufficiently
small neighborhoods $U_a$ and $U_b$ of the endpoints $a$ and $b$ in
such a way that they satisfy exactly the jump conditions
\begin{align}
\hat P^{(a)}_+(z)= {}& \hat P^{(a)}_-(z) J_{\S}(z),&&\text{$z\in \Sigma_S\cap U_a$}, \label{Pab1} \\
\hat P^{(b)}_+(z)= {}& \hat P^{(b)}_-(z) J_{\S}(z),&&\text{$z\in \Sigma_S\cap U_b$.} \label{Pab2}
\end{align}
Similar to the $P^{(a)}(z)$ and $P^{(b)}(z)$ defined in \eqref{Pa} and \eqref{Pb} respectively, the local parametrices $\hat P^{(a)}(z)$ and $\hat P^{(b)}(z)$ are expressed by
\begin{align}
  \hat P^{(a)}(z)  \colonequals  {}&\sigma_3 A(n^{2/3} f_a(z)) e^{-\frac{n}{2} \phi(z)\sigma_3}\sigma_3,\\
  \hat P^{(b)}(z)  \colonequals  {}& A(n^{2/3} f_b(z)) e^{-\frac{n}{2}
 \phi(z)\sigma_3},
\end{align}
where the functions $f_a$ and $f_b$ are as in \eqref{def fb} and \eqref{def fa}, $A$ is as in \eqref{def A}, and the neighborhoods $U_a$ and $U_b$ as well as the contour $\Sigma_S$ can be taken the same as in \eqref{Pa} and \eqref{Pb}. We omit the details of the verification of \eqref{Pab1} and \eqref{Pab2} here, since almost identical arguments were used in Section \ref{section: local par}.

\subsection{Third transformation $\hat S\mapsto \hat P$}

Define analogously to $P(z)$ in \eqref{def P},
\begin{equation}\label{def hatP}
  \hat P(z) =
  \begin{cases}
    \hat S(z) & \textnormal{for $z \in \compC \setminus (\overline{U_a} \cup \overline{U_b} \cup \Sigma_S)$,} \\
    \hat S(z) \hat P^{(a)}(z)^{-1} \frac{1}{\sqrt{2}} (n^{2/3}f_a(z))^{-\frac{1}{4}\sigma_3}
    \begin{pmatrix}
      1 & 1 \\
      -1 & 1
    \end{pmatrix}
    e^{-\frac{\pi i}{4}\sigma_3} & \textnormal{for $z \in U_a \setminus \Sigma_S$,} \\
    \hat S(z) \hat P^{(b)}(z)^{-1} \frac{1}{\sqrt{2}} (n^{2/3} f_b(z))^{-\frac{1}{4}\sigma_3}
    \begin{pmatrix}
      1 & 1 \\
      -1 & 1
    \end{pmatrix}
    e^{-\frac{\pi i}{4}\sigma_3} & \textnormal{for $z \in U_b \setminus \Sigma_S$.}
  \end{cases}
\end{equation}
Then like the RH conditions satisfied by $P$, $\P$ satisfies the following RH conditions.

\subsubsection*{RH problem for $\hat P$}

\begin{enumerate}[label=(\alph{*})]
\item $\P = (\P_1, \P_2)$, where $\hat P_2$ is analytic in $\mathbb C\setminus \Sigma_P$, and
$\hat P_1$ is analytic in $\cyld \setminus \Sigma_P$,
\item we have
  \begin{equation}
    \hat P_+(z) = \hat P_-(z) J_{\P}(z),\quad \textnormal{for $z \in \Sigma_P$,}
  \end{equation}
  where $\Sigma_P$ is the same as in \eqref{SigmaP}, and
  \begin{equation} \label{eq:last_RHP_for_P}
    J_{\P}(z)=
    \begin{cases}
      J_{\S}(z) & \textnormal{for $z \in \Sigma_S \setminus (\overline{U_a} \cup \overline{U_b})$,} \\
      \frac{1}{\sqrt{2}} e^{\frac{\pi i}{4} \sigma_3}
      \begin{pmatrix}
        1 & -1 \\
        1 & 1
      \end{pmatrix}
      (n^{2/3} f_a(z))^{\frac{1}{4} \sigma_3} P^{(a)}(z) &\textnormal{for $z \in \partial U_a$,} \\
      \frac{1}{\sqrt{2}} e^{\frac{\pi i}{4}  \sigma_3}
      \begin{pmatrix}
        1 & -1 \\
        1 & 1
      \end{pmatrix}
      (n^{2/3} f_b(z))^{\frac{1}{4} \sigma_3} P^{(b)}(z) &\textnormal{for $z \in \partial U_b$,} \\
      \begin{pmatrix}
        0 & 1 \\
        -1 & 0
      \end{pmatrix}
      & \textnormal{for $z \in (a,b)$,}
    \end{cases}
  \end{equation}
  \addtocounter{enumi}{1}
  \begin{enumerate}[label=(\alph{enumi}\arabic{*}), leftmargin=0pt, itemindent=0pt]
  \item as $z\to\infty$, $\hat P_2(z) = \bigO(z^{-(k+1)})$,
  \item as $e^z \to \infty$, $\P_1$ behaves as $\hat P_1(z) = e^{kz} + \bigO(e^{(k-1)z})$, and as $e^z \to 0$, $\P_1$ behaves as $\P_1(z) = \bigO(1)$.
  \item
 \begin{align}
    \hat{P}(z) = {}& (\bigO(|z-a|^{-1/4}), \bigO(|z-a|^{-1/4})) & & \text{ as $z\to a$,} \\ 
    \hat{P}(z) = {}& (\bigO(|z-b|^{-1/4}), \bigO(|z-b|^{-1/4})) & & \text{ as $z\to b$.} 
  \end{align}

  \end{enumerate}
\end{enumerate}

\subsection{Construction of the outer parametrix}

The RH problem for $\P$ has, as $P$, the property that its jump matrix tends to the
identity matrix uniformly as $n\to\infty$, except on $[a,b]$. We will first construct a
solution to the following RH problem for $\P^{(\infty)}$, which is the
limiting RH problem (formally, ignoring small neighborhoods of $a$ and $b$) for $\hat P$ as $n\to\infty$.

\subsubsection*{RH problem for $\hat P^{(\infty)}$}

\begin{enumerate}[label=(\alph{*})]
\item $\hat P^{(\infty)} = (\hat P^{(\infty)}_1, \hat P^{(\infty)}_2)$, where $\hat P^{(\infty)}_2$
  is an analytic function in $\compC \setminus [a,b]$, and $\hat
P^{(\infty)}_1$ is an analytic function in $\cyld \setminus [a,b]$,
\item $\hat P^{(\infty)}$ satisfies the jump relation
  \begin{equation}
    \label{jump hatPinfty} \hat P_+^{(\infty)}(x)=\hat
    P_-^{(\infty)}(x)
    \begin{pmatrix}
      0&1\\
      -1&0
    \end{pmatrix},
    \quad \mbox{ for
      $x\in(a,b)$,}
  \end{equation}
  \addtocounter{enumi}{1}
  \begin{enumerate}[label=(\alph{enumi}\arabic{*}), leftmargin=0pt, itemindent=0pt]
  \item \label{enu:RHP_of_hatP^infty:c1}
    as $z \to \infty$, $\P^{(\infty)}_2(z) = \bigO(z^{-(k+1)})$,
  \item \label{enu:RHP_of_hatP^infty:c2}
    as $e^z \to +\infty$, $\P^{(\infty)}_1$ behaves as $\P^{(\infty)}_1(z) = e^{kz} + \bigO(e^{(k-1)z})$, and as $e^z \to 0$, $\P^{(\infty)}_1$ behaves as $\P^{(\infty)}_1(z) = \bigO(1)$.
  \item
    \begin{align}
      \hat{P}^{(\infty)}(z) = {}& (\bigO(|z-a|^{-1/4}), \bigO(|z-a|^{-1/4})) & & \text{ as $z\to a$,} \\ 
      \hat{P}^{(\infty)}(z) = {}& (\bigO(|z-b|^{-1/4}), \bigO(|z-b|^{-1/4})) & & \text{ as $z\to b$.} 
    \end{align}
  \end{enumerate}
\end{enumerate}

Inspired by the construction of $P^{(\infty)}$ in Section \ref{subsec:RH_of_P^infty_and_F}, we search for $\hat P^{(\infty)}$ in the form  $\hat P^{(\infty)}(z) =  (\hat F(\Jinv_2(z)), \hat F(\Jinv_1(z)))$,
where $\Jinv_1$ and $\Jinv_2$ are, as before, the two inverses of
the map $\Jlike$ defined in \eqref{eq:defn_of_I_1} and \eqref{eq:defn_of_I_2}. Hence
\begin{equation} \label{def hatPinfty}
  \hat F(s)  \colonequals
  \begin{cases}
    \hat P^{(\infty)}_2(\Jlike(s)) & \textnormal{for $s \in \compC \setminus \bar{D}$,} \\
    \hat P^{(\infty)}_1(\Jlike(s)) & \textnormal{for $s \in D \setminus [-\frac{1}{2},\frac{1}{2}]$,}
  \end{cases}
\end{equation}
and like $F(s)$ in Section \ref{subsec:RH_of_P^infty_and_F}, $\hat{F}$ can be analytically continued to $[-\frac{1}{2}, \frac{1}{2})$. At $\frac{1}{2}$, $\hat F$ has a pole of order $k$ if $k>0$, a removable singularity if $k=0$ and a zero of multiplicity $-k$ if $k<0$. From the RH conditions for $\hat P^{(\infty)}$, we deduce the
following RH problem for $\hat F$.

\subsubsection*{RH problem for $\hat F$}

\begin{enumerate}[label=(\alph{*})]
\item $\hat F$ is analytic in
$\compC\setminus(\gamma_1\cup\gamma_2)$ if $k \leq 0$, and analytic in $\compC \setminus (\gamma_1 \cup \gamma_2 \cup \{ \frac{1}{2} \})$ if $k > 0$,
\item for $s\in\gamma_1\cup\gamma_2$, we have
\begin{align}
  \hat F_+(s) = {}& \hat F_-(s), & & \textnormal{for $s \in \gamma_1$,}
  \label{eq:1st_jump_condition_of_RHP_for_F(s)} \\
  \hat F_+(s) = {}& -\hat F_-(s), & & \textnormal{for $s \in \gamma_2$,}
  \label{eq:2nd_jump_condition_of_RHP_for_F(s)}
\end{align}
\item we have the asymptotic conditions
  \begin{align}
    \hat F(s) = {}& \bigO(s^{-(k+1)}), & &\textnormal{as $s \to \infty$,} \\
    \hat F(s) = {}& e^{k(\frac{1}{2}c_1 + c_0)} (s-\frac{1}{2})^{-k} + \bigO((s-\frac{1}{2})^{-k+1}), & & \textnormal{as $s \to \frac{1}{2}$}, \\
    \hat{F}(s) = {}& \bigO(\lvert s - s_a \rvert^{-\frac{1}{2}}), & & \text{as $s \to s_a$,} \\
    \hat{F}(s) = {}& \bigO(\lvert s - s_b \rvert^{-\frac{1}{2}}), & & \text{as $s \to s_b$.}
  \end{align}
\end{enumerate}
It is verified directly that
\begin{equation}\label{def hatF1}
  \hat F(s) =\frac{\sqrt{(\frac{1}{2}-s_a)(\frac{1}{2}-s_b)}}{\sqrt{(s-s_a)(s-s_b)}} e^{k(\frac{1}{2}c_1 + c_0)} (s-\frac{1}{2})^{-k},
  \quad \textnormal{for $s \in \compC
    \setminus \gamma_2$,}
\end{equation}
solves the above RH problem. Here we choose the branch of the square root $\sqrt{(z - s_a)(z - s_b)}$ that is analytic except on $\gamma_2$ and close to $z$ as $z\to\infty$,

\subsection{The convergence of $\hat P \to \hat P^{(\infty)}$} \label{sec:convergence-hatp-to}

Define analogous to $\F(s)$ in \eqref{def F2}
\begin{equation}\label{def hatF2}
  \hat\F(s)  \colonequals
  \begin{cases}
    \hat P_2(\Jlike(s)) & \textnormal{if $s \in \compC \setminus \bar{D}$ and $\Jlike(s) \notin \Sigma_P$,} \\
    \hat P_1(\Jlike(s)) & \textnormal{if $s \in D \setminus [-\frac{1}{2}, \frac{1}{2}]$ and $\Jlike(s) \notin \Sigma_P$.}
  \end{cases}
\end{equation}
We have the scalar jump conditions
\begin{align}
  \hat\F_+(s) = {}&  \hat\F_-(s),&& \text{for $s\in\gamma_1$}, \label{jump hatF1} \\
  \hat\F_+(s) = {}& -\hat\F_-(s),&&\text{for $s\in\gamma_2$}, \label{jump hatF2}
\end{align}
and the shifted jump conditions
\begin{align}
 \hat\F_+(s) = {}& J_{\P, 11}(\Jlike(s))
\hat\F_-(s) + J_{\P,21}(\Jlike(s)) \hat\F_-(\Jinv_1(\Jlike(s))),
&& \text{for $s\in\Sigma''$}, \label{jump hatF3} \\
  \hat\F_+(s) = {}& J_{\P, 12}(\Jlike(s)) \hat\F_-((\Jinv_2(\Jlike(s)))) +
  J_{\P,22}(\Jlike(s)) \hat\F_-(s),
  &&\text{for $s\in\Sigma'$}. \label{jump hatF4}
\end{align}
The asymptotic conditions are the same as those for $\hat{F}(s)$
\begin{align}
  \hat \F(s) = {}& \bigO(s^{-(k+1)}), & &\textnormal{as $s \to \infty$,} \\
  \hat \F(s) = {}& e^{k(\frac{1}{2}c_1 + c_0)} (s-\frac{1}{2})^{-k} + \bigO((s-\frac{1}{2})^{-k+1}), & & \textnormal{as $s \to \frac{1}{2}$}, \\
  \hat{\F}(s) = {}& \bigO(\lvert s-s_a \rvert^{-1/2}), & & \text{as $s\to s_a$,} \\
  \hat{\F}(s) = {}& \bigO(\lvert s-s_b \rvert^{-1/2}), & & \text{as $s\to s_b$.}
\end{align}
Next we define, analogous to $R(s)$ in \eqref{def R},
\begin{equation}\label{def hatR}
  \hat R(s)  \colonequals  \frac{\hat\F(s)}{\hat F(s)} , \quad \textnormal{for $s \in \compC \setminus \Jlike^{-1}(\Sigma_P)$.}
\end{equation}
Then like $R$, $\hat R$ is analytic at $\frac{1}{2}$ and across $(\gamma_1\cup\gamma_2)$, and satisfies the following shifted RH problem.

\subsubsection*{Shifted RH problem for $\hat R$}

\begin{enumerate}[label=(\alph{*})]
\item $\hat R$ is analytic in $\compC\setminus (\Sigma'\cup\Sigma'')$, where $\Sigma'$ and $\Sigma''$ are defined in \eqref{eq:defn_of_Sigma'},
\item $\hat R$ has the jump conditions
\begin{align}
 \hat R_+(s) = {}&  J_{\hat{R}, 11}(s) \hat R_-(s) +
 J_{\hat{R},21}(s) \hat R_-(\Jinv_1(\Jlike(s))),
&& \text{for $s\in\Sigma''$}, \label{jump  hatR1} \\
  \hat R_+(s) = {}& J_{\hat{R}, 12}(s) \hat R_-(\Jinv_2(\Jlike(s))) + J_{\hat{R},22}(s) \hat R_-(s),
  &&\text{for $s\in\Sigma'$}, \label{jump hatR2}
\end{align}
where
\begin{align}
  J_{\hat{R}, 11}(s) = {}& J_{\P, 11}(\Jlike(s)), & J_{\hat{R}, 21}(s) = {}& J_{\P,21}(\Jlike(s)) \frac{\hat F(\Jinv_1(\Jlike(s)))}{\hat F(s)}, \\
  J_{\hat{R}, 12}(s) = {}&  J_{\P,12}(\Jlike(s)) \frac{\hat F(\Jinv_2(\Jlike(s)))}{\hat F(s)}, & J_{\hat{R}, 22}(s) = {}& J_{\P, 22}(\Jlike(s)).
\end{align}
\item $\hat{R}$ is bounded, and $\hat R(s)= 1+ \bigO(s^{-1})$ as $s\to\infty$.
\end{enumerate}
As $n\to\infty$, we have the uniform asymptotic estimates analogous to \eqref{eq:J_R:1} and \eqref{eq:J_R:2}
\begin{align}
  J_{\hat{R}, 11}(s)= {}& 1+\bigO(n^{-1}), &  J_{\hat{R}, 21}(s)= {}& \bigO(n^{-1}), &&\text{for $s\in \Sigma''$}, \\
  J_{\hat{R}, 12}(s)= {}& \bigO(n^{-1}), &  J_{\hat{R}, 22}(s)= {}& 1+\bigO(n^{-1}), &&\text{for $s\in \Sigma'$},
\end{align}
These estimates imply, in a similar way as \eqref{eq:J_R:1} and \eqref{eq:J_R:2} do in Section
\ref{sec:convergence-p-to}, the uniform convergence of
$\hat R$ to $1$:
\begin{equation}\label{asymp hatR}
  \hat R(s)=1+\bigO(n^{-1}),\quad \text{as $n\to\infty$, for
    $s\in\mathbb C\setminus\left(\Sigma'\cup\Sigma''\right)$.}
\end{equation}
Hence, by \eqref{def hatR}, \eqref{def hatPinfty}, and \eqref{def hatF2}, we have, like \eqref{eq:uniform_conv_P_1} and \eqref{eq:uniform_conv_P_2},
\begin{align}
  \hat P_1(z) = {}& (1 + \bigO(n^{-1})) \hat{P}^{(\infty)}_1(z), \quad \text{as $n \to \infty$, for $z \in \cyld \setminus \Sigma_P$,} \label{eq:uniform_conv_hatP_1} \\
  \hat P_2(z) = {}& (1 + \bigO(n^{-1})) \hat{P}^{(\infty)}_2(z), \quad \text{as $n \to \infty$, for $z \in \compC \setminus \Sigma_P$.} \label{eq:uniform_conv_hatP_2}
\end{align}

\section{Proof of main results}\label{section: proof results}

In this section we collect the asymptotics of $p^{(n)}_{n+k}(z)$ and $q^{(n)}_{n+k}(e^z)$, from the analysis in Sections \ref{section: RH p} and \ref{section: RH q}. The goal is to prove Theorem \ref{theorem: asympt2}.

 \subsection{The asymptotics of $p^{(n)}_{n+k}(z)$}

The main task in the computation of the asymptotics for $p^{(n)}_{n+k}$ consists of the inversion of the transformations  $Y\mapsto T\mapsto S\mapsto P$. By \eqref{def Y}, \eqref{def T}, \eqref{eq:def S}, \eqref{def P} and the asymptotics obtained in Section \ref{subsec:RH_of_P^infty_and_F}, we will find the asymptotics of $p^{(n)}_{n+k}$. In Figure \ref{fig:Sigma_P} it is shown that the complex plane is divided into the outside region, upper and lower bulk regions and two edge regions by $\Sigma_P$. We restrict ourselves to the upper half plane because of symmetry, and do the computation in each of the four regions.

\paragraph{Outside region}

For $z$ in the outside region,
we have
\begin{equation} \label{eq:inverse_transformation_of_p_outside}
  p^{(n)}_{n+k}(z) = Y_{1}^{(n+k,n)}(z)=T_{1}(z)e^{n\gfn(z)} = S_{1}(z)e^{n\gfn(z)}=P_{1}(z)e^{n\gfn(z)}.
\end{equation}
By (\ref{eq:uniform_conv_P_1}) and (\ref{def Pinfty}),
\begin{equation} \label{eq:p^n_n+k_outside}
  p^{(n)}_{n+k}(z)= (1 + \bigO(n^{-1}))F(\Jinv_1(z))e^{n\gfn(z)},\quad \text{as $n\to\infty$},
\end{equation}
where $F$ is defined in \eqref{eq:defn_of_F(s)_by_f_and_g0}. Substituting the identity (\ref{eq:relation_between_G_and_F}) for $F$ into \eqref{eq:p^n_n+k_outside}, we have
\begin{equation}\label{asymp pn out}
  p^{(n)}_{n+k}(z)=(1+\bigO(n^{-1}))G_k(\Jinv_1(z))e^{n\gfn(z)},\quad
\text{as $n\to\infty$.}
\end{equation}
This proves \eqref{eq:asy_of_p_outside} for $z$ in the outside region.

\paragraph{Bulk region}

Similar to \eqref{eq:inverse_transformation_of_p_outside}--\eqref{asymp pn out}, for $z$ in the upper part of the lens but not in $\overline{U_a}$ and  $\overline{U_b}$, we obtain
\begin{equation} \label{eq:p_in_upper_lens}
  \begin{split}
    p^{(n)}_{n+k}(z) = {}& Y_{1}^{(n+k,n)}(z)=T_{1}(z)e^{n\gfn(z)} \\
    = {}& (S_1(z) + S_2(z) e^{-n\phi(z) + z}) e^{n\gfn(z)} \\
    = {}& P_1(z) e^{n\gfn(z)} + P_2(z) e^z e^{n(V(z) - \tilde{\gfn}(z) + \ell)} \\
    = {}& (1 + \bigO(n^{-1})) F(\Jinv_1(z))e^{n\gfn(z)} \\
    & + (1 + \bigO(n^{-1})) F(\Jinv_2(z))  e^z e^{n(V(z) - \tilde{\gfn}(z) + \ell)} \\
    = {}& (1 + \bigO(n^{-1})) G_k(\Jinv_1(z))e^{n\gfn(z)} \\
    & + (1 + \bigO(n^{-1})) G_k(\Jinv_2(z)) e^{n(V(z) - \tilde{\gfn}(z) + \ell)},
  \end{split}
\end{equation}
as $n\to\infty$.
In the last identity of \eqref{eq:p_in_upper_lens} we use (\ref{eq:relation_between_G_and_F}) and the identity $z =\Jlike(\Jinv_2(z))$. We obtain the formula \eqref{eq:asy_of_p^n_n+k_bulk_region} for $z$ in the upper bulk region.

In particular, if $x \in (a, b)$ and $z \to x$ from above, we have by \eqref{var eq g} that $V(x) - \tilde{\gfn}_+(x) + \ell = \gfn_-(x)$, and further from the definition \eqref{eq:expr_of_gfn_and_tilde_gfn-intro} of $\gfn(z)$, we have $\gfn_{\pm}(x) = \int \log \lvert x-y \rvert d\mu_V(y) \pm \pi i \int^b_x d\mu_V$. On the other hand, as $z \to x$ from above, by \eqref{eq:defn_of_I_+} and \eqref{eq:defn_of_I_-}, $\Jinv_1(z)$ and $\Jinv_2(z)$ converge to $\Jinv_+(x)$ and $\Jinv_-(x)$ respectively. Noting that $\Jinv_-(x) = \bar{\Jinv}_+(x)$, we have from \eqref{eq:p_in_upper_lens} and \eqref{eq:asy_of_p^n_n+k_bulk_region}
\begin{equation}\label{asymp pn in}
  p^{(n)}_{n+k}(x)=r_k(x)e^{n\int\log|x-y|d\mu_V(y)}\left[\cos \left( n\pi
    \int_{x}^bd\mu_V(t) +\theta_k(x)\right) + \bigO(n^{-1}) \right],
\end{equation}
where $r_k(x)$ and $\theta_k(x)$, as defined in \eqref{eq:defn_r_theta}, are the modulus and argument of $2G_k(\Jinv_+(x)) = 2c^k_1 \frac{(\Jinv_+(x) + \frac{1}{2}) (\Jinv_+(x) - \frac{1}{2})^k}{\sqrt{(\Jinv_+(x)-s_a)(\Jinv_+(x)-s_b)}}$.

\paragraph{Edge region}

For brevity we only consider the case $z \in U_b$, the case $z \in U_a$ can be treated similarly. As shown in Figure \ref{fig:Sigma_P}, the part of $U_b$ in the upper half plane is divided by the lens $\Sigma_S$ into two parts, one in the lens and one out of the lens. If $z \in U_b \cap \compC^+$ is outside the lens, we obtain
\begin{equation} \label{eq:in_U_b_out_lens}
  p^{(n)}_{n+k}(z) = Y_{1}^{(n+k,n)}(z)=T_{1}(z)e^{n\gfn(z)} = S_{1}(z) e^{n\gfn(z)},
\end{equation}
and by (\ref{def P}),
\begin{equation} \label{eq:S_1_and_S_2_in_P}
  \begin{split}
   (S_1, S_2) = {}& \sqrt{2} (P_1, P_2) e^{(\frac{\pi i}{4} - \frac{z}{2}) \sigma_3}
    \begin{pmatrix}
      1 & 1 \\
      -1 & 1
    \end{pmatrix}^{-1}
    (n^{2/3} f_b(z))^{\frac{1}{4} \sigma_3} P^{(b)}(z) \\
    = {}& \frac{1}{\sqrt{2}} (P_1, P_2)
    \begin{pmatrix}
      e^{\frac{\pi i}{4} - \frac{z}{2}} n^{\frac{1}{6}} f_b(z)^{\frac{1}{4}} & -e^{\frac{\pi i}{4} - \frac{z}{2}} n^{-\frac{1}{6}} f_b(z)^{-\frac{1}{4}} \\
      e^{-\frac{\pi i}{4} + \frac{z}{2}} n^{\frac{1}{6}} f_b(z)^{\frac{1}{4}} & e^{-\frac{\pi i}{4} + \frac{z}{2}} n^{-\frac{1}{6}} f_b(z)^{-\frac{1}{4}} \\
    \end{pmatrix}
    P^{(b)}(z).
  \end{split}
\end{equation}
By \eqref{Pb}, \eqref{def Pinfty} and \eqref{eq:uniform_conv_P_1}--\eqref{eq:uniform_conv_P_2}, we further obtain
\begin{equation} \label{eq:final_result_in_U_b_out_lens}
  \begin{split}
    p^{(n)}_{n+k}(z) = {}& \sqrt{\pi} \left[ \left(P_1(z) - iP_2(z) e^z \right) n^{\frac{1}{6}} f_b(z)^{\frac{1}{4}} \Ai(n^{\frac{2}{3}}f_b(z)) \right. \\
    &- \left. \left( P_1(z) + iP_2(z) e^z \right) n^{-\frac{1}{6}} f_b(z)^{-\frac{1}{4}} \Ai'(n^{\frac{2}{3}}f_b(z)) \right] e^{\frac{n}{2} (\gfn(z) - \tilde{\gfn}(z) + V(z) + \ell)} \\
 = {}& \sqrt{\pi} \left[ \left( \F(\Jinv_1(z)) - i\F(\Jinv_2(z)) e^{\Jlike(\Jinv_2(z))} \right)n^{\frac{1}{6}} f_b(z)^{\frac{1}{4}} \Ai(n^{\frac{2}{3}}f_b(z)) \right. \\
    & - \left. \left( \F(\Jinv_1(z)) + i\F(\Jinv_2(z)) e^{\Jlike(\Jinv_2(z))} \right) n^{-\frac{1}{6}} f_b(z)^{-\frac{1}{4}} \Ai'(n^{\frac{2}{3}}f_b(z)) \right] e^{\frac{n}{2} (\gfn(z) - \tilde{\gfn}(z) + V(z) + \ell)} \\
    = {}& \sqrt{\pi} \left[ \left( \G_k(\Jinv_1(z)) - i\G_k(\Jinv_2(z)) \right) n^{\frac{1}{6}} f_b(z)^{\frac{1}{4}} \Ai(n^{\frac{2}{3}}f_b(z)) \right. \\
    & - \left. \left( \G_k(\Jinv_1(z)) + i\G_k(\Jinv_2(z)) \right) n^{-\frac{1}{6}} f_b(z)^{-\frac{1}{4}} \Ai'(n^{\frac{2}{3}}f_b(z)) \right] e^{\frac{n}{2} (\gfn(z) - \tilde{\gfn}(z) + V(z) + \ell)},
  \end{split}
\end{equation}
where $\G_k$ is defined, analogous to the formula \eqref{eq:relation_between_G_and_F} for $G_k$, as
\begin{equation} \label{eq:defn_of_cal_G}
  \G_k(s)  \colonequals
  \begin{cases}
    \F(s) & \text{if $s \in \compC \setminus \bar{D}$ and $\Jlike(s) \notin \Sigma_P$,} \\
    e^{\Jlike(s)} \F(s) & \textnormal{if $s \in D \setminus [-\frac{1}{2}, \frac{1}{2}]$ and $\Jlike(s) \notin \Sigma_P$.}
  \end{cases}
\end{equation}
From \eqref{def R} and \eqref{asymp R}, we have that
\begin{equation} \label{eq:relation_between_G_and_G}
  \G_k(s) = G_k(s) (1 + \bigO(n^{-1})),\qquad\mbox{ as $n\to\infty$}.
\end{equation}
Hence we obtain \eqref{eq:asy_of_p^n_n+k_right_edge_region} for $z$ in the edge region $U_b$, upper half plane, and outside of the lens.

Let us now focus on the asymptotics of $p^{(n)}_{n+k}(z)$ for $z = b + f'_b(b)^{-1} n^{-2/3} t$ which is in the upper half plane and outside of the lens, where $t$ is bounded. Then
\begin{equation}
  \Ai(n^{2/3}f_b(z)) = \Ai(t) + \bigO(n^{-2/3}),\quad \text{as $n\to\infty$,}
\end{equation}
Direct computation yields, as $n\to\infty$, by \eqref{eq:Joukowsky_like_transform}--\eqref{eq:defn_of_s_a_and_s_b} and \eqref{eq:defn_of_I_1}--\eqref{eq:defn_of_I_2},
\begin{align}
  \Jinv_1(z) = {}& s_b + \frac{(s_b + \frac{1}{2})(s_b - \frac{1}{2})}{\sqrt{s_b}} f'_b(b)^{-1/2} n^{-1/3} \sqrt{t} (1 + \bigO(n^{-2/3})t), \label{Jinv_1_at_b} \\
  \Jinv_2(z) = {}& s_b - \frac{(s_b + \frac{1}{2})(s_b - \frac{1}{2})}{\sqrt{s_b}} f'_b(b)^{-1/2} n^{-1/3} \sqrt{t} (1 + \bigO(n^{-2/3})t), \label{Jinv_2_at_b}
\end{align}
and that as $s \to s_b$, by \eqref{eq:relation_between_G_and_G} and \eqref{eq:defn_of_G_and_G_hat},
\begin{equation}
  \G_k(s) = (1 + \bigO(n^{-1})) \left( 2^{-\frac{1}{2}} (\frac{1}{4} + \frac{1}{c_1})^{-\frac{1}{4}} c^{k - 1}_1 (\sqrt{\frac{1}{4} + \frac{1}{c_1}} - \frac{1}{2})^{k - 1} + \bigO(s - s_b) \right) \frac{1}{\sqrt{s - s_b}},
\end{equation}
where all square roots take the principal value. Hence when $t$ is bounded
\begin{align}
  \G_k(\Jinv_1(z)) - i\G_k(\Jinv_2(z)) = {}& \left( \frac{1}{2} \left( \frac{1}{4} + \frac{1}{c_1} \right)^{-\frac{1}{8}} \left( \sqrt{\frac{1}{4} + \frac{1}{c_1}} - \frac{1}{2} \right)^{k-1} c^{k-\frac{1}{2}}_1 + \bigO(n^{-\frac{1}{3}}) \right) n^{\frac{1}{6}} t^{-\frac{1}{4}}, \label{eq:G-iG} \\
  \G_k(\Jinv_1(z)) + i\G_k(\Jinv_2(z)) = {}& \bigO(n^{-\frac{1}{6}}) t^{\frac{1}{4}}. \label{eq:G+iG}
\end{align}
Substituting \eqref{eq:G-iG} and \eqref{eq:G+iG} into \eqref{eq:final_result_in_U_b_out_lens} and noting that $f_b(b) = 0$ and $f'_b(b) > 0$, we obtain \eqref{eq:asy_of_p^n_n+k_very_near_b} for $z$ outside of the lens.

If $z \in U_b \cap \compC^+$ and inside the lens, then like \eqref{eq:p_in_upper_lens},
\begin{equation}
  p^{(n)}_{n+k}(z) = (S_1(z) + S_2(z) e^{-n\phi(z) + z}) e^{n\gfn(z)},
\end{equation}
and like \eqref{eq:S_1_and_S_2_in_P},
\begin{multline}
  (S_1 + S_2 e^{-n\phi(z) + z}, S_2) = \\
  \sqrt{2} (P_1, P_2) e^{(\frac{\pi i}{4} - \frac{z}{2}) \sigma_3}
  \begin{pmatrix}
    1 & 1 \\
    -1 & 1
  \end{pmatrix}^{-1}
  (n^{2/3} f_b(z))^{\frac{1}{4} \sigma_3} P^{(b)}(z)
  \begin{pmatrix}
    1 & 0 \\
    e^{-n\phi(z) + z} & 1
  \end{pmatrix}.
\end{multline}
Hence by \eqref{Pb}, \eqref{def Pinfty} and \eqref{eq:uniform_conv_P_1}--\eqref{eq:uniform_conv_P_2}, and using the identity $\Ai(x) + \omega\Ai(\omega x) + \omega^2\Ai(\omega^2 x) = 0$, we find that the result in \eqref{eq:final_result_in_U_b_out_lens} still holds, and so do the subsequent asymptotic formulas \eqref{eq:defn_of_cal_G}--\eqref{eq:G+iG}. Thus we can still prove \eqref{eq:asy_of_p^n_n+k_right_edge_region} and \eqref{eq:asy_of_p^n_n+k_very_near_b}.

 \subsection{The asymptotics of $q^{(n)}_{n+k}(e^z)$}

The derivation of the asymptotics for $q^{(n)}_{n+k}(e^z)$ is similar, and we
need to invert the transformations $X\mapsto \hat T\mapsto \hat
S\mapsto \hat P$ using (\ref{def X}), (\ref{def hatT}), (\ref{def
hatS}), and (\ref{def hatP}). For brevity, we only consider the outside region and the bulk region.

\paragraph{Outside region}

If $z$ is in the upper half plane and not in the lens or $U_a$, $U_b$, we have
\begin{equation}
  q^{(n)}_{n+k}(e^z) = X_{1}^{(n+k,n)}(z) = \hat T_{1}(z)e^{n\tilde\gfn(z)} = \hat S_{1}(z)e^{n\tilde\gfn(z)}=\hat P_{1}(z)e^{n\tilde\gfn(z)}.
\end{equation}
By (\ref{def hatF2}) and (\ref{def hatR}), we find similar to \eqref{eq:inverse_transformation_of_p_outside}
\begin{equation}
  q^{(n)}_{n+k}(e^z)=\hat\F(\Jinv_2(z))e^{n\tilde\gfn(z)}=\hat R(\Jinv_2(z))\hat
F(\Jinv_2(z))e^{n\tilde\gfn(z)}.
\end{equation}
By the formula (\ref{def hatF1}) for $\hat{F}$ and the asymptotic formula \eqref{asymp hatR} for $\hat{R}$,
this yields
\begin{equation}\label{asymp qn out}
  q^{(n)}_{n+k}(e^z) = (1+\bigO(n^{-1})) \frac{\sqrt{(\frac{1}{2}-s_a)(\frac{1}{2} - s_b)}}
  {\sqrt{(\Jinv_2(z)-s_a)(\Jinv_2(z)-s_b)}} e^{k(\frac{1}{2}c_1 + c_0)} (\Jinv_2(z) - \frac{1}{2})^{-k} e^{n\tilde\gfn(z)},\quad\text{
    as $n\to\infty$.}
\end{equation}
In \eqref{asymp qn out} and later in \eqref{eq:bulk_asy_of_q}, $\sqrt{(z - s_a)(z - s_b)}$ is chosen to be close to $z$ as $z \to \infty$ and has branch cut along $\gamma_2$. Substituting $s_a$ and $s_b$ by \eqref{eq:defn_of_s_a_and_s_b}, we prove \eqref{eq:asy_of_q_outside} for $z$ in the outside region.

\paragraph{Bulk region}

Similar to \eqref{eq:p_in_upper_lens},
\begin{equation}
  q^{(n)}_{n+k}(e^z) = \hat T_{1}(z)e^{n\tilde\gfn(z)} = \hat S_{1}(z)e^{n\tilde\gfn(z)} + \hat{S}_2(z)e^{n(V(z) -\gfn(z) + \ell)} = \hat P_{1}(z)e^{n\tilde\gfn(z)} + \hat{P}_2(z)e^{n(V(z) -\gfn(z) + \ell)}.
\end{equation}
By \eqref{def hatF2}, \eqref{def hatR}, \eqref{def hatF1} and \eqref{asymp hatR}, we find that as $n \to  \infty$,
\begin{equation} \label{eq:bulk_asy_of_q}
  \begin{split}
  q^{(n)}_{n+k}(e^z) ={}& \hat\F(\Jinv_2(z))e^{n\tilde\gfn(z)} + \hat\F(\Jinv_1(z))e^{n(V(z) -\gfn(z) + \ell)} \\
  ={}& \hat R(\Jinv_2(z))\hat
  F(\Jinv_2(z))e^{n\tilde\gfn(z)} + \hat R(\Jinv_1(z)) \hat F(\Jinv_1(z))e^{n(V(z) -\gfn(z) + \ell)} \\
  ={}& (1+\bigO(n^{-1})) \frac{\sqrt{(\frac{1}{2}-s_a)(\frac{1}{2}-s_b)}}
    {\sqrt{(\Jinv_2(z)-s_a)(\Jinv_2(z)-s_b)}} e^{k(\frac{1}{2}c_1 + c_0)} (\Jinv_2(z) - \frac{1}{2})^{-k} e^{n\tilde\gfn(z)} \\
  & + (1+\bigO(n^{-1})) \frac{\sqrt{(\frac{1}{2}-s_a)(\frac{1}{2}-s_b)}}
    {\sqrt{(\Jinv_1(z)-s_a)(\Jinv_1(z)-s_b)}} e^{k(\frac{1}{2}c_1 + c_0)} (\Jinv_1(z) - \frac{1}{2})^{-k} e^{n(V(z) -\gfn(z) + \ell)}.
  \end{split}
\end{equation} Substituting $s_a$ and $s_b$ by \eqref{eq:defn_of_s_a_and_s_b}, we prove \eqref{eq:asy_of_q^n_n+k_bulk_region} for $z$ in the upper bulk region.

As $z \to x \in \realR$ from above, noting that $V(x) - {\gfn}_-(x) + \ell = \tilde\gfn_+(x)$ by \eqref{var eq g}, $\Jinv_2(z) \to \Jinv_-(x)$, and $\Jinv_1(z) \to \Jinv_+(x)$, and using the identities $\Jinv_-(x) = \bar{\Jinv}_+(x)$ and $\tilde\gfn_{\pm}(x) = \int \log \lvert e^x-e^y \rvert d\mu_V(y) \pm \pi i \int^b_x d\mu_V$, we have like \eqref{asymp pn in},
\begin{equation}
  q^{(n)}_{n+k}(e^x) = \hat{r}_k(x) e^{n\int \log \lvert e^x - e^y \rvert d\mu_V(y)} \left[\cos \left( n\pi \int^b_x d\mu_V(t) + \hat{\theta}_k(x)\right) + \bigO(n^{-1}) \right],
\end{equation}
where $\hat{r}_k(x)$ and $\hat{\theta}_k(x)$, as defined in \eqref{eq:defn_r_theta_hat}, are the modulus and argument of $2 \hat{G}_k(\Jinv_-(x)) = 2\sqrt{\frac{(\frac{1}{2}-s_a)(\frac{1}{2}-s_b)}{(\Jinv_-(x)-s_a)(\Jinv_-(x)-s_b)}} e^{k(\frac{1}{2}c_1 + c_0)} (\Jinv_-(x) - \frac{1}{2})^{-k}$.

\subsection{Proof of Theorem \ref{theorem: asympt2}}

The asymptotic results obtained in the last two subsections nearly prove items \ref{enu:theorem: asympt2:A}, \ref{enu:theorem: asympt2:B} and part of \ref{enu:theorem: asympt2:C} and \ref{enu:theorem: asympt2:D} of Theorem \ref{theorem: asympt2}. However, in the statement of the theorem, the regions where asymptotic formulas are given, are $A_{\delta}$, $B_{\delta}$, $C_{\delta}$, and $D_{\delta}$, which are similar but not exactly equal to the outside, upper bulk, left edge and right edge regions that depend on $\Sigma_P$. We observe that if $\delta$ is a fixed small enough number, we can take the radius of $U_a$ and $U_b$ large enough so that they cover $C_{\delta}$ and $D_{\delta}$. On the other hand, we can also take the radius of $U_a$ and $U_b$ small enough and the shape of the lens thick enough to let the upper bulk region cover $B_{\delta}$, and we can take the radius of $U_a$ and $U_b$ small enough and the shape of the lens thin enough to let the outside region cover $A_{\delta}$. In this way, by using different contours $\Sigma_P$, the asymptotic results in the outside, upper bulk, left edge and right edge regions are translated into results in regions $A_{\delta}$, $B_{\delta}$, $C_{\delta}$, and $D_{\delta}$ respectively.

Although we have not proved all the asymptotic formulas in items \ref{enu:theorem: asympt2:C} and \ref{enu:theorem: asympt2:D} of Theorem \ref{theorem: asympt2}, the remainders can be proved using the method presented in the previous two subsections, and we omit the details.

To compute $h^{(n)}_{n+k}$ and prove item \ref{enu:theorem: asympt2:5} of Theorem \ref{theorem: asympt2}, we note that it appears in the leading coefficient of $\Clike p_{n + k}(z)$, see \eqref{eq:kappa_is_leading_coeff_of_Clike_p}. Using \eqref{def Y}, \eqref{def T}, \eqref{eq:def S} and \eqref{def P}, we have, for $z$ in $\strip$, outside of the lens and away from $a$ and $b$,
\begin{equation} \label{eq:leading_term_in_Cp}
  \Clike p_{n + k}(z)=Y_{2}^{(n+k,n)}(z)=e^{n\ell} e^{-n\tilde{\gfn}(z)}T_2(z)=e^{n\ell} e^{-n\tilde{\gfn}(z)}S_2(z)=e^{n\ell} e^{-n\tilde{\gfn}(z)}P_2(z).
\end{equation}
Since $\tilde{\gfn}(z) = z + \bigO(e^{-z})$ as $z \to +\infty$ in $\strip$, \eqref{eq:leading_term_in_Cp} yields
\begin{equation} \label{eq:leading_coeff_of_P_2_is_related_to_kappa}
  P_2(z) = \frac{-h^{(n)}_{n+k}}{2\pi i} e^{-n\ell} e^{-(k+1)z} + \bigO(e^{-(k+2)z}), \quad \text{as $z \to +\infty$.}
\end{equation}
By \eqref{eq:uniform_conv_hatP_2}, \eqref{def P F}, and \eqref{def F1}, we have as $n\to\infty$,
\begin{equation} \label{eq:leading_coeff_of_P_2_algebraic}
  \begin{split}
    \lim_{z \to +\infty} P_2(z) e^{(k+1)z} ={}& \lim_{z \to +\infty} P_2^{(\infty)}(z) e^{(k+1)z}(1+\bigO(n^{-1}))\\ ={}& \lim_{z \to +\infty} F(\Jinv_2(z)) e^{(k+1)z}(1 + \bigO(n^{-1})) \\
    = {}& \lim_{z \to +\infty} c^k_1 \left( \Jinv_2(z) - \frac{1}{2} \right)^{k+1} \frac{e^{-c_1 \Jinv_2(z) - c_0}}{\sqrt{(\Jinv_2(z) - s_a)(\Jinv_2(z) - s_b)}} e^{(k+1)z} (1 + \bigO(n^{-1})).
  \end{split}
\end{equation}
From the formula \eqref{eq:Joukowsky_like_transform} of $\Jlike(s)$ which is the inverse function of $\Jinv_2(z)$, we have
\begin{equation}
  \Jinv_2(z) = \frac{1}{2} + e^{\frac{c_1}{2} + c_0} e^{-z} + \bigO(e^{-2z}), \quad \text{as $z \to +\infty$,}
\end{equation}
and we obtain that
\begin{equation} \label{eq:leading_coeff_of_P_2}
  \begin{split}
    \lim_{z \to +\infty} c^k_1 \left( \Jinv_2(z) - \frac{1}{2} \right)^{k+1} \frac{e^{-c_1 \Jinv_2(z) - c_0}}{\sqrt{(\Jinv_2(z) - s_a)(\Jinv_2(z) - s_b)}} e^{(k+1)z} ={}& c^k_1 e^{k(\frac{c_1}{2} + c_0)} \frac{i}{\sqrt{(\frac{1}{2} - s_a)(s_b - \frac{1}{2})}} \\
    ={}& ic^{k+\frac{1}{2}}_1 e^{k(\frac{c_1}{2} + c_0)},
  \end{split}
\end{equation}
where $s_a$ and $s_b$ are expressed in $c_1$ by \eqref{eq:defn_of_s_a_and_s_b}. Formulas \eqref{eq:leading_coeff_of_P_2}, \eqref{eq:leading_coeff_of_P_2_algebraic} and \eqref{eq:leading_coeff_of_P_2_is_related_to_kappa} yield Theorem \ref{theorem: asympt2}\ref{enu:theorem: asympt2:5}.

\appendix

\section{Proofs of several technical results}\label{section: proofs
lemmas}

\subsection{Proof of Proposition \ref{propq}} \label{subsec:proof_of_algebraic_prop}

Our proof is similar to that of \cite[Proposition 2.1]{Bleher-Kuijlaars04a}. By the formula of the probability density function \eqref{jpdf}, the average of $\prod_{j=1}^n(e^z-e^{\lambda_j})$ can be expressed as
\begin{equation} \label{eq:explicit_formula_of_2nd_average}
  \mathbb E_n'(\prod_{j=1}^n(e^z-e^{\lambda_j}))=\frac{1}{{Z}_n'}\int_{\mathbb R^n}
\prod_{j=1}^n (e^z-e^{\lambda_j})\prod_{i<j} (\lambda_j - \lambda_i) \
  \prod_{i<j} (e^{\lambda_j}-e^{\lambda_i})\ \prod_{j=1}^n
  e^{-nV(\lambda_j)}\  d\lambda_j.
\end{equation}
From this formula, it is clear that $\mathbb E_n'(\prod_{j=1}^n(e^z-e^{\lambda_j}))$ is a linear combination of $e^{kz}$ with $k = 0, 1, \dotsc, n$, and that the coefficient of $e^{nz}$ is equal to $1$ since the probability measure is normalized. To show that it is equal to $q^{(n)}_n(e^z)$, we only need to verify that it satisfies the orthogonality conditions \eqref{orthoI}, which characterize $q^{(n)}_n(e^z)$ uniquely.

Expanding the Vandermonde determinant over the symmetric group $S_n$ gives
\begin{equation} \label{eq:expansion_of_Vandermonde}
  \prod_{i<j} (\lambda_j-\lambda_i)=\det(\lambda_i^{j-1})_{i,j=1,\ldots, n}=\sum_{\sigma\in S_n}(-1)^\sigma \prod_{j=1}^n\lambda_j^{\sigma(j) - 1}.
\end{equation}
Substituting \eqref{eq:expansion_of_Vandermonde} into \eqref{eq:explicit_formula_of_2nd_average}, we obtain
\begin{equation} \label{eq:2nd_average_after_permutation}
  \begin{split}
    \mathbb E_n'(\prod_{j=1}^n(e^z-e^{\lambda_j})) ={}& \frac{1}{{Z}_n'} \int_{\mathbb R^n} \sum_{\sigma\in S_n}(-1)^\sigma \prod_{j=1}^n (e^z-e^{\lambda_j}) \prod_{i<j} (e^{\lambda_j}-e^{\lambda_i}) \prod_{j=1}^n\lambda_j^{\sigma(j) - 1} e^{-nV(\lambda_j)}\  d\lambda_j \\
    ={}& \frac{n!}{{Z}_n'} \int_{\mathbb R^n}\prod_{j=1}^n (e^z-e^{\lambda_j}) \prod_{i<j} (e^{\lambda_j}-e^{\lambda_i}) \prod_{j=1}^n\lambda_j^{j - 1} e^{-nV(\lambda_j)}\ d\lambda_j.
  \end{split}
\end{equation}
Substituting the identity
\begin{equation}
  \prod_{j=1}^n (e^z-e^{\lambda_j})
  \prod_{i<j} (e^{\lambda_j}-e^{\lambda_i})=\det
  \begin{pmatrix}
    1&e^{\lambda_1}&\ldots &e^{n\lambda_1}\\
    1&e^{\lambda_2}&\ldots &e^{n\lambda_2}\\
    \vdots&\vdots& &\vdots\\
    1&e^{\lambda_n}&\ldots &e^{n\lambda_n}\\
    1&e^{z}&\ldots &e^{nz}
  \end{pmatrix}
\end{equation}
into \eqref{eq:2nd_average_after_permutation}, we obtain after integrating with respect to $\lambda_i$ that
\begin{equation} \label{qdet}
  \mathbb E_n'(\prod_{j=1}^n(e^z-e^{\lambda_j}))=\frac{n!}{Z_n'}\det
  \begin{pmatrix}
    m_{00}&m_{01}&\ldots &m_{0n}\\
    m_{10}&m_{11}&\ldots &m_{1n}\\
    \vdots&\vdots& &\vdots\\
    m_{n-1, 0}&m_{n-1, 1}&\ldots &m_{n-1, n}\\
    1&e^{z}&\ldots &e^{nz}
  \end{pmatrix},
  \quad \text{where} \quad
  m_{jk}=\int_{\mathbb
R}\lambda^je^{k\lambda}e^{-nV(\lambda)}d\lambda.
\end{equation}
Then it is straightforward to verify that for $k = 0, 1, \dots, n-1$,
\begin{equation} \label{qdet2}
  \int_{\mathbb R}\mathbb E_n'(\prod_{j=1}^n(e^z-e^{\lambda_j}))z^k e^{-nV(z)}dz = \frac{n!}{{Z}_n'}\det
  \begin{pmatrix}
    m_{00}&m_{01}&\ldots &m_{0n}\\
    m_{10}&m_{11}&\ldots &m_{1n}\\
    \vdots&\vdots& &\vdots\\
    m_{n - 1, 0}&m_{n - 1, 1}&\ldots &m_{n-1, n}\\
    m_{k0}&m_{k1}&\ldots &m_{kn}
  \end{pmatrix}
  = 0.
\end{equation}
Thus we prove that $\mathbb E_n'(\prod_{j=1}^n(e^z-e^{\lambda_j}))$ satisfies the orthogonality condition \eqref{orthoI} that determines $q^{(n)}_n(e^z)$, and then it follows that $\mathbb E_n'(\prod_{j=1}^n(e^z-e^{\lambda_j})) = q^{(n)}_n(e^z)$.

\subsection{Proof of Proposition \ref{prop:Joukowsky_like}} \label{subsec:Proof_of_Prop_Joukowsky_like}

In this proof, we fix $c_1 \in \realR^+$ and $c_0 \in \realR$, and $\Jlike$ stands for $\Jlike_{c_1, c_0}$ such that $\Jlike(s) = c_1 s + c_0 - \log \frac{s - \frac{1}{2}}{s + \frac{1}{2}}$. Recall that $s_a = -\sqrt{\frac{1}{4} + \frac{1}{c_1}}, s_b = \sqrt{\frac{1}{4} + \frac{1}{c_1}}$ as in \eqref{eq:defn_of_s_a_and_s_b}, and $a = \Jlike(s_a), b = \Jlike(s_b)$ as in \eqref{eq:defining_formula_of_a_b}.

To prove part \ref{enu:prop:Joukowsky_like:a}, we show that the equation $\Jlike(s) = x$:
\begin{enumerate}[label=(\arabic{*})]
\item has a unique solution $s$ in the upper half plane $\compC^+ = \{ s = u + iv $ with $v > 0 \}$ if $x\in (a,b)$,
\item has no solution in $\compC^+$ if $x\in\mathbb R\setminus (a, b)$.
\end{enumerate}
Moreover, as $x$ runs from $a$ to $b$, the solutions $s=s(x)$ form an arc in $\compC^+$ from $s_a$ to $s_b$. Then this arc is the desired $\gamma_1$ in Proposition \ref{prop:Joukowsky_like}, and the complex conjugate of $\gamma_1$ is the arc $\gamma_2$.

For $s = u + iv$ with $v>0$, $\Jlike(s) \in \realR$ if and only if the identity for its imaginary part
\begin{equation} \label{eq:argument_equation_for_gamma}
  c_1 v - \arccot \frac{u^2 + v^2 - \frac{1}{4}}{v} = 0
\end{equation}
is satisfied, where the range of $\arccot$ is $(0, \pi)$. It is a direct consequence of \eqref{eq:argument_equation_for_gamma} that $v < \frac{\pi}{c_1}$. Under the condition $0 < v < \frac{\pi}{c_1}$, \eqref{eq:argument_equation_for_gamma} is equivalent to
\begin{equation} \label{eq:simplified_eq_for_argument}
  u^2 = \frac{1}{4} + v\cot(c_1 v) - v^2.
\end{equation}
By direct calculation we find that the right-hand side of \eqref{eq:simplified_eq_for_argument} is a decreasing function in $v$ for $0<v<\frac{\pi}{c_1}$. Moreover, as $v\to 0$, it tends to $\frac{1}{4} + \frac{1}{c_1}$, and as $v\to \frac{\pi}{c_1}$, it tends to $-\infty$.

Thus for $\Jlike(s)$ to be real where $s = u + iv$ with $v > 0$, $u$ has to be in $(s_a, s_b)$, and for any $u$ in this interval there is a unique $v > 0$ to make \eqref{eq:simplified_eq_for_argument} hold. The locus of all such $s = u + iv$ is an arc in $\compC^+$ connecting $s_a$ and $s_b$. As a consequence of \eqref{eq:simplified_eq_for_argument}, $v$ increases as $u$ runs from $s_a$ to $0$, and then decreases as $u$ runs from $0$ to $s_b$. At any $s$ in this arc,
\begin{equation}
  \frac{d \Jlike(s)}{ds} = c_1 - \frac{1}{s^2 - \frac{1}{4}} \neq 0,
\end{equation}
and it follows that $\Jlike$ is a homeomorphism from this arc to the interval $[a,b]$, which proves part \ref{enu:prop:Joukowsky_like:a} of Proposition \ref{prop:Joukowsky_like}.

\medskip

Next we prove part \ref{enu:prop:Joukowsky_like:b}. It is easy to check that $\Jlike$ maps the ray $(s_b, \infty)$ to $(b, \infty)$ and the ray $(-\infty, s_a)$ to $(-\infty, a)$ homeomorphically. Then it suffices to show that $\Jlike$ is a univalent map from $\compC^+ \setminus \bar{D}$ onto $\compC^+$, and the univalent property of $\Jlike$ on $\compC^- \setminus \bar{D}$ follows by complex conjugation. To this end, we use the following elementary lemma:
\begin{lem}[Exercise 10 in Section 14.5 of \cite{Conway95}] \label{lem:Conway}
  Suppose that $G$ and $\Omega$ are simply connected Jordan regions and $f$ is a continuous function on the
  closure of $G$ such that $f$ is analytic on $G$ and $f(G) \subseteq \Omega$. If $f$ maps $\partial G$ homeomorphically
  onto $\partial \Omega$, then $f$ is univalent on $G$ and $f(G) = \Omega$.
\end{lem}
But this lemma is not directly applicable, since both $\compC^+ \setminus \bar{D}$ and $\compC^+$ are unbounded. Let $g(s)  \colonequals
-i\frac{s-i}{s+i}$ be the conformal map from the unit disk to the
upper half plane, we find that $g^{-1} \circ \Jlike \circ g$
is a map from the simply connected region $g^{-1}(\compC^+ \setminus
\bar{D})$ into the unit disk, and the map is homeomorphic on the
boundary. A direct application of Lemma \ref{lem:Conway} shows that
$g^{-1} \circ \Jlike \circ g$ is univalent in $g^{-1}(\compC^+
\setminus \bar{D})$ and onto the unit disk, hence $\Jlike$ is
univalent in $\compC^+ \setminus \bar{D}$ and onto the upper half
plane, and part \ref{enu:prop:Joukowsky_like:b} is proved.

To prove part \ref{enu:prop:Joukowsky_like:c}, we find by direct calculation that $\Jlike$ maps homeomorphically
\begin{enumerate}[label=(\arabic{*})]
\item the interval $(s_a, -\frac{1}{2})$ to the ray $(-\infty, a)$,
\item the interval $(\frac{1}{2}, s_b)$ to the ray $(b, \infty)$,
\item the upper side of the interval $(-\frac{1}{2}, \frac{1}{2})$ to the horizontal line $\mathbb R - \pi i$, and
\item the lower side of the interval $(-\frac{1}{2}, \frac{1}{2})$ to the horizontal line $\mathbb R + \pi i$.
 \end{enumerate}
 Then it suffices to show that $\Jlike$ maps $D \cap \compC^+$ onto $\strip \cap \compC^-$ univalently. We use Lemma \ref{lem:Conway} again. Similar to the conformal map $g$, we use the conformal map $h(s)  \colonequals  \log g(s) = \log \frac{-is-1}{s-i}$ that transforms the unit disk to $\strip \cap \compC^+$. We omit the details since the arguments are very similar to those in the proof of part \ref{enu:prop:Joukowsky_like:b}.

\subsection{Proof of Lemma \ref{lem:c_0_and_c_1}}

First, we show that for any $x_1 \in \realR^+$, \eqref{intro mod eq 2} has a unique solution as an equation in $x_0$. Note that
\begin{equation} \label{eq:derivative_wrt_x}
  \begin{split}
    \frac{d}{dx_0} \frac{1}{2\pi i} \oint_{\gamma} V' \left(\Jlike_{x_1, x_0}(s) \right)
    \frac{ds}{s-\frac{1}{2}}
    ={}& \frac{1}{2\pi i} \oint_{\gamma} V'' \left(\Jlike_{x_1, x_0}(s) \right) \frac{ds}{s-\frac{1}{2}}
    ={} \frac{-1}{\pi} \Im \int_{\gamma_1} V'' \left(\Jlike_{x_1, x_0}(s) \right) \frac{ds}{s-\frac{1}{2}} \\
    ={}& \frac{1}{\pi} \int^{\pi}_0 V'' \left(\Jlike_{x_1, x_0}(s(\theta)) \right) \Im \frac{d \log(s(\theta) - \frac{1}{2})}{d\theta} d\theta, 
  \end{split}
\end{equation}
where we parametrize $s \in \gamma_1$ by its argument $\theta$ that
runs from $0$ to $\pi$. This parametrization is well defined since as $s$ moves along $\gamma_1$, its imaginary part increases as its real part  increases from $s_a$ to $0$, and then decreases as its real part continues to increase from $0$ to $s_b$, as shown in the proof of Proposition \ref{prop:Joukowsky_like}.

Below we show that the right-hand side of \eqref{eq:derivative_wrt_x} is bounded below by a positive constant for all $x_0 \in \realR$. Since $V''(\Jlike_{x_1, x_0}(s(\theta)))$ is bounded below by a positive number by the strong convexity of $V$, we need only to prove for all $\theta \in (0, \pi)$, $\Im \log(s(\theta) - \frac{1}{2}) = \arg(s(\theta) - \frac{1}{2})$ is an increasing function.
We show the increasing for $\theta \in (\frac{\pi}{2}, \pi)$ and $\theta \in (0, \frac{\pi}{2})$ separately. For geometric reasons, when $\theta \in (\frac{\pi}{2}, \pi)$, $\arg(s(\theta)-\frac{1}{2})$ is increasing with $\theta$ since both
 $\Re s(\theta) < 0$ and $\Im s(\theta) > 0$ are decreasing. For $\theta \in (0, \frac{\pi}{2})$, we use the identity
\begin{equation}
  \Im \log(s(\theta) - \frac{1}{2}) = \Im \left( x_1 s(\theta) + x_0 + \log(s(\theta) + \frac{1}{2}) \right) - \Im \Jlike_{x_1, x_0}(s(\theta)).
\end{equation}
Here $\Im \Jlike_{x_1, x_0}(s(\theta))$, by the construction of $\gamma_1$, vanishes, $\Im s(\theta)$ increases as $\theta$ runs from $0$ to $\frac{\pi}{2}$ and for geometric reasons $\Im \log(s + \frac{1}{2})$ also increases as $\theta$ runs from $0$ to $\frac{\pi}{2}$. Thus we have that for $\theta \in (0, \frac{\pi}{2})$,  $\Im \log(s(\theta) - \frac{1}{2}) = x_1 \Im s(\theta) + \Im \log(s(\theta) + \frac{1}{2})$ is increasing.

Now we have that as a function in $x_0$,
$\frac{1}{2\pi i} \oint_{\gamma} V' (\Jlike_{x_1, x_0}(s) )/(s-\frac{1}{2}) ds$ is a bijection from $\mathbb R$ to $\mathbb R$, since its derivative is bounded below by a positive constant. Hence by continuity, there must
be a unique $x_0$ to make this function equal to $1$. Given $x_1 \in \realR^+$, we denote the unique $x_0$ that solves \eqref{intro mod eq 2} by $c_0(x_1)$. Similarly we can show that $c_0(x_1)$ is a continuous function in $x_1$.

Although we do not have a simple formula for $c_0(x_1)$, we show below that
\begin{align}
  \frac{1}{2\pi i} \oint_{\gamma} V' (\Jlike_{x_1, c_0(x_1)}(s)) ds < {}& x_1^{-1},& & \text{for $x_1$ sufficiently small,}\\
  \frac{1}{2\pi i} \oint_{\gamma}  V'(\Jlike_{x_1, c_0(x_1)}(s)) ds > {}& x_1^{-1},& & \text{for $x_1$ sufficiently large.}
\end{align}
Hence, by continuity, it follows that there exists $c_1 \in \realR^+$ that, together with $c_0 = c_0(c_1)$, solves \eqref{intro mod eq 1}--\eqref{intro mod eq 2}.

As $x_1 \to 0^+$, from \eqref{eq:simplified_eq_for_argument}, it follows that the shape of
$\gamma$ is close to the circle with radius $x_1^{-1/2}$ and center $0$. Hence if we parametrize $s \in \gamma$ as before by its argument $\theta$, we have for $\theta \in [0, 2\pi)$,
\begin{equation}
  s(\theta)=e^{i\theta}x_1^{-1/2}+o(1), \quad \lim_{x_1\to 0^+}\frac{s'(\theta)}{s(\theta)-\frac{1}{2}}=i, \quad \lim_{x_1\to 0^+} V' (\Jlike_{x_1, x_0}(s(\theta))) = V'(x_0).
\end{equation}
By the dominated convergence theorem, we have
\begin{equation} \label{eq:first_large_circle}
  \begin{split}
    \lim_{x_1\to 0^+}\frac{1}{2\pi i} \oint_{\gamma} V' (\Jlike_{x_1, x_0}(s)) \frac{ds}{s-\frac{1}{2}} = {}& \frac{1}{2\pi i}\lim_{x_1\to 0^+}\int_0^{2\pi}V'(\Jlike_{x_1, x_0}(s(\theta))) \frac{s'(\theta)}{s(\theta)-\frac{1}{2}}d\theta\\
    = {}& \frac{1}{2\pi i}\int_0^{2\pi}V'(x_0) i d\theta = V'(x_0).
\end{split}
\end{equation}
We find
 $\lim_{x_1\to 0^+}c_0(x_1)= \tilde{x}_0$, where $\tilde{x}_0$ is the unique value such that $V'(\tilde{x}_0) = 1$.
From the results obtained above, we have that
\begin{equation} \label{eq:estimate_of_2_1_small}
  \frac{1}{2\pi i} \oint_{\gamma} V' (\Jlike_{x_1, c_0(x_1)}(s)) ds = o(x_1^{-1}) \quad \text{as $x_1 \to 0^+$},
\end{equation}
since the shape of contour $\gamma$ approaches to the circle with radius $x_1^{-1/2}$, and the integrand tends
uniformly to $V'(\tilde{x}_0)=1$.

 On the other hand, for large values of $x_1$, we use the expression
\begin{equation} \label{eq:parametrization_of_Res_1}
  \frac{1}{2\pi i} \oint_{\gamma} V'(\Jlike_{x_1, x_0}(s)) ds = \frac{-1}{\pi} \Im \int_{\gamma_1} V' (\Jlike_{x_1, x_0}(s)) ds =
  \frac{-1}{\pi} \int^{s_b}_{s_a} V' (\Jlike_{x_1, x_0}(s(u))) \Im v'(u) du,
\end{equation}
where $s \in \gamma_1$ is expressed as a function in its real part $u = \Re s$, and  $v(u)>0$ is defined by the condition that $s(u)=u+iv(u)\in\gamma_1$, and $s_a, s_b$ are the two endpoints of $\gamma_1$, as denoted in the beginning of Appendix \ref{subsec:Proof_of_Prop_Joukowsky_like}, with the parameters $c_1, c_0$ substituted by $x_1, x_0$.
Let us decompose the integral at the right of \eqref{eq:parametrization_of_Res_1} as $I_1+I_2+I_3$, where
\begin{align}
I_1 = {}& \frac{-1}{\pi} \int^{-\frac{1}{2}}_{s_a} V' (\Jlike_{x_1, x_0}(s(u))) \Im v'(u) du,\\
I_2 = {}& \frac{-1}{\pi} \int_{-\frac{1}{2}}^{\frac{1}{2}} V' (\Jlike_{x_1, x_0}(s(u))) \Im v'(u) du,\\
I_3 = {}& \frac{-1}{\pi} \int_{\frac{1}{2}}^{s_b} V' (\Jlike_{x_1, x_0}(s(u))) \Im v'(u) du.
\end{align}
From \eqref{eq:simplified_eq_for_argument}, it is not difficult to find that as $x_1 \to \infty$,
\begin{equation} \label{eq:est_of_Im_s(theta_12)}
  v(-\frac{1}{2}) = v(\frac{1}{2}) = \frac{\pi}{2} x_1^{-1} + o(x_1^{-1}).
\end{equation}
We know that $V'$ is an increasing function in $u$ and that $v(u)$ is an even function.
From Appendix \ref{subsec:Proof_of_Prop_Joukowsky_like} we have that $v(u)$ is increasing for $u \in (s_a, 0)$ and decreasing for $u \in (0, s_b)$. Hence the integral $I_2$ is positive.
Using the monotonicity of $V'$ and integration by parts for $I_1$ and $I_3$, we similarly obtain
\begin{equation}
  \begin{split}
    I_1 + I_3 > {}& \frac{1}{\pi} V'(\Jlike_{x_1, x_0}(s(\frac{1}{2}))) (v(\frac{1}{2}) - v(s_b)) - \frac{1}{\pi} V'(\Jlike_{x_1, x_0}(s(-\frac{1}{2}))) (v(-\frac{1}{2}) - v(s_a)) \\
    ={}& \frac{1}{\pi} \left( V'(\Jlike_{x_1, x_0}(s(\frac{1}{2}))) - V'(\Jlike_{x_1, x_0}(s(-\frac{1}{2}))) \right) v(\frac{1}{2}),
  \end{split}
\end{equation}
where in the last line we used the identities $v(s_a) = v(s_b) = 0$ and $v(-\frac{1}{2}) = v(\frac{1}{2})$. Hence \eqref{eq:parametrization_of_Res_1} and the estimates of $I_2$ and $I_1 + I_3$ above imply that
\begin{equation}\label{res1}
  \frac{1}{2\pi i} \oint_{\gamma} V'(\Jlike_{x_1, x_0}(s)) ds > \frac{1}{\pi}\left( V' (\Jlike_{x_1, x_0}(s(\frac{1}{2})) - V' (\Jlike_{x_1, x_0}(s(-\frac{1}{2})) \right) v(\frac{1}{2}).
\end{equation}
As $x_1 \to \infty$,
\begin{equation} \label{eq:estimate_of_J(pm_half)}
  \Jlike_{x_1, x_0}(s(-\frac{1}{2})) = x_0 - \frac{x_1}{2} + o(x_1), \quad \Jlike_{x_1, x_0}(s(\frac{1}{2})) = x_0 + \frac{x_1}{2} + o(x_1),
\end{equation}
where the two $o(x_1)$ terms are independent to $x_0$. By \eqref{eq:estimate_of_J(pm_half)} and the assumption $V''(x) > c > 0$ for all $x$, we have that if $x_1$ is large enough, then uniformly for all $x_0 \in \realR$
\begin{equation} \label{eq:difference_of_V'_at_two_s}
  V' (\Jlike_{x_1, x_0}(s(\frac{1}{2}))) - V' (\Jlike_{x_1, x_0}(s(-\frac{1}{2})) > cx_1.
\end{equation}
Substituting \eqref{eq:difference_of_V'_at_two_s} and \eqref{eq:est_of_Im_s(theta_12)} into \eqref{res1}, we have that as $x_1 \to \infty$ and $x_0 = c_0(x_1)$,
\begin{equation} \label{eq:estimate_of_2_1_large}
  \frac{1}{2\pi i} \oint_{\gamma} V'(\Jlike_{x_1, c_0(x_1)}(s)) ds \gg x_1^{-1}.
\end{equation}
We note that $\frac{1}{2\pi i} \oint_{\gamma} V'(\Jlike_{x_1,
  c_0(x_1)}(s)) ds$ is continuous in $x_1$, since $\frac{1}{2\pi i}
\oint_{\gamma} V'(\Jlike_{x_1, x_0}(s)) ds$ is continuous in $x_1,
x_0$ and $c_0(x)$ is continuous. Then we find that the estimates \eqref{eq:estimate_of_2_1_small} and \eqref{eq:estimate_of_2_1_large} imply that there is a pair $(c_1, c_0 = c_0(c_1))$ such that both \eqref{intro mod eq 2} and \eqref{intro mod eq 1} are satisfied.

\section{Explicit construction of the equilibrium measure for quadratic and quartic
$V$}
\label{section: equi examples}

In this appendix we use the method developed in Section \ref{section:equi} to find the endpoints of the support of the equilibrium measure explicitly for quadratic and quartic external fields $V$. In the quadratic case, we consider a monomial external field $V(x)=\frac{x^2}{t}$, but the same method can be applied to all quadratic $V$. We also construct the density function of the equilibrium measure. In the quartic case, we confine our attention to $V$ such that $V(x) - \frac{x}{2}$ is an even function. Under this condition the equilibrium measure is symmetric around the origin. In contrast to the quadratic $V$ that is automatically convex, we also consider quartic $V$ that is one-cut but not convex.

\subsubsection*{External field $V(x) = \frac{x^2}{2t}$}

In this case, $V'(x) = \frac{x}{t}$, and a simple calculation of residue yields
\begin{equation}
  \frac{1}{2\pi i} \oint_{\gamma} V' \left(c_1 s + c_0 - \log \frac{s - \frac{1}{2}}{s + \frac{1}{2}} \right) ds = \frac{1}{t}, \quad \frac{1}{2\pi i} \oint_{\gamma} \frac{V' \left(c_1 s + c_0 - \log \frac{s - \frac{1}{2}}{s + \frac{1}{2}} \right)}{s-\frac{1}{2}} ds = \frac{c_0}{t}+\frac{c_1}{2t}.
\end{equation}
Thus by Lemma \ref{lem:c_0_and_c_1}, we have
\begin{equation}
  c_1 = t, \quad c_0 = \frac{t}{2}.
\end{equation}
The support of the equilibrium measure, as expressed by \eqref{eq:defining_formula_of_a_b}, is
\begin{equation}
  \begin{split}
    [a,b] ={}& [\Jlike_{t, \frac{t}{2}}(s_a), \Jlike_{t, \frac{t}{2}}(s_a)] \\
    ={}& \left[\frac{1}{2}(t-\sqrt{t^2+4t}) - \log\frac{t+2+\sqrt{t^2+4t}}{2}, \frac{1}{2}(t+\sqrt{t^2+4t}) - \log\frac{t+2-\sqrt{t^2+4t}}{2}\right].
    \end{split}
\end{equation}
In particular, for $t=1$, we have
\begin{equation} \label{eq:support_of_eq_measure_standard_quadrartic}
  [a,b] = \left[\frac{-\sqrt{5}+1}{2} - \log\frac{3+\sqrt{5}}{2}, \frac{\sqrt{5}+1}{2} - \log\frac{3-\sqrt{5}}{2}
  \right].
\end{equation}
To find the equilibrium density, we have as a particular case of \eqref{formula M} that
\begin{equation}\label{M quadr}
  M(s) =
  \begin{cases}
    \frac{-1}{t} \log \frac{s - \frac{1}{2}}{s + \frac{1}{2}}, & \text{for $s \in \compC \setminus \bar{D}$,} \\
    s + \frac{1}{2}, & \text{for $s \in D$.}
  \end{cases}
\end{equation}
Then by \eqref{psiM}, after a straightforward calculation, we obtain the following expression
\begin{equation} \label{eq:formula_of_Psi-quadratic}
  \psi_V(x) =\frac{1}{\pi}\Im \Jinv_+(x),
\end{equation}
where $\Jinv_+$ is as before the boundary value of the inverse of
$\Jlike = \Jlike_{t, \frac{t}{2}}$ which parametrizes the curve $\gamma_1$.

\subsubsection*{External field $V(x) = x^4/4 + ux^2/2 + x/2$}

In this case, $V'(x) = x^3 + ux + \frac{1}{2}$, and the calculation of residues yields
\begin{gather}
  \frac{1}{2\pi i} \oint_{\gamma} V' \left(c_1 s + c_0 - \log \frac{s - \frac{1}{2}}{s + \frac{1}{2}} \right) ds = \frac{c_1^2}{4} + 3c_1 + 3c_0^2 + u,\label{quartic eq1} \\
  \frac{1}{2\pi i} \oint_{\gamma} \frac{V' \left(c_1 s + c_0 - \log \frac{s - \frac{1}{2}}{s + \frac{1}{2}} \right)}{s-\frac{1}{2}} ds = \frac{c_1^3}{8} + (\frac{3c_0}{4} + \frac{3}{2})c_1^2 + (\frac{3c_0^2}{2} + 6c_0 + \frac{u}{2})c_1 + c_0^3 + uc_0 + \frac{1}{2}.
\end{gather}
As a consequence of the relation $V(x)=V(-x)+x$, the equilibrium measure $\mu_V$ is symmetric around the origin. Indeed, changing variables $s\mapsto -s$ and $t\mapsto -t$ in the energy functional (\ref{energy}), it is straightforward to verify that $I_V(\mu_V)=I_V(\tilde\mu_V)$, where $\tilde\mu_V$ is defined by the fact that $\tilde\mu_V(A)=\mu_V(-A)$ for any Borel set $A$. From the uniqueness of the equilibrium measure, it follows that $\mu_V = \tilde\mu_V$. In particular this implies that the support of the equilibrium measure is of the form $[-b,b]$. By \eqref{eq:first_def_of_c_0}, we have $c_0=0$. Substituting this and \eqref{quartic eq1} into \eqref{intro mod eq 1}, we obtain the equation
\begin{equation} \label{eq:equation_satisfied_by_c_1_quartic}
  c_1^3 + 12c_1^2 + 4uc_1 - 4 = 0.
\end{equation}
\begin{rmk}
Although the equilibrium measure, which is the limiting mean eigenvalue distribution of the random matrix ensemble as $n\to\infty$, is symmetric around the origin, this is not true for the finite $n$ joint probability distribution of eigenvalues \eqref{jpdf0}. The latter would only be invariant under the change of variables $\lambda_i \to -\lambda_i$ if the term $x/2$ in $V(x)$ were replaced by $(\frac{1}{2} - \frac{1}{2n})x$.
\end{rmk}

For any value of $u$, the equation \eqref{eq:equation_satisfied_by_c_1_quartic} has a unique positive solution
by Descartes' rule of signs. We have an explicit formula for $c_1 \in \realR^+$ in
$u$ by the formula for the roots of a cubic equation, but we will
not write down the long formula. Together with $c_0 = 0$, $c_1>0$ gives us a solution to the pair of equations \eqref{intro mod eq 1} and \eqref{intro mod eq 2}. Under the condition that the equilibrium measure is one-cut supported, this pair $c_0, c_1$ yields expressions for the support and the density function of the equilibrium measure, but we omit the formulas.

We note that the external field $V$ is convex if $u \geq 0$. If $u$ is negative, it is not but the construction of the equilibrium measure given above can still be carried out formally. When $u$ is negative but sufficiently close to $0$, we can check that the equilibrium measure constructed in this way is still a probability measure. When $u$ is a large negative number, the constructed density function $\psi_V(x)$ is negative on an interval centered at $0$, and therefore not a probability density. This means that the external field is not one-cut regular, and our construction fails.
Based on the analogy with matrix models without external source, the symmetry of the equilibrium measure and numerical simulations, we conjecture that $V$ is one-cut regular for values of $u$ such that $\psi_V(0)>0$.

From \eqref{eq:equation_satisfied_by_c_1_quartic}, we derive that $u = \frac{1}{c_1} - 3c_1 - \frac{1}{4}c^2_1$, where $c_1$ is the positive solution to \eqref{eq:equation_satisfied_by_c_1_quartic}. This makes $u$ a strictly decreasing function of $c_1$.
Since $c_0 = 0$, it is easy to see that $\Jinv_+(0)$ is on the imaginary axis, and we denote it as $\Jinv_+(0) = i\, p$ ($p > 0$). From the relation
\begin{equation}
  \Jlike_{c_1, c_0}(\Jinv_+(0)) = c_1 \Jinv_+(0) - \log\frac{\Jinv_+(0) - \frac{1}{2}}{\Jinv_+(0) + \frac{1}{2}} = 0,
\end{equation}
we derive that $c_1 = \frac{2}{p} \arctan \frac{1}{2p}$ and $c_1$ is a strictly decreasing function in $p$, which means that $u$ is a strictly increasing function in $p$.

Like \eqref{M quadr}, with our quartic $V$ (using the fact that $c_0 = 0$), we have by \eqref{formula M}
\begin{equation}\label{M quart}
  M(s) =
  \begin{cases}
    -(3c^2_1 s^2 + u) \log \frac{s - \frac{1}{2}}{s + \frac{1}{2}} + 3c_1s
  \left( \log \frac{s - \frac{1}{2}}{s + \frac{1}{2}} \right)^2 - \left( \log \frac{s - \frac{1}{2}}{s + \frac{1}{2}}
  \right)^3 - 3c^2_1s & s \in \compC \setminus \bar{D}, \\
   c^3_1 s^3 + uc_1s + 3c^2_1s + \frac{1}{2} & s \in
   D.
  \end{cases}
\end{equation}
Similarly to the quadratic case, we can recover the equilibrium
density using (\ref{psiM}). In particular at zero we have
\begin{equation}
  \begin{split}
    \psi_V(0) ={}& \frac{1}{\pi}\Im M_-(\Jinv_+(0)) = \frac{1}{\pi} \Im \left( c^3_1\Jinv_+(0)^3 + (uc_1 + 3c^2_1)\Jinv_+(0) + \frac{1}{2} \right) \\
    ={}& \frac{1}{\pi} \left( -c^3_1 p^3 + (1 - \frac{c^3_1}{4}) p \right) \\
    ={}& \frac{p}{\pi} \left( 1 - c_1^3(p^2 + \frac{1}{4}) \right) = \frac{p}{\pi} \left( 1 - (\frac{8}{p} + \frac{2}{p^3}) (\arctan\frac{1}{2p})^3 \right).
\end{split}
\end{equation}
Here we used \eqref{eq:equation_satisfied_by_c_1_quartic} to pass from the first to the second line.
Thus $\psi_V(0) > 0$ if and only if  $(\frac{8}{p} + \frac{2}{p^3}) (\arctan\frac{1}{2p})^3 < 1$, which is equivalent to $p>p_c$ for some value $p_c>0$. Since $u$ is an increasing function in $p$, this is equivalent to $u > u_c$, where $u_c$ can be approximated numerically as $u_c \approx -1.9250$. Although we have not rigorously proved that for $u > u_c$ the external field is one-cut regular, numerical results are convincing. When $u = -1.925$, the constructed equilibrium measure is shown in Figure \ref{fig:quartic_critical_equilibrium}. It suggests that around $u = u_c \approx -1.925$ the transition between one-cut and two-cut equilibrium measures occurs.
\begin{figure}[ht]
  \centering
  \includegraphics[width=0.6\linewidth]{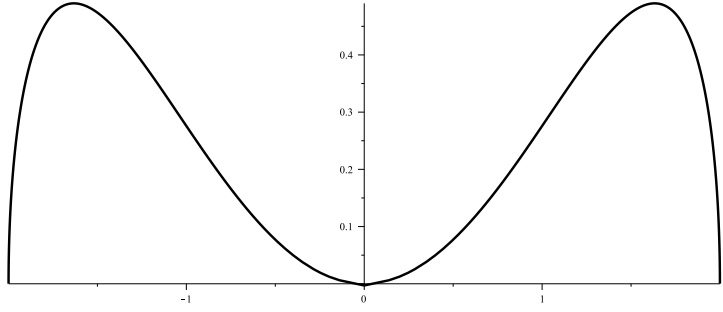}
  \caption{The density function $\psi_V$ of the equilibrium measure for the external field $V(x) = \frac{x^4}{4} - \frac{1.925 x^2}{2} + \frac{x}{2}$.}
  \label{fig:quartic_critical_equilibrium}
\end{figure}

\section{Asymptotics of $p^{(n)}_n(x)$ when $V(x) = \frac{x^2}{2}$} \label{sec:saddle_point_for_quadratic}

In this appendix, we give an alternative derivation of the asymptotic results in Theorem \ref{theorem: asympt2} when the external field is $V(x) = \frac{x^2}{2}$. The derivation is based on the contour integral formula of multiple Hermite polynomials in \cite[Theorems 2.1 and 2.3]{Bleher-Kuijlaars05}. This method can essentially reproduce all results in Theorem \ref{theorem: asympt2} for quadratic external field, but for brevity we only give the derivation for $p^{(n)}_n(x)$ where $x \in \realR$ and is away from the edges of the equilibrium measure. Although this contour integral method cannot be applied when the external field is not quadratic, it shows that the transformation $\Jlike_{c_1, c_0}$ arises naturally in the uniform external source model.

The result \cite[Theorem 2.1]{Bleher-Kuijlaars05} states that the monic polynomial $P_n(x)$ of degree $n$ that satisfies
\begin{equation}
  \int^{\infty}_{-\infty} P_n(x) e^{na_j x} e^{-n\frac{x^2}{2}} dx = 0, \quad \text{for $j = 1, 2, \dotsc, n$,}
\end{equation}
is expressed by an integral over the imaginary axis:
\begin{equation}
  P_n(x) = \frac{\sqrt{n}}{\sqrt{2 \pi} i} \int^{i \infty}_{-i \infty} e^{\frac{n}{2} (t - x)^2} \prod^n_{i = 1} (t - a_i) dt.
\end{equation}
When $a_j = \frac{j - 1}{n}$ as in \eqref{eq:a_j_are_equispaced}, we have, in our notations, $P_n=p_n^{(n)}$ where $V(x)=x^2/2$.
Setting $t = s + \frac{1}{2}$, we have
\begin{equation} \label{eq:contour_intgral_p^n_n}
  p^{(n)}_n(x) = \frac{\sqrt{n}}{\sqrt{2 \pi} i} \int^{i \infty}_{-i \infty} e^{n F_n(s; x)} ds,
\end{equation}
where
\begin{equation} \label{eq:F_n_added}
  F_n(s; x) = \frac{1}{2}(s + \frac{1}{2} - x)^2 + \sum^n_{i = 1} \log (s + \frac{1}{2} - \frac{i - 1}{n}) \frac{1}{n}.
\end{equation}
For $s$ away from the interval $[-\frac{1}{2}, \frac{1}{2}]$, we have the following uniform (in $s$ and $x$) asymptotic expansion as $n \to \infty$,
\begin{equation}
F_n(s; x) = F(s; x) + \frac{1}{n} \log \sqrt{\frac{s + \frac{1}{2}}{s - \frac{1}{2}}} + \bigO(\frac{1}{n^2}),
\end{equation}
where
\begin{equation} \label{eq:defn_of_F}
  F(s; x) = \frac{1}{2}(s + \frac{1}{2} - x)^2 + (s + \frac{1}{2}) \log(s + \frac{1}{2}) - (s - \frac{1}{2}) \log(s - \frac{1}{2}) - 1,
\end{equation}
and we take the principal branch of the logarithm and square root. Hence
\begin{equation}
  \frac{d}{ds} F(s; x) = s + \frac{1}{2} - \log \frac{s - \frac{1}{2}}{s + \frac{1}{2}} - x = \Jlike_{1, \frac{1}{2}}(s) - x.
\end{equation}
Below we consider the zeros $s$ of the derivative $\frac{d}{ds}F(s;x)$ and express them as functions in $x$. We use the functions $\Jinv_1(x), \Jinv_2(x)$ and their boundary values $\Jinv_\pm(x)$ as defined in \eqref{eq:defn_of_I_1}--\eqref{eq:defn_of_I_-} with $c_1 = 1$ and $c_0 = \frac{1}{2}$. Note that $s_a = -\frac{\sqrt{5}}{2}$ and $s_b = \frac{\sqrt{5}}{2}$ as given in \eqref{eq:defn_of_s_a_and_s_b}; we denote
\begin{equation}
  a = \Jlike_{1, \frac{1}{2}}(s_a)
  = \frac{-\sqrt{5} + 1}{2} + 2 \log  \frac{\sqrt{5} - 1}{2}, \quad
  b = \Jlike_{1, \frac{1}{2}}(s_b)
  = \frac{\sqrt{5} + 1}{2} + 2 \log  \frac{\sqrt{5} + 1}{2},
\end{equation}
as in \eqref{eq:defining_formula_of_a_b}, which agree with \eqref{eq:support_of_eq_measure_standard_quadrartic}.
We can say the following about the zeros of $\frac{d}{ds}F(s; x)$:
\begin{enumerate}[label=(\arabic{*})]
\item
  if $x > b$, then there are two zeros of $\frac{d}{ds}F(s; x)$: $\Jinv_1(x) \in (s_b, \infty)$ and $\Jinv_2(x) \in (\frac{1}{2}, s_b)$,
\item
  if $x < a$, then there are two zeros of $\frac{d}{ds}F(s; x)$: $\Jinv_1(x) \in (-\infty, s_a)$ and $\Jinv_2(x) \in (s_a, -\frac{1}{2})$,
\item
  if $x \in (a, b)$, then there are two zeros of $\frac{d}{ds}F(s; x)$: $\Jinv_+(x) \in \gamma_1$ and $\Jinv_-(x) \in \gamma_2$.
\end{enumerate}
By explicit computation, we find that for $x \in (-\infty,a) \cup (b,+\infty)$, then along the vertical line $\{ z = \Jinv_1(x) + it \mid t \in \realR \}$, $\Re F(z)$ attains its maximum at $z = \Jinv_1(x)$. If we deform the contour $i\mathbb R$ of integration in \eqref{eq:contour_intgral_p^n_n} to the vertical line through $\Jinv_1(x)$, the standard application of the saddle point method yields
\begin{equation} \label{eq:saddle_point_outside}
  \begin{split}
    p^{(n)}_n(x) ={}& \frac{\sqrt{n}}{\sqrt{2 \pi} i} \int^{\Jinv_1(x) + n^{-\frac{1}{3}} i}_{\Jinv_1(x) - n^{-\frac{1}{3}} i} e^{n F_n(s; x)} ds (1 + o(n^{-1})) \\
    ={}& \frac{\sqrt{n}}{\sqrt{2 \pi} i} \int^{\Jinv_1(x) + n^{-\frac{1}{3}} i}_{\Jinv_1(x) - n^{-\frac{1}{3}} i} e^{n F(s; x)} \sqrt{\frac{s + \frac{1}{2}}{s - \frac{1}{2}}} ds (1 + \bigO(n^{-1})) \\
    ={}& \frac{\sqrt{n} e^{n F_n(\Jinv_1(x); x)}}{\sqrt{2 \pi} i} \int^{\Jinv_1(x) + n^{-\frac{1}{3}} i}_{\Jinv_1(x) - n^{-\frac{1}{3}} i}
    \exp \left( \frac{n}{2}(s-\Jinv_1(x))^2 \left. \frac{d^2}{ds^2} F(s; x) \right\lvert_{s = \Jinv_1(x)} \right) \sqrt{\frac{s + \frac{1}{2}}{s - \frac{1}{2}}} ds (1 + \bigO(n^{-\frac{1}{2}})) \\
    ={}& e^{n F(\Jinv_1(x); x)} \frac{\Jinv_1(x) + \frac{1}{2}}{\sqrt{\Jinv_1(x)^2 - \frac{5}{4}}} (1 + \bigO(n^{-\frac{1}{2}})).
  \end{split}
\end{equation}

If $x \in (a, b)$, by explicit computation, we find that along the vertical line that passes through $\Jinv_+(x)$ and $\Jinv_-(x)$, $\Re F(z)$ attains its maximum at two points $z = \Jinv_+(x)$ and $z = \Jinv_-(x)$. (Note that although $F(z)$ is discontinuous on the interval $[-\frac{1}{2}, \frac{1}{2}]$, $\Re F(z)$ is continuous everywhere.) Then we take the contour in \eqref{eq:contour_intgral_p^n_n} as this vertical line. When the contour crosses the interval $[-\frac{1}{2}, \frac{1}{2}]$, $F(z)$ is no longer a good approximation of $F_n(z)$, but we can estimate the magnitude of $F_n(z)$ by other methods, (say, some rough and direct estimate of \eqref{eq:F_n_added}) and still find the vertical line suitable for saddle point analysis. The standard application of saddle point method yields, like \eqref{eq:saddle_point_outside},
\begin{equation}
  \frac{\sqrt{n}}{\sqrt{2 \pi} i} \int^{\Jinv_{\pm}(x) + n^{-\frac{1}{3}} i}_{\Jinv_{\pm}(x) - n^{-\frac{1}{3}} i} e^{n F_n(s; x)} ds = e^{n F(\Jinv_{\pm}(x); x)} \frac{\Jinv_{\pm}(x) + \frac{1}{2}}{\sqrt{\Jinv_{\pm}(x)^2 - \frac{5}{4}}} (1 + \bigO(n^{-\frac{1}{2}})),
\end{equation}
and
\begin{equation} \label{eq:saddle_point_bulk}
  \begin{split}
    p^{(n)}_n(x) ={}& \frac{\sqrt{n}}{\sqrt{2 \pi} i} \left( \int^{\Jinv_+(x) + n^{-\frac{1}{3}} i}_{\Jinv_+(x) - n^{-\frac{1}{3}} i} e^{n F_n(s; x)} ds + \int^{\Jinv_-(x) + n^{-\frac{1}{3}} i}_{\Jinv_-(x) - n^{-\frac{1}{3}} i} e^{n F_n(s; x)} ds \right)(1+o(n^{-1})) \\
    ={}& 2 e^{n \Re F(\Jinv_+(x); x)} \frac{\lvert \Jinv_+(x) + \frac{1}{2} \rvert}{\lvert \Jinv_+(x) - \frac{5}{4} \rvert^{\frac{1}{2}}} \left[\cos \left( n \Im F(\Jinv_+(x); x) + \arg \left( \frac{\Jinv_+(x) + \frac{1}{2}}{\sqrt{\Jinv_+(x)^2 - \frac{5}{4}}} \right)\right) + \bigO(n^{-\frac{1}{2}}) \right],
\end{split}
\end{equation}
where the square roots take the principal value. It is not obvious that the asymptotic formulas \eqref{eq:saddle_point_outside} and \eqref{eq:saddle_point_bulk} agree with the formulas \eqref{eq:asy_of_p_outside} and \eqref{eq:asy_of_p_bulk_real}. To convince the reader, we show that \eqref{eq:saddle_point_outside} is equivalent to \eqref{eq:asy_of_p_outside} (with $k = 0$ and $x \in \realR$) in the leading term.

It is easy to check that
\begin{equation} \label{eq:easy_relation_quadratic_V}
  \frac{\Jinv_1(x) + \frac{1}{2}}{\sqrt{\Jinv_1(x)^2 - \frac{5}{4}}} = G_0(\Jinv_1(z))
\end{equation}
where $G_0$ is the function defined in \eqref{eq:defn_of_G_and_G_hat} with $c_1 = 1$. We need also to show that $F(\Jinv_1(x); x) = \gfn(x)$ where $\gfn(x)$ is defined in \eqref{eq:expr_of_gfn_and_tilde_gfn-intro}. Since it is not hard to verify by direct computation that $\gfn(x) = \log(x) + o(1)$ and $F(\Jinv_1(x); x) = \log(x) + o(1)$,
we need only to show that the function $\gfn'(x) = G(x)$, defined in \eqref{eq:defn_of_G_and_tilde_G}, satisfies
\begin{equation} \label{eq:essential_relation_quadratic_V}
  G(x) = \frac{d}{dx} F(\Jinv_1(x); x),
\end{equation}
Note that by the relation $x = \Jlike_{1, \frac{1}{2}}(\Jinv_1(x))$, we have
\begin{equation}
  \begin{split}
    F(\Jinv_1(x); x) ={}& \tilde{F}(\Jinv_1(x)) \\
     \colonequals {}& \frac{1}{2} \left( \log \frac{\Jinv_1(x) + \frac{1}{2}}{\Jinv_1(x) - \frac{1}{2}} \right)^2 + (\Jinv_1(x) + \frac{1}{2}) \log(\Jinv_1(x) + \frac{1}{2}) - (\Jinv_1(x) - \frac{1}{2}) \log(\Jinv_1(x) - \frac{1}{2}) - 1,
  \end{split}
\end{equation}
where we consider $\tilde F$ as a function of $u=\Jinv_1(x)$, and
\begin{equation}
  \frac{d}{dx} F(\Jinv_1(x); x) = \frac{d}{d u} \tilde{F}(u) \left( \frac{d \Jlike(u)}{d u} \right)^{-1} = \log \frac{\Jinv_1(x) + \frac{1}{2}}{\Jinv_1(x) - \frac{1}{2}}.
\end{equation}
On the other hand, by the identities \eqref{def M} and \eqref{formula M}, we have
\begin{equation} \label{eq:expression_of_G(x)_quadratic_V}
  \begin{split}
    G(x) ={}& -\frac{1}{2 \pi i} \oint_{\gamma} \frac{\Jlike(\xi)}{\xi - \Jinv_1(x)} d \xi \\
    ={}& -\frac{1}{2 \pi i} \oint_{\gamma} \frac{\xi + \frac{1}{2}}{\xi - \Jinv_1(x)} d \xi - \frac{1}{2 \pi i} \oint_{\gamma} \frac{\log(\frac{\xi + \frac{1}{2}}{\xi - \frac{1}{2}})}{\xi - \Jinv_1(x)} d \xi.
  \end{split}
\end{equation}
By the calculation of residue, it is obvious that the first contour integral in the second line of \eqref{eq:expression_of_G(x)_quadratic_V} vanishes, and after some effort, we find the second contour integral has value $-2 \pi i \log \frac{\Jinv_1(x) + \frac{1}{2}}{\Jinv_1(x) - \frac{1}{2}}$. Thus \eqref{eq:essential_relation_quadratic_V} is proved, and together with \eqref{eq:easy_relation_quadratic_V} the equivalence between \eqref{eq:saddle_point_outside} and \eqref{eq:asy_of_p_outside} is obtained.

\section*{Acknowledgements}
The authors acknowledge support from the Belgian Interuniversity Attraction Pole P06/02, P07/18. TC was also supported by the European Research Council under the European Union's Seventh Framework Programme (FP/2007/2013)/ ERC Grant Agreement n. 307074, by FNRS, and by the ERC project FroMPDE. DW thanks Professor Pierre van Moerbeke for the fellowship in Universit\'{e} catholique de Louvain where the project initiated, and is grateful to the Universit\'{e} catholique de Louvain for hospitality. Part of DW's work was done in the department of mathematics, University of Michigan. Support of an NUS start-up grant \# R-146-000-164-133 is gratefully acknowledged.

\bibliographystyle{abbrv}
\bibliography{bibliography}

\end{document}